\pdfoutput=1
\documentclass[11pt,usenames,dvipsnames]{article} 
\usepackage{amsmath,amssymb}
\usepackage{mathtools}
\usepackage{natbib}
\usepackage[font=small]{caption}
\usepackage{enumitem}
\usepackage[usenames]{color}
\usepackage{bm}
\usepackage{ying}
\usepackage{comment}
\usepackage{multirow}
\usepackage{rotating}
\usepackage{fullpage}
\usepackage{authblk}
\usepackage{geometry}
\usepackage{parskip}
\allowdisplaybreaks
\usepackage{circuitikz}
\usepackage{booktabs}
\usepackage{float}
\usepackage{placeins}
\usepackage{subcaption}
\usepackage{tikz,tabularx}
\usetikzlibrary{patterns}
\usetikzlibrary{arrows}
\usetikzlibrary{tikzmark, positioning, fit, shapes.misc}
\usetikzlibrary{decorations.pathreplacing, calc}
\tikzset{brace/.style={decorate, decoration={brace}},
  brace mirrored/.style={decorate, decoration={brace,mirror}},
}

\newcolumntype{g}{>{\columncolor{red}}c}

\PassOptionsToPackage{dvipsnames}{xcolor}
\usepackage{graphicx,amssymb} 
\definecolor{myblue}{RGB}{24, 13, 164}
\usepackage[colorlinks,
            linkcolor=red,
            anchorcolor=blue,
            citecolor=blue
            ]{hyperref} 
\usepackage{algorithm}
\usepackage{algorithmic}
\usepackage{soul}

\newcommand{\test}{\textnormal{test}}

\let\hat\widehat
\let\tilde\widetilde

\def\asto{{\stackrel{\textrm{a.s.}}{\to}}}

\def \iid {\stackrel{\text{i.i.d.}}{\sim}}

\def \calib {\textrm{calib}}
\def \train {\textrm{train}}
\def \sel {\textrm{sel}}
\def \labelled {\textrm{label}}

\def \test {\textrm{test}}

\def \fdr {\textnormal{FDR}}

\def \bh {\textnormal{BH}}
\def \homo {\textnormal{homo}}
\def \hete {\textnormal{hete}}
\def \dtm {\textnormal{dtm}}

\newcommand{\basename}{\textnormal{SCS} }

\theoremstyle{plain}

\providecommand{\keywords}[1]
{
  \fontsize{9}{12}\selectfont
  \textbf{\textit{Keywords:}} #1
}

\def\@#1\@{\begin{align}#1\end{align}}
\def\$#1\${\begin{align*}#1\end{align*}}

\usepackage{color-edits}
\addauthor[Ying]{ying}{magenta}
\addauthor[Tian]{tian}{blue}

\usepackage{xcolor}

\newcommand{\revise}[1]{#1}
\newcommand{\rrevise}[1]{#1}
\newcommand{\rrrevise}[1]{#1}
\newcommand{\frevise}[1]{#1}

\title{Optimized Conformal Selection: Powerful Selective Inference \\After Data-Dependent Score Construction}
\author[1]{Tian Bai}
\affil[1]{Department of Statistics, Stanford University}
\author[2]{Ying Jin}
\affil[2]{Department of Statistics and Data Science, University of Pennsylvania}
\date{}

\begin{document}

\maketitle

\begin{abstract} 
\fontsize{9}{11}\selectfont
We study data-driven score construction in conformal selection, which uses conformal p-values and multiple testing to select ``interesting'' instances with large unobserved labels from  a pool of unlabeled data, while controlling the FDR in finite sample. While existing solutions require the score function to be constructed independent of the calibration and test data, it is desirable to select or construct powerful conformity score in a data-driven manner  and to leverage all labeled data for score training when presented with many score choices and limited labeled data. We propose Optimized Conformal Selection (OptCS), a general framework that allows valid multiple testing after flexible data-driven score construction. We introduce general conditions under which OptCS constructs valid conformal p-values despite substantial data reuse  and handles complex p-value dependencies to maintain  finite-sample FDR control. We instantiate this general recipe to propose three FDR-controlling procedures, each optimizing the models differently: (i) selecting a powerful score function among multiple candidates, (ii) training the score with all labeled data without sample splitting, and (iii) combining full-sample training and model selection. We demonstrate the efficacy of our methods via simulation studies and real applications in drug discovery and large language models.
\end{abstract} 

\keywords{Conformal inference,  conformal p-value, model selection, FDR control, sample splitting.}
\fontsize{10}{12}\selectfont

% !TEX root = main.tex
\section{Introduction}

Conformal inference is a versatile, distribution-free framework for predictive inference~\citep{vovk2005algorithmic}. 
Given any prediction model, it employs any user-specified conformity score to measure how well the outcomes ``conform'' to the predictions, and then leverages the exchangeability among scores to deliver valid statistical evidence (p-value). Its strength lies in its flexibility as a ``wrapper'' for valid inference, regardless of the score function used. This principle underlies its use in various tasks, including prediction set construction~\citep{lei2018distribution} and recent extensions in two-sample testing~\citep{hu2024two} and multiple testing~\citep{bates2023testing,jin2023selection}. 
However, despite this flexibility, a common challenge is score construction~\citep{liang2024conformal}: determining the appropriate conformity score and/or prediction model to improve task-specific performance.

We study data-driven score construction in the context of model-free selective inference~\citep{jin2023model}, which aims to identify samples with large unobserved \rrrevise{response variables} with  error control. 
Formally, given labeled data $\{(X_i,Y_i)\}_{i=1}^n$ and unlabeled samples $\{X_{n+j}\}_{j=1}^m$, the goal is to select a subset $\mathcal{S}\subseteq \{1,\dots,m\}$ such that most of them obey $Y_{n+j}>c_{n+j}$, where $c_{n+j}$ are user-specified thresholds above which the outcomes are considered interesting. This is formalized by false discovery rate (FDR) control:
\@ \label{eq:fdr}
\fdr := \EE\Bigg[  \frac{\sum_{j=1}^m \ind\{j\in \cS, Y_{n+j}\leq c_{n+j}\}}{1\vee|\cS|} \Bigg] \leq q, 
\@ 
where $q\in (0,1)$ is a user-specified nominal level.   
Since the FDR measures the expected proportion of incorrectly selected instances, rigorous  FDR control is crucial for the trustworthiness of screening procedures and the efficiency of downstream investigations on the selected units (see representative applications below). 

Conformal selection~\citep{jin2023selection} addresses this problem by extending  conformal inference ideas to multiple test samples (see also extensions to  covariate shift~\citep{jin2023model} and online settings~\citep{xu2024online}, and with constraints and objectives \rrrevise{such as diversity}~\citep{wu2023optimal,huoreal,wangconformalized,nair2025diversifying}). 
Given any pre-specified conformity score function (wrapped around any prediction model), \rrrevise{conformal selection is an automatic pipeline: it compares each test score to the scores on the labeled (calibration) data, constructs p-values for the test units, and determines the selection set among the unlabeled (test) data. Its performance thus relies on the choice of the conformity score function (which implicitly depends on the prediction model).}

To maintain FDR control,  the \rrrevise{score function construction} is required to be \rrrevise{independent of the calibration and test data it applies to}. Akin to the celebrated split conformal prediction idea~\citep{lei2018distribution}, \rrrevise{if the score needs to be trained, tuned, or selected, this is usually done on a separate training fold (often a subset of labeled data).} This prevents practitioners from (i) data-driven model selection, where one selects a score based on \rrrevise{its downstream selection performance  such as the resulting selection-set size}, and (ii) \rrrevise{full-data training, where all  labeled data are used for model training without sample splitting}. We illustrate the need for both forms of data-dependent score construction through two typical applications.  
\begin{itemize}
    \item \textbf{Drug discovery}. Early stages of drug discovery aim to find highly viable drug candidates---molecules, antibiotics, or proteins---with desirable biological properties~\citep{szymanski2011adaptation,macarron2011impact}. Such tasks can be characterized by finding drug candidates $j\in\{1,\dots,m\}$ with $Y_{n+j}>c_{n+j}$ for an unknown property $Y_{n+j}$ and thresholds $c_{n+j}$. Given any model that predicts $Y_{n+j}$ based on physical/chemical features $X_{n+j}$, conformal selection can be used to shortlist drug candidates with FDR control, enabling efficient pre-evaluation before costly validations~\citep{jin2023selection}. 
    However, with numerous pretrained models by recent advances in machine learning~\citep{xiong2019pushing,dgllife,huang2020deeppurpose,carracedo2021review}, \rrrevise{it is often unclear a priori which conformity score (and the pretrained model it relies on) will be most effective for the downstream FDR-controlled selection. Rather than fixing this choice in advance, a scientist may desire to use the available data to select a score with strong
downstream selection performance.}
    
    \item \textbf{LLM abstention}. Large language models (LLMs) show tremendous potential in high-stakes tasks  such as generating radiology reports, but issues like hallucination pose serious risks~\citep{huang2023survey,weidinger2021ethical}. In this context, conformal selection can be used to select radiology images $j\in \{1,\dots,m\}$ for which the generated reports align with experts, represented by an unknown ``alignment indicator'' $Y_{n+j}=1$. \cite{gui2024conformal} uses a set of labeled data (images and human reports) to train an ``alignment predictor'' that predicts $Y_{n+j}$ based on features of LLM outputs $X_{n+j}$, and applies conformal selection---needing a separate labeled calibration dataset---to select new images with ``aligned'' outputs. \rrrevise{Since the expert-labeled data are often scarce, it is naturally desirable to train the predictor on as many labeled data as possible to improve the selection power.}
\end{itemize} 

The challenge  in performing such data-driven \rrrevise{score construction} is that it is often incompatible with  the exchangeability requirement. \rrrevise{Once the calibration and test data are involved in selecting and training the score function, the resulting scores may no longer be exchangeable. In the context of conformal selection, it remains unclear how the score can be constructed using the data it applies to while maintaining valid FDR control.}

\begin{figure}
    \centering
    \captionsetup{font=small}
    \includegraphics[width=0.95\linewidth]{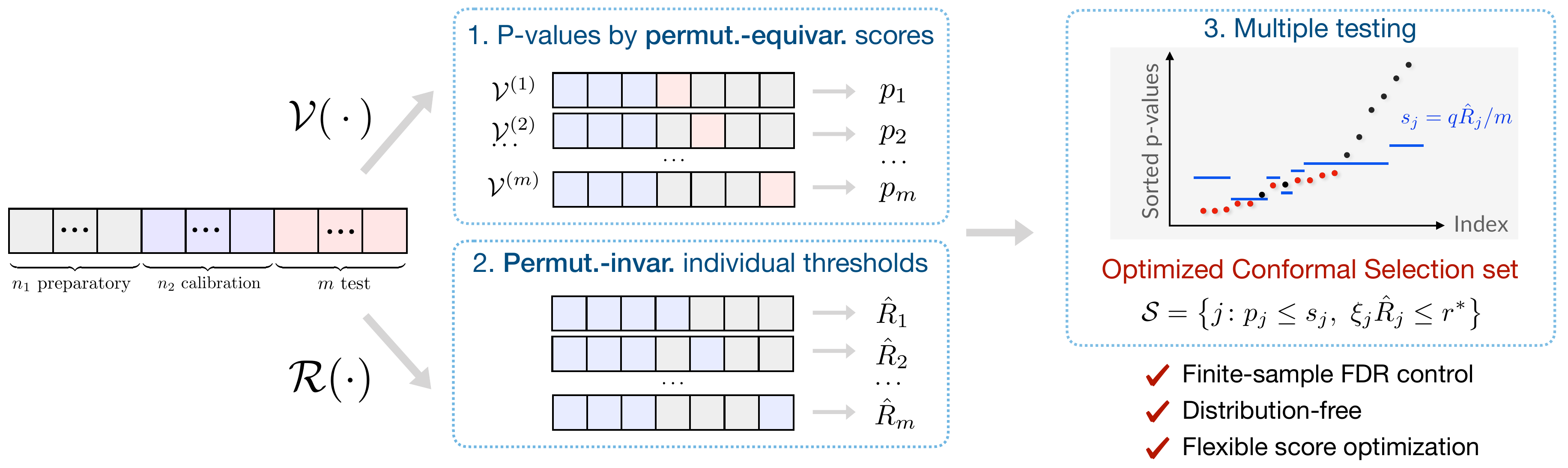}
    \caption{{\small Overview of OptCS. (Optional) Split labeled data into preparatory and calibration sets. \underline{1.\,P-value}: construct p-values $\{p_j\}$ based on conformity scores from any data-dependent score construction process $\mathcal{V}(\cdot)$ obeying a mild permutation-equivariance condition. %
    \underline{2.\,Threshold}: calibrate individual thresholds $\hat{R}_j$ obeying a mild permutation-invariance condition. \underline{3.\,Multiple testing}: combine $\{p_j\}$ and $\{\hat{R}_j\}$ to produce a selection set $\mathcal{S}$ with finite-sample, distribution-free FDR control.}}
    \label{fig:overview}
\end{figure}

\noindent\textbf{Our contributions.}
We propose Optimized Conformal Selection (OptCS), a general framework for FDR-controlled conformal selection after data-dependent score construction. \rrrevise{Here, score
construction broadly includes selecting among candidate scores, training prediction models, and combining these operations to improve downstream
selection power.} We summarize our contributions below.
 
\begin{itemize}
    \item \textbf{A general framework.} Section~\ref{sec:method} presents the general OptCS procedure (Figure~\ref{fig:overview}).  We identify two key objects: a
    score-generating map \(\mathcal V\), which constructs conformal p-values when scores are constructed with data, and an auxiliary map \(\mathcal R\) 
    which calibrates the dependent p-values for FDR control. We provide general permutation invariance conditions under which the selection set controls FDR in finite samples.
    
    \item \textbf{A family of practical procedures.}
    Section~\ref{sec:procedures} instantiates OptCS into three procedures for different practical regimes:
    OptCS-MSel for selecting among pre-trained candidate scores, OptCS-Full for full-data
    score training with a fixed model class, and OptCS-Full-MSel for combining full-data
    training and score selection. These procedures provide concrete, diverse ways to reuse data for
    improving selection power while preserving validity. See Table~\ref{tab:optcs-features} for a comparison. 
    \item \textbf{Empirical and practical guidance.}
    Sections~\ref{sec:simu} and~\ref{sec:real} evaluate these procedures through simulations and real-data applications
    in drug discovery and LLM alignment. The experiments illustrate the power gains from
    data-dependent score construction and clarify the computation--power
    tradeoffs of the procedures.
\end{itemize} 

\vspace{1em}

Figure~\ref{fig:preview} previews the performance of the three procedures with distinct model optimization strategies in our numerical experiments. Each strategy, applicable in different scenarios, significantly improves power upon baseline methods. 
Finally, we demonstrate the efficacy of the methods through numerical simulations in Section~\ref{sec:simu} and real applications in drug discovery and large language model abstention in Section~\ref{sec:real}. 

\begin{figure}[htbp]
	\centering
	\includegraphics[width=0.78\linewidth]{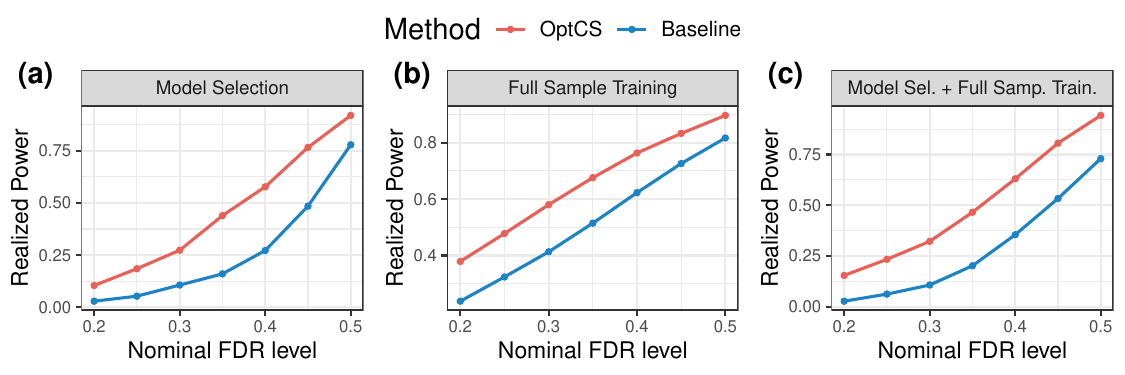}
	\caption{Preview of numerical results. Left: Performance of OptCS-MSel which selects from pre-trained models. Middle: Performance of OptCS-Full which leverages full data for model fitting with a given model class. Right: Performance of OptCS-Full-MSel which combines model selection and full-sample training.}
	\label{fig:preview}
\end{figure}

\rrevise{\noindent\textbf{When to use which?}   The three instantiations of OptCS  employ  distinct score construction strategies suitable in different scenarios, trading off statistical power,
data reuse, and computational cost. 
{OptCS-MSel} is appropriate when one has access to a fixed collection of pre-trained
scores and wishes to select among them in a data-driven way
\emph{without} retraining; this happens, e.g., when training is expensive but multiple
off-the-shelf models are available.
{OptCS-Full} is suitable when a model class (e.g., a training algorithm) is specified, and retraining is feasible.
Finally, {OptCS-Full-MSel}  is useful
when one expects meaningful power gain from both avoiding sample split and model selection and retraining multiple scores is affordable.} \frevise{When retraining is expensive, one may use Non-LOO variants of OptCS-Full and OptCS-Full-MSel with improved computational efficiency; see Table~\ref{tab:optcs-features} and Section~\ref{sec:procedures}.}
 
\begin{table}[htbp]
	\centering
	\footnotesize
	\setlength{\tabcolsep}{2pt}
	\renewcommand{\arraystretch}{1}
	\begin{tabular}{lcccc}
		\toprule
		Method & Model selection & Full-data training & Training cost  \\
		\midrule
		SCS & no & no & 1 fit   \\
		OptCS-MSel & yes & no & $K$ fits (if not pre-trained)  \\
		OptCS-Full & no & yes & $(n_2+m)$  fits \\
		\frevise{OptCS-Full-Non-LOO} & no & yes & $1$ fit \\
		OptCS-Full-MSel & yes & yes & $K(n_2+m)$ fits \\
		\frevise{OptCS-Full-MSel-Non-LOO} & yes & yes & $K$ fits \\
		\bottomrule
	\end{tabular}
	\caption{Comparison of OptCS procedures and baseline, where $K$ is the number of model classes, $n_2$ is the number of calibration data, and $m$ is the number of test data.}
	\label{tab:optcs-features}
\end{table} 
  
\noindent\textbf{Data and Code.} Reproducibility  code for both our simulation and real data experiment can be found at \url{https://github.com/Tian-Bai/OptCS}.

% !TEX root = main.tex
\section{Preliminaries} \label{sec:prelim}

This section  reviews the key ideas of conformal selection, and  discusses the challenges in data-driven model optimization, followed by   
a brief overview of related works.

\subsection{Recap: Conformal selection}
\label{subsec:recap_cs}

Assume access to  labeled data $\cD_\labelled = \{(X_i, Y_i)\}_{i=1}^n$ and  unlabeled test data $\cD_\test = \{ X_{n+j} \}_{j=1}^m$ whose responses $\{Y_{n+j}\}_{j=1}^m$ are unobserved. 
Conformal selection begins by randomly splitting the labeled data $\cD_\labelled$ into two disjoint subsets, the training set $\cD_\train := \{(X_i,Y_i)\}_{i=1}^{n_1}$ and the calibration set $\cD_\calib :=  \{(X_i,Y_i)\}_{i=n_1+1}^{n}$. First, $\cD_\train$ is used to train a model and build a \emph{monotone} conformity score function $v: \cX \times \cY \rightarrow \mathbb{R}$ where $\cX$ and $\cY\subseteq \RR$ represent the sample space of $\{X_i\}$ and $\{Y_i\}$.
\begin{definition} \label{def:monotone_vanilla}
A score function $v\colon \cX\times\cY\to \RR$ is monotone if $v(x,y)\leq v(x,y')$ whenever $y\leq y'$ for any $x\in \cX$ and $y,y'\in \cY$. 
\end{definition}
\rrevise{A monotone} score function $v$ can wrap around any prediction model obtained in the training step, e.g., $v(x, y) = y - \hat\mu(x)$ where $\hat\mu(x)$ is a point prediction for $y$ given $x$. Other popular choices include those based on quantile regression~\citep{romano2019conformalized} and conditional cumulative distribution functions~\citep{chernozhukov2021distributional}.  
We then compute the calibration scores \frevise{$V_i = v(X_{i}, Y_{i})$, $i=n_1+1,\dots,n$}, and test scores $\hat{V}_{n+j} = v(X_{n+j}, c_{n+j})$, $j=1,\dots,m$. These scores can be understood as ``test statistics'' for large outcomes. 
Conformal selection then builds a conformal p-value for each test point: 
{\setlength{\belowdisplayskip}{6pt}
\setlength{\abovedisplayskip}{6pt} 
\@\label{eq:cf_pval}
    p_j = \frac{  1 + \sum_{i=n_1+1}^n \ind \{V_i \leq \hat{V}_{n+j}\}  }{n_2+1}  ,\quad j=1,\dots,m,
\@}
where %
$n_2=n-n_1$.   
Intuitively, a small p-value suggests  $\hat{V}_{n+j}=v(X_{n+j},c_{n+j})$ is unusually \revise{small} compared to $\{V(X_i,Y_i)\}$, indicating that  $Y_{n+j}$ is likely to exceed $c_{n+j}$, i.e., is ``of interest''. 
The selection set $\cS$ is then obtained by applying the Benjamini-Hochberg (BH) procedure~\citep{benjamini1995controlling} to $\{p_j\}_{j=1}^m$ at nominal level $q$. 
As long as  $(X_{n+j},Y_{n+j})$ is exchangeable with  $\{(X_i,Y_i)\}_{i=n_1+1}^n$, it controls FDR in finite sample~\citep{jin2023selection}.
From now on, we refer to this procedure as Split Conformal Selection (SCS) for convenience. 
Following \cite{jin2023selection}, we exclusively focus on the powerful ``clipped''-type scores.
\begin{assumption} \label{assump:1}
    The thresholds $\{c_i\}_{i=n_1+1}^n$ are observed for the labeled data, such that $\{(X_i,c_i,Y_i)\}_{i=n_1+1}^{n+m}$ are exchangeable. The score functions considered are clipped, i.e., they are \revise{monotone and obey} $v(X_i,Y_i) = v(X_i,c_i)$ when $Y_i\leq c_i$. 
\end{assumption}
\rrrevise{The condition \(v(X_i,Y_i)=v(X_i,c_i)\) for \(Y_i\le c_i\) means that the score is clipped
at the threshold: any response  below \(c_i\) is treated as $c_i$. This pushes the scores $V_i$ to the largest values possible (subject to monotonicity) to yield small p-values~\eqref{eq:cf_pval} and powerful selection}. 
As pointed out in~\cite{jin2023selection}, any monotone score  can be transformed to a clipped score with  strictly higher power once $\{c_i\}_{i=1}^n$ are observed. As we shall see later, the fact that $v(X_i, Y_i) = v(X_i, c_i)$ when $Y_i\leq c_i$ is crucial for flexible score construction without knowing the test labels.   Although only Assumption~\ref{assump:1} is needed for the validity of OptCS, following~\cite{jin2023selection} we typically use  scores of the form $v(X_i, Y_i) = M \cdot \ind \{Y_i > c_i\} - f(X_i)$, where $M$ is a large constant and $f$ is a point predictor.  
The score implicitly depends on the thresholds, i.e., $v(x,y,c)$, yet we suppress the dependence on $c$ for convenience.

\subsection{Challenges in optimizing conformity score with data reuse} 
\label{subsec:challenges}

We take a moment  to review the theoretical underpinning of \basename and discuss the challenges in preserving FDR control while permitting data-driven \rrrevise{score construction}. %
At a high level, the FDR control of \basename relies on two facts: 
\begin{enumerate}[label=(\roman*)]
    \item \emph{Valid p-values}. Should $\hat{V}_{n+j}$ be replaced by $v(X_{n+j},Y_{n+j})$, the p-value would follow \frevise{a super-uniform} distribution due to exchangeability. Thus, a small value of $p_j$ indicates evidence for $\hat{V}_{n+j} < v(X_{n+j},Y_{n+j})$ and hence $Y_{n+j}>c_{n+j}$ due to monotonicity. 
    \item \emph{Favorable p-value dependence}. While the conformal p-values are dependent since they use the same set of calibration data,  their positive dependence~\citep{bates2023testing,jin2023selection} ensures FDR control with the BH procedure.
\end{enumerate} 
Ensuring the independence between the score function $v$ and the data $\cD_\calib\cup \cD_\test$ it applies to is essential for both facts above. 
\rrrevise{However, this requirement prevents two natural forms of data reuse in score construction: selecting a score with strong downstream selection performance from a pre-specified finite candidate class, and training a score using calibration data involved in p-value construction.} 

We illustrate the challenge via a naive attempt at the first form, model selection.
Consider  $\{v{(\cdot,\cdot; k)} \}_{k=1}^K$ where each $v(\cdot,\cdot;k)\colon \cX\times\cY\to \RR$ is a pre-specified conformity score function. A natural idea is to choose the score that leads to the largest \basename set, i.e., 
$
    \hat{k} = \argmax_{k\in[K]} |\cS{(k)}|,
$
where $\cS{(k)}$ is the \basename output with the score $v{(\cdot \,, \cdot \,; k)}$. 
The final selection set is 
$\cS = \cS{(\hat{k})}$, which is equivalent to applying BH to p-values 
\[
p_j = \frac{\sum_{i=n_1+1}^n \ind \{v(X_i,Y_i;\hat{k}) \leq v(X_{n+j},c_{n+j};\hat{k}) \} + 1}{n_2+1}.
\]
Here, $\hat{k}$ depends on $\cD_\calib$ and $(X_{n+j},c_{n+j})$ non-symmetrically. Thus,  $v(X_i,Y_i; \hat{k})$ and $v(X_{n+j},Y_{n+j}; \hat{k})$ are no longer exchangeable, yielding invalid p-values. It is also hard to justify the positive dependence, since $\hat{k}$ is a complex function of the calibration and test data. Indeed, we observe drastic violation of FDR using such a ``greedy'' approach in experiments. 

Analogously, for full-data training, involving a calibration sample in training the function $v$ can easily break the exchangeability between $v(X_i,Y_i)$ and $v(X_{n+j},Y_{n+j})$.
In general, involving the calibration and test data in score construction introduces substantial data reuse which  invalidates the p-values and complicates FDR control. 
To address these challenges, we (i) construct valid p-values under data reuse, and (ii) develop multiple testing with FDR control despite intricate p-value dependence. 
 
\noindent\textbf{Notations.}
For any set $A$, we write $[A]$ as its unordered version, which provides the same information as the order statistics of values in $A$. For any integer $m \in \mathbb{N}$, we denote $[m] := \{1, 2, \dots, m\}$. For a vector $\mathbf{z}\in \RR^d$, we denote $\mathbf{z}_i$ as the $i$-th element. For a matrix $M\in \RR^{d_1\times d_2}$, we denote $M_{ij}$ as the element on the $i$-th row and $j$-th column, $M_{i,:}$ as the $i$-th row vector, and $M_{:,j}$ the $j$-th column vector. For values $a_1,\dots,a_n$, we sometimes write $a_{1:n}=(a_1,\dots,a_n)$. We denote the full observations  $Z_i = (X_i, Y_i)$, $i \in[n+m]$ and the partial observations  $\hat{Z}_{n+j} = (X_{n+j}, c_{n+j})$, $j \in[m]$.

\subsection{Related work}

\noindent\textbf{Conformal selection and model optimization.} 
This work builds upon the conformal selection framework~\citep{jin2023selection}, motivated by the need for model optimization in applications such as drug discovery~\citep{bai2024conformal} and  language models~\citep{gui2024conformal}. 
More broadly, it adds to the literature on model selection and optimization in conformal inference, which covers (i) prediction set construction for a single test point and (ii) selecting multiple test points. In (i), recent work primarily performs data-driven score optimization in a separate fold \textit{with} data splitting~\citep{stutz2021learning,yang2024selection,huang2024uncertainty,xie2024boosted}. 
Strategies in conformal prediction to avoid sample splitting, such as full conformal prediction and~\cite{liang2024conformal}, often rely on maintaining certain permutation invariance, which is connected to the key ideas of OptCS. 
Works for (ii) are limited to outlier and novelty detection~\citep{zhang2022automs,marandon2024adaptive,liang2024integrative,lee2025full}, which we discuss below. 

\noindent\textbf{Multiple testing with dependent p-values.} 
Our method calibrates individual selection thresholds to control the FDR under complicated dependence. This strategy builds upon the conditional calibration method of~\citet{fithian2022conditional}, which tests deterministic hypotheses with dependent p-values by finding an individual threshold for each p-value, typically via Monte Carlo, conditional on a sufficient statistic whose conditional distribution is known. 
In OptCS, the unordered set of the calibration data and the $j$-th test point serves as the conditioning statistic; since our p-values are valid conditional on it, the individual thresholds (Section~\ref{subsec:general_method}) apply exactly without Monte Carlo approximation. Our setting also differs in that the hypotheses are random, and we additionally introduce homogeneous and deterministic pruning; see Remark~\ref{rmk:cc}. 
A similar strategy is used in~\citet{jin2023model} to address the dependence among weighted conformal p-values under covariate shift, which relies on a conditional independence and permutation invariance structure distinct from ours; here, the thresholds and p-values must be carefully designed to accommodate the dependence induced by data reuse and model selection.

\noindent\textbf{Data reuse in outlier and novelty detection.} 
The idea of reusing data for both model fitting and p-value construction while maintaining validity connects to several recent works on outlier detection. 
AutoMS~\citep{zhang2022automs} selects among candidate novelty detectors the one yielding the most discoveries, using leave-one-out Jackknife scores, which is conceptually related to OptCS-MSel and OptCS-Full-MSel; its validity relies on learning stability and yields a high-probability bound on the FDP and asymptotic FDR control, whereas OptCS controls the FDR in finite samples. 
AdaDetect~\citep{marandon2024adaptive} sets aside a fold of the (inlier) calibration data as a ``null'' class and trains outlier detectors by contrasting it with the remaining calibration and test data, ensuring validity through permutation invariance across the latter. Its training step is close to OptCS-Full in that both involve test samples in training, but OptCS-Full uses all calibration data for training and instead imputes the unknown test labels as nulls without compromising validity. Moreover, evaluating the performance of each candidate model makes the roles of calibration and test data intrinsically asymmetric, so OptCS only requires permutation invariance among the calibration data and the $j$-th test point. That said, the PRDS analysis of AdaDetect can be adapted to an alternative proof of FDR control for OptCS-Full (Appendix~\ref{app:subsec_alt_full_proof}). 
\citet{liang2024integrative} use labeled outliers to train and select one-class classifiers in a way that is permutation invariant to the calibration data and each test point, and they also use conditional calibration for FDR control; their training strategies and integrated p-values are specialized to outlier detection, whereas OptCS accommodates general, possibly heterogeneous thresholds $c_{n+j}$ and general score optimization, including full-sample training. 
Concurrently, \citet{lee2025full} construct full-conformal e-values relative to a null reference sample for novelty detection and apply the e-BH procedure. In contrast, OptCS-Full uses both labeled positives and negatives to train a supervised score, and follows a p-value and conditional calibration route (see Appendix~\ref{app:fcnd_connection}).

\noindent\textbf{Other related methods.} 
\citet{gazin2025selecting} study the selection of informative conformal prediction sets, which involves selection followed by prediction set construction; their selective-conditional method constructs p-values with only the preselected calibration subset, whereas OptCS uses the entire calibration set in both training and calibration. 
The recent ACS method~\citep{gui2025acs} improves power by an interactive testing process that allows a subset of calibration data to enter training. 
Numerical comparisons in Appendix~\ref{app:subsec_exp_od} show that ACS and the outlier detection methods above can be less powerful when adapted to our problems. 
A more detailed discussion of these connections is in Appendix~\ref{app:subsec_related_work}.

% !TEX root = main.tex

\section{Optimized Conformal Selection} \label{sec:method}

\subsection{General procedure}
\label{subsec:general_method}
 We begin by introducing a general procedure that encompasses all  use cases of OptCS to be introduced later, including model selection and/or full-sample model training.  
 
\noindent\underline{\textbf{Step 1: P-values.}}
The first innovation of OptCS is a general construction of valid conformal p-values with data-driven scores. By ``data-driven'', we mean the score choice depends on the calibration and test data. We introduce the notion of score-generating functional (Figure~\ref{fig:score_functional_visual}), with individual score functions for each test point.

\begin{definition}%
\label{def:score_gen}
Let $\cL=[n_2]\cup\{n+1,\dots,n+m\}$. 
    A score-generating functional is a function $\mathcal V:(\mathcal X\times\mathcal Y)^{n+m}\to\mathbb R^{[m]\times\mathcal L}$ with $n_1, n_2, n, m \in \NN$, $n=n_1+n_2$. For notational simplicity, we denote $\cV^{(j)}: (\cX \times \cY)^{n+m} \rightarrow \mathbb{R}^{\cL}$ such that $\cV^{(j)}(\mathbf{z})=\cV(\mathbf{z})_{j,:}$ for any $\mathbf{z}\in (\cX\times \cY)^{n+m}$. 
\end{definition}  

\begin{figure}[!ht]
\centering
\captionsetup{font=small}
\resizebox{0.6\textwidth}{!}{%
\begin{circuitikz} 
\node [font=\LARGE] at (4,12.5) {$\mathcal{V}\Big($};
\foreach \i in {0,1,2} {
    \fill[black!10] (4.75 + \i*1,13) rectangle (5.75 + \i*1,12);   
    \draw [ line width=1pt ] (4.75 + \i*1,13) rectangle (5.75 + \i*1,12);
}
\foreach \i in {3,4,5} {
    \fill[blue!10] (4.75 + \i*1,13) rectangle (5.75 + \i*1,12);   
    \draw [ line width=1pt ] (4.75 + \i*1,13) rectangle (5.75 + \i*1,12);
}
\foreach \i in {6,7,8} {
    \fill[red!10] (4.75 + \i*1,13) rectangle (5.75 + \i*1,12);  
    \draw [ line width=1pt ] (4.75 + \i*1,13) rectangle (5.75 + \i*1,12);
}
\draw [ line width=1pt ] (4.75,13) rectangle (5.75,12);
\node [font=\LARGE] at (14.65,12.5) {$\Big)=$};
\node [font=\LARGE] at (6.3,12.5) {$\cdots$};
\node [font=\LARGE] at (9.3,12.5) {$\cdots$};
\node [font=\LARGE] at (12.3,12.5) {$\cdots$}; 
\foreach \y in {0, 1, 2, 3, 4, 5} {
    \fill[white] (15.5 + \y*0.8, 13.6) rectangle (16.3 + \y*0.8, 12.8);
    \draw [line width=1pt] (15.5 + \y*0.8, 13.6) rectangle (16.3 + \y*0.8, 12.8);
}
\node [font=\large] at (17.9, 12.6) {$\vdots$};
\foreach \y in {0, 1, 2, 3, 4, 5} {
    \fill[white] (15.5 + \y*0.8, 12.2) rectangle (16.3 + \y*0.8, 11.4);
    \draw [line width=1pt] (15.5 + \y*0.8, 12.2) rectangle (16.3 + \y*0.8, 11.4);
}
\draw [line width=1pt,decorate,decoration={brace,amplitude=10pt,mirror}] (4.75,11.8) -- (7.75,11.8);
\node  at (6.25,11.2) {$n_1$ preparatory};
\draw [line width=1pt, decorate, decoration={brace,amplitude=10pt,mirror}] (7.75,11.8) -- (10.75,11.8);
\node  at (9.25,11.2) {$n_2$ calibration};
\draw [line width=1pt, decorate, decoration={brace, amplitude=8pt, mirror}] (15.5, 11.3) -- (17.9, 11.3);
\node [font=\small] at (16.7, 10.7) {$n_2$ calib.~scores};
\draw [line width=1pt, decorate, decoration={brace, amplitude=8pt, mirror}] (17.9, 11.3) -- (20.3, 11.3);
\node [font=\small] at (19.2, 10.7) {$m$ test scores};
\draw [line width=1pt, decorate, decoration={brace,amplitude=10pt,mirror}] (10.75,11.8) -- (13.75,11.8);
\node [font=\large] at (21,13.25) {$\mathcal{V}^{(1)}$};
\node [font=\large] at (21,12.6) {$\vdots$};
\node [font=\large] at (21,11.85) {$\mathcal{V}^{(m)}$};
\node  at (12.25,11.2) {$m$ test};
\node [font=\LARGE] at (4,12.5) {$\mathcal{V}\Big($};
\end{circuitikz}
}%
\caption{\small Visualization of a score-generating functional.}
\label{fig:score_functional_visual}
\end{figure}

We refer to the first $n_1$ labeled samples as preparatory data and the next $n_2$ labeled samples as calibration data. 
The last $m$ inputs represent test data $\{\hat{Z}_{n+j}\}_{j=1}^m$ with unknown labels.
\rrevise{This separation keeps the framework general and simplifies the descriptions of steps like model pre-training. The calibration data will be assigned scores to construct p-values, and the exchangeability condition is over the calibration and test data only. One may simply take $n_1=0$ to involve all the labeled data in both score generation and calibration.}
The output of $\cV$ is an $m\times(n_2+m)$ matrix, and scores in $\cV^{(j)}$ will be used to construct the $j$-th p-value. %

\begin{remark}
    Definition~\ref{def:score_gen} strictly generalizes the concept of score functions in SCS. As in Section~\ref{sec:prelim}, \basename computes $\{V_i\}_{i=n_1+1}^n$ and  $\{\hat{V}_{n+j}\}_{j=1}^m$ using a pre-determined score function $V$. This corresponds to $\cV^{(j)} (\mathbf{z})  \equiv (V_{n_1+1}, \dots, V_{n}, \hat{V}_{n+1}, \dots, \hat{V}_{n+m})$ for all $j$. 
\end{remark} 

Intuitively, $\cV$ encapsulates 
all the model optimization operations that search for a powerful score. It may include a process that selects from multiple score functions based on the selection performance in SCS \revise{(Section~\ref{subsec:modsel})}, or a training process that involves all the labeled and unlabeled data \revise{(Section~\ref{subsec:loo}), or both (Section~\ref{subsec:loo_sel})}. 

Given data $Z_{1:n_1}$, $Z_{n_1+1:n}$, and $\hat{Z}_{n+1:n+m}$, we denote the generated scores as  
{\setlength{\belowdisplayskip}{4pt}
\setlength{\abovedisplayskip}{6pt}\@\label{eq:scores}
(  V_{n_1+1}^{(j)}, \cdots ,  V_{n}^{(j)}, \hat{V}_{n+1}^{(j)}, \dots, \hat{V}_{n+m}^{(j)} ) :=  \cV^{(j)}(Z_{1:n_1}, \, Z_{n_1+1:n}, \, \hat{Z}_{n+1:n+m}) .
\@}
\revise{Henceforth, we use $V_i^{(j)}$ and $\hat{V}_{n+\ell}^{(j)}$ to exclusively denote  outputs of a score-generating function $\cV$.} We then compute the conformal p-values in the same way as SCS:%
{\setlength{\belowdisplayskip}{4pt}
\setlength{\abovedisplayskip}{6pt}\@
\label{form:conf_p}
p_j =  \frac{1 + \sum_{i =n_1+1}^{n} \ind\{V_i^{(j)} \leq \hat{V}_{n+j}^{(j)}\}}{n_2 + 1}.
\@}

To ensure valid p-values, we only require each set of scores~\eqref{eq:scores} in $\cV^{(j)}$ to be monotone and permutation equivariant to the calibration data and the $j$-th test point, formalized in the two definitions below.\footnote{When $\cV$ is a randomized algorithm, we require Definition~\ref{def:monotone} to hold conditional on the randomness in the algorithm, and Definition~\ref{def:permu_equiv} to hold in a distributional sense, i.e., $=$ replaced by $\stackrel{d}{=}$.} By associating  each p-value with an individual score-generating process, it is much easier to fulfill these requirements by designing each individual process.  
For generic input values, we write the inputs as $(z'_{1:n_1}, z_{1:n_2}, \hat{z}_{n+1:n+m})$, where $z_i=(x_i,y_i)$ and $\hat{z}_{n+j} = (x_{n+j},c_{n+j})\in \cX\times\cY$. 
For a vector $\mathbf{a}$ indexed by $\cL$, we write $a_i$ for its entry indexed by $i\in\mathcal L$.
 
\begin{definition}[Monotonicity for the null]
\label{def:monotone} 
    A score-generating functional  
    is monotone for the $j$-th null, if for any $z_{1:n_1}'\in (\cX\times\cY)^{n_1}$, $z_{1:n_2}\in (\cX\times\cY)^{n_2}$, $\hat{z}_{n+1:n+m}\in (\cX\times\cY)^{m}$, and $z_{n+j}=(x_{n+j},y_{n+j})$, whenever $y_{n+j}\leq c_{n+j}$, it holds that $\hat{\mathbf{V}}^{(j)}_{n+j}\geq {\mathbf{V}}^{(j)}_{n+j}$ and $\hat{\mathbf{V}}^{(j)}_{\ell} = {\mathbf{V}}^{(j)}_{\ell}$ for any $\ell\neq n+j$, $\ell\in \cL$, where
    \$
    \hat{\mathbf{V}}^{(j)}&:= \cV^{(j)}( \textcolor{Gray}{z_{1:n_1}',z_{1:n_2}, \hat{z}_{n+1:n+j-1},} \,\hat{z}_{n+j}, \textcolor{Gray}{\hat{z}_{n+j+1:n+m}})\in \RR^{\cL}, \\
    {\mathbf{V}}^{(j)}&:= \cV^{(j)}( \textcolor{Gray}{z_{1:n_1}',z_{1:n_2}, \hat{z}_{n+1:n+j-1},} \, {z}_{n+j}, \textcolor{Gray}{\hat{z}_{n+j+1:n+m}})\in \RR^{\cL}.
    \$

\end{definition}

\begin{definition}[Permutation equivariance]
\label{def:permu_equiv}
    A score-generating functional $\cV$ is permutation equivariant if for all $j\in [m]$,  any $z_{1:n_1}'$, $z_{1:n_2}$, and $z_{n+1:n+m}$, and  any permutation $\pi$ of $\cL$ that fixes every index outside $\{1,\dots,n_2,n+j\}$, it holds that
    \[
    \cV^{(j)} (z_{1:n_1}', z_{1:n_2}, z_{n +1}, \dots, z_{ n +j }, \dots, z_{n+m})_{\pi(i)}=\cV^{(j)} (z_{1:n_1}', z_{\pi(1:n_2)}, z_{n+1}, \dots, z_{\pi(n+j)}, \dots, z_{n+m})_{i}
    \]
    for all $i\in\mathcal L$. 
\end{definition} 
 
The above definitions  impose two intuitive requirements on the ($j$-th set of) scores which extend SCS. 
Definition~\ref{def:monotone} ensures the ordering of scores informs the ordering of $Y_{n+j}$ and $c_{n+j}$, extending Definition~\ref{def:monotone_vanilla} to data-dependent scores. Definition~\ref{def:permu_equiv} posits that permuting the $(n_2+1)$ calibration and test inputs permutes their outputs correspondingly with values unchanged. Altogether, they ensure the p-value validity (formalized in Section~\ref{subsec:general_theory}) despite  data-driven model optimization. \rrevise{Remark~\ref{rmk:general_V_idea} previews a unified principle in all our procedures to fulfill these conditions.}

\begin{remark}[General design of $\cV^{(j)}$]
\label{rmk:general_V_idea}
    \rrevise{To fulfill Definitions~\ref{def:monotone} and~\ref{def:permu_equiv}, in our specific procedures in Section~\ref{sec:procedures}, 
    the common template  is to set $V_i^{(j)} = v(X_i,Y_i;\hat\mu_i^{(j)})$ in~\eqref{eq:scores}, where $v(x,y;\mu)$ is a fixed mapping  given any  model $\mu$, and $\hat\mu_i^{(j)}$ is obtained in a process that is permutation invariant to $\{Z_\ell\}_{\ell=n_1+1}^n\cup\{\hat{Z}_{n+j}\}$. Notably, the $j$-th test point must enter a (label-dependent) score-generating process with an imputed label that coincides with the truth under the null for validity. 
    This allows flexible design for powerful scores. In Section~\ref{subsec:modsel}, $\hat\mu_i^{(j)}\equiv \hat\mu_j$, a model selected from several candidates; in Section~\ref{subsec:loo}, $\hat\mu_i^{(j)}$ is trained with (nearly) all labeled data; and Section~\ref{subsec:loo_sel} combines the two. Later, we will construct $\cV$  concretely and show how the conditions are met.}
\end{remark}

\noindent\underline{\textbf{Step 2: Calibrate selection criteria.}}  The flexible score-generating process can introduce intricate dependence among these p-values; e.g., they might not be PRDS~\citep{benjamini2001control}. To maintain FDR control, our second innovation is to calibrate individual selection criteria for the p-values by ``auxiliary selection sizes'' %
through any function obeying a permutation invariance condition.

\begin{definition}%
\label{def:R_mon_invar}
    A function $\cR \colon (\cX\times \cY)^{m+n} \to (\RR^{+})^{m}$ is permutation invariant under the $j$-th null, if for any $z_{1:n_1}'$, $z_{1:n_2}$, $\hat{z}_{n+1:n+m}$, and $z_{n+j}=(x_{n+j},y_{n+j})$ with $y_{n+j}\leq c_{n+j}$, and any permutation $\pi$ of $\{1, \dots, n_2, n+j\}$,
    it obeys  
    \$
    &\cR(\textcolor{Gray}{z_{1:n_1}',z_{1:n_2}, \hat{z}_{n+1:n+j-1},} \,\hat{z}_{n+j}, \textcolor{Gray}{\hat{z}_{n+j+1:n+m}})_j = \cR (\textcolor{Gray}{z_{1:n_1}',z_{1:n_2}, \hat{z}_{n+1:n+j-1},} \, {z}_{n+j}, \textcolor{Gray}{\hat{z}_{n+j+1:n+m}})_j, \\
    &\cR (z_{1:n_1}', z_{1:n_2}, z_{(n+1):(n+m)})_j  =\cR (z_{1:n_1}', z_{\pi(1:n_2)}, z_{n+1}, \dots, z_{\pi(n+j)}, \dots, z_{n+m})_j.
    \$
\end{definition}
The permutation invariance condition implies $\cR$ is auxiliary to the ordering of $\{Z_i\}_{i=n_1+1}^n$ and the $j$-th test point. \revise{The role of $\cR$ is to calibrate the selection set to preserve FDR control given the power optimization operations in $\cV$ and $\{p_j\}$}, which we elaborate in Section~\ref{subsec:general_theory}.
We define auxiliary selection sizes as the output of $\cR$: 
{\setlength{\belowdisplayskip}{4pt}
\setlength{\abovedisplayskip}{4pt}
\@\label{eq:aux_sel_construct}
(\hat{R}_1,\dots,\hat{R}_m) := \cR( Z_{1:n_1}, Z_{n_1+1:n}, \hat{Z}_{n+1: n+m}). 
\@}
While permutation invariance is the only requirement, we will see  (Remark~\ref{rm:Rj_choice}) that ideally $\hat{R}_j$ shall mimic the BH output applied to p-values~\eqref{form:conf_p} at the same level. \revise{In Section~\ref{subsec:modsel}, we set $\hat{R}_j$ to approximate the BH with a selected model; in Section~\ref{subsec:loo}, we directly use the selection size of BH with the generated scores.} 
 
\noindent\underline{\textbf{Step 3: Multiple testing.}} Given the p-values~\eqref{form:conf_p}  and the auxiliary selection sizes~\eqref{eq:aux_sel_construct}, we compute preliminary selection thresholds $s_j := q \hat{R}_j /m$, and construct the final selection set following~\cite{fithian2022conditional,jin2023model}: 
{\setlength{\belowdisplayskip}{4pt}
\setlength{\abovedisplayskip}{4pt}
\@\label{eq:final_R}
\cS = \big\{ j\colon  p_j \leq s_j, ~ \xi_j  \hat{R}_{j}   \leq  r^*  \big\}, 
\@}
where $r^* := \max\big\{ r\colon  \sum_{j=1}^m \ind {\{  p_j\leq s_j, \, \xi_j  \hat{R}_j \leq r \}}   \geq r \big\}$,
with three options for $\{\xi_j\}$: (a) \emph{Heterogeneous pruning}: independently generate $\{\xi_j\} \iid \textrm{Unif}[0,1]$; (b) \emph{Homogeneous pruning}:  $\xi_j\equiv \xi$ for an independent $\xi\sim \textrm{Unif}[0,1]$; (c) \emph{Deterministic pruning}: $\xi_j\equiv 1$.  
Here, (a,b) introduce external randomness that ``smooths'' the pruning step. In numerical experiments, this improves both the power and stability of the final selection set, consistent with~\cite{jin2023model}. \rrrevise{We recommend using the homogeneous pruning due to power and stability (see empirical evidence in Appendix~\ref{app:subsec_simu_std}).}
The general procedure of OptCS is summarized in Algorithm~\ref{algo:optcs}. 

\begin{algorithm}
    \small
    \captionsetup{font=small}
    \caption{Optimized Conformal Selection (General procedure)}
    \label{algo:optcs}
\begin{algorithmic}[1]
    \REQUIRE{Labeled data $\{(X_i, Y_i)\}_{i=1}^n$, test data $\{X_{n+j}\}_{j=1}^m$, thresholds $\{c_{n+j}\}_{j=1}^m$, FDR target $q \in (0,1)$,  score-generating functional $\cV(\cdot)$, auxiliary selection function  $\cR(\cdot)$, pruning method $\texttt{prune}\in \{\texttt{homo}, \texttt{hete}, \texttt{dtm}\}$.}  
    \STATE (Optional) Split into preparatory  $\{(X_i,Y_i)\}_{i=1}^{n_1}$ and calibration data $\{(X_i,Y_i)\}_{i=n_1+1}^{n}$.
    \FOR{$j=1, \dots, m$}  
        \STATE\label{line:optcs_score} Compute conformity scores $\{V_i^{(j)}\}_{i=n_1+1}^n$ and $\{\hat{V}_{n+\ell}^{(j)}\}_{\ell=1}^m$ via~\eqref{eq:scores}. \hfill{\textcolor{Gray}{\texttt{//Subroutine 1}}} 
        \STATE Compute p-value $p_j$ via~\eqref{form:conf_p}. 
        \STATE\label{line:optcs_R} Compute auxiliary selection size $\hat{R}_j$ via~\eqref{eq:aux_sel_construct}. \hfill{\textcolor{Gray}{\texttt{//Subroutine 2}}} 
        \STATE Set $\xi_j \equiv \xi\sim \textrm{Unif}([0,1])$ if $\texttt{prune}=\texttt{homo}$, $\xi_j \iid \textrm{Unif}([0,1])$ if $\texttt{prune}=\texttt{hete}$, $\xi_j\equiv 1$ if $\texttt{prune}=\texttt{dtm}$.
    \ENDFOR
    \STATE Compute selection set $\cS $ following~\eqref{eq:final_R}.
    \ENSURE{Selection set $\cS$.}
\end{algorithmic}
\end{algorithm}

\begin{remark}[Connection to~\citet{fithian2022conditional}]
    \label{rmk:cc}
    \rrrevise{The pruning step extends the idea of conditional calibration~\citep{fithian2022conditional} which addresses FDR control with dependent p-values by finding an individual threshold (often through Monte-Carlo approximation) for each p-value conditional on a sufficient statistic. Here,  the unordered collection of the calibration samples and the
\(j\)-th test point serves as the conditioning statistic, and our  threshold \(q\widehat R_j/m\) applies exactly without the need of Monte-Carlo sampling. The key advance of OptCS is in constructing the conditionally super-uniform p-values and  thresholds---in our problem context---to which the conditional calibration idea can be applied.}
\end{remark}

\subsection{Theoretical guarantee}
\label{subsec:general_theory}

We establish general finite-sample FDR control of OptCS, provided that $\cV(\cdot)$ and  $\cR(\cdot)$ obey the appropriate conditions. 
The proof of Theorem~\ref{thm:optcs_fdr_control} is in Appendix~\ref{proof:optcs_fdr_control}.

\begin{theorem}
\label{thm:optcs_fdr_control}
    Suppose for all $j\in [m]$, the following conditions hold: (1) $(Z_{n_1+1}, \dots, Z_n, Z_{n+j})$ is exchangeable given $\{Z_i\}_{i=1}^{n_1}\cup\{\hat{Z}_{n+\ell}\}_{\ell \neq j}$; (2) the score-generating functional $\cV$ is monotone and permutation equivariant (Definitions~\ref{def:monotone} and~\ref{def:permu_equiv}); (3) the auxiliary selection functional $\cR(\cdot)$ is permutation invariant under the $j$-th null (Definition~\ref{def:R_mon_invar}).
   Then, for any nominal FDR level $q \in (0,1)$ and any of the three pruning methods, the output $\cS$ of OptCS obeys FDR control~\eqref{eq:fdr} in finite sample, 
    where the expectation is taken over both the calibration and test data. 
\end{theorem} 

We now sketch the intuitions and intermediate steps for proving Theorem~\ref{thm:optcs_fdr_control}. 
All ``nuisance'' steps, including model optimization and auxiliary threshold calculation, are conducted in a way such that the conformal p-values~\eqref{form:conf_p} remain valid (\frevise{super-uniformly} distributed) conditional on them.  This is achieved by the permutation equivariance of $\cV$ and the permutation invariance of $\hat{R}_j$. %
To begin with, under any pruning option, we have the following FDR decomposition. Lemma~\ref{lem:fdr_decomp} adapts Lemma C.1 of~\cite{jin2023model}, and we include a proof in Appendix~\ref{app:proof_fdr_decomp} for completeness. 

\begin{lemma} 
\label{lem:fdr_decomp}
    The output of Algorithm~\ref{algo:optcs} obeys  
    \[
    \fdr \leq \sum_{j=1}^m \EE\bigg[  \frac{\ind\{p_j \leq q \hat{R}_j/m,  Y_{n+j}\leq c_{n+j}\}}{\hat{R}_j}   \bigg].
    \]
\end{lemma}

It thus suffices to bound each  summation term. 
Set
\[
\bar{R}_{j}=\cR(Z_{1:n_1},Z_{n_1+1:n}, \hat{Z}_{n+1:n+j-1},Z_{n+j},\hat{Z}_{n+j+1:n+m})_j,
\]
\revise{the ``oracle'' of $\hat{R}_j$ with the $j$-th full data $Z_{n+j}$, and $[\cZ_j] = [Z_{n_1+1},\dots,Z_n,Z_{n+j}]$ as the unordered set of the $(n_2+1)$ full observations.} Under the conditions of Theorem~\ref{thm:optcs_fdr_control}, FDR control follows from two facts: 
\begin{enumerate}[label=(\alph*)]
    \item ~ $Y_{n+j}\leq c_{n+j}$ implies $\hat{R}_j = \bar{R}_j$, and $\bar{R}_j$ is measurable w.r.t.~$[\cZ_j]$. 
    \item ~ $\PP(p_j \leq t, Y_{n+j}\leq c_{n+j} \given [\cZ_j] ) \leq t$ for all $t\in [0,1]$ that is measurable w.r.t.~$[\cZ_j]$. 
\end{enumerate} 

\revise{Here, (a) bounds each expectation term by a function of $\bar{R}_j$. Since $\bar{R}_j$ depends only on the unordered set $[\cZ_j]$, the threshold $t$ is measurable with respect to $[\cZ_j]$, and therefore (b) further bounds each term by $q/m$.}

\revise{While any $\cR$ obeying Definition~\ref{def:R_mon_invar} suffices for FDR control, we remark that $\hat{R}_j$ should approximate the BH output applied to~\eqref{form:conf_p}}; see Appendix~\ref{app:subsec_discuss_R} for a formal result.  

\begin{remark}[Design principle of $\hat{R}_j$]
\label{rm:Rj_choice}
  Let $\cS_{\bh}$ be the output of BH applied to $\{p_j\}$ in~\eqref{form:conf_p} at  level $q$. Under deterministic pruning,
the output of OptCS satisfies $\cS\subseteq\cS_{\rm BH}$. For all three pruning options, if $\hat{R}_j \equiv |\cS_{\bh}|$ for all $j\in [m]$, then $\cS=\cS_{\bh}$. Thus, $\hat{R}_j$ should  approximate $|\cS_{\bh}|$  to achieve high power. We follow this principle to derive procedures in the next section.
\end{remark}

\section{Instantiations of OptCS}
\label{sec:procedures}
In this section, we specify  $\cV(\cdot)$ and $\cR(\cdot)$ in Algorithm~\ref{algo:optcs}, which leads to three instantiations of OptCS. %
We exploit two ways to improve the power: (i) select a good model among multiple choices, and  (ii) use all labeled data for model fitting. 
Section~\ref{subsec:modsel} proposes OptCS-MSel which allows (i).
Section~\ref{subsec:loo} proposes OptCS-Full which allows (ii). 
Section~\ref{subsec:loo_sel} proposes OptCS-Full-MSel which simultaneously achieves (i) and (ii).

\subsection{OptCS-MSel: Model selection over trained models}
\label{subsec:modsel}

The first scenario we consider is when multiple pre-trained models have been obtained in a separate training stage, and a scientist aims to use data at hand to select \frevise{a powerful model} and produce an FDR-controlling selection set. 
\rrevise{
Such scenarios are common in modern machine learning tasks, such as  drug discovery  where pre-trained models for drug properties are available yet too costly to retrain, or  computer vision and natural language processing where it is standard to employ powerful pre-trained foundation models. OptCS-MSel  ``approximates'' the selection performance of all the model choices and selects the best while preserving validity.}

\begin{figure} 
    \centering
    \includegraphics[width=0.8\linewidth]{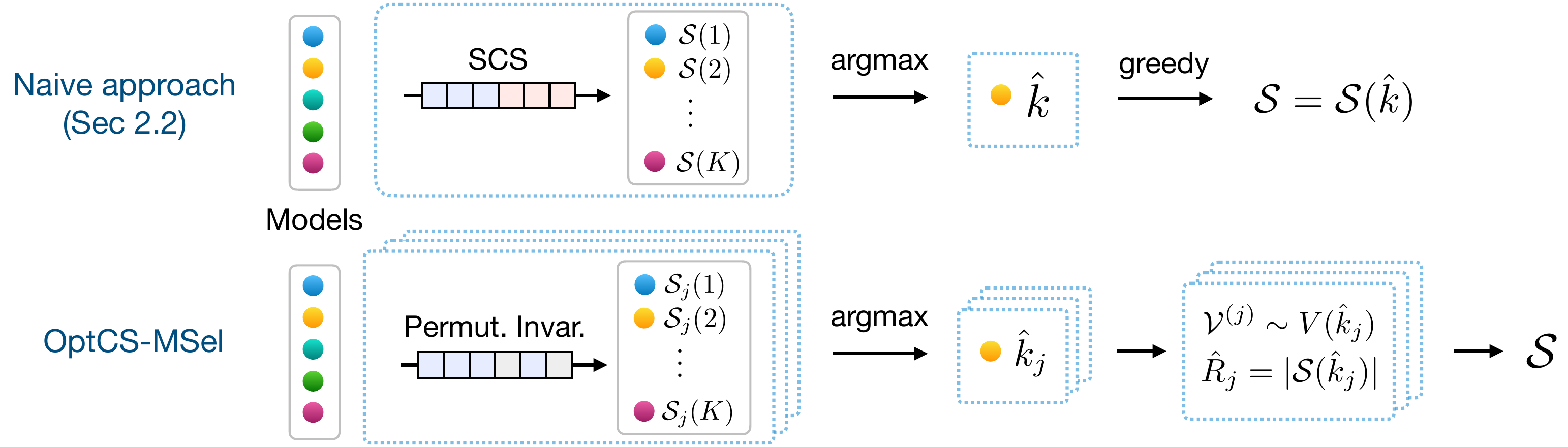}
    \caption{{\small OptCS-MSel modifies the greedy approach by selecting a model $\hat{k}_j$ in a way that is permutation invariant to calibration data and the $j$-th test point, replacing \basename in the evaluation with a permutation-invariant $\mathcal{S}_j(\cdot)$.}} %
    \label{fig:idea_optcs_msel}
\end{figure} 

Following  Section~\ref{subsec:challenges}, we assume $K$ candidate scores $\{v(\cdot,\cdot;k)\colon \cX\times\cY\to \RR\}_{k=1}^K$ independent of the calibration and test points, %
\rrevise{e.g., obtained from an external pre-training stage or trained on an independent dataset or the first $n_1$ samples (not necessarily exchangeable with the target domain) optionally reserved to fit the models.} 
Also recall $V_i(k) = v(X_i,Y_i;k)$ and $\hat{V}_{n+j}(k) = v(X_{n+j},c_{n+j};k)$ are conformity scores using model $k\in [K]$.  
As discussed, greedily selecting the score with the largest selection set in \basename invalidates FDR control.  

\rrevise{OptCS-MSel restores validity by \emph{permutation-invariant} model selection. 
Following Remark~\ref{rmk:general_V_idea}, we construct the $j$-th set of scores using a model $\hat{k}_j$ selected in a way that is permutation invariant to $\{Z_i\}_{i=n_1+1}^n\cup \{\hat{Z}_{n+j}\}$.} %
For each $j\in [m]$, let $\cS_j(k)$ be 
the BH output when applied to $\{\tilde{p}_\ell^{(j)}(k)\}_{\ell\in[m]\backslash\{j\}} \cup \{0\}$ at nominal level $q$, where  %
\[
    \tilde{p}_\ell^{(j)} (k) = \frac{\sum_{i=n_1+1}^{n} \ind\{V_i(k) \leq \hat{V}_{n+\ell}(k)\}+ \ind\{\hat{V}_{n+j}(k)\leq \hat{V}_{n+\ell}(k)\}}{n_2+1}.
\]
Then, we set $\hat{k}_j = \argmax_k \big|\cS_j(k)\big|$, and construct  $\cV$ in Line~\ref{line:optcs_score} of Algorithm~\ref{algo:optcs} via 
\[
\cV^{(j)}(Z_1, \dots, Z_n, \hat{Z}_{n+1}, \dots, \hat{Z}_{n+m}):=  \big(V_{n_1+1}(\hat{k}_j),  \dots, V_{n}(\hat{k}_j), \hat{V}_{n+1}(\hat{k}_j), \dots, \hat{V}_{n+m}(\hat{k}_j) \big),
\]
as well as  
$\hat{R}_j = \big|  \cS_j(\hat{k}_j) \big|$ in Line~\ref{line:optcs_R} of Algorithm~\ref{algo:optcs}. We visualize the comparison between OptCS and the naive approach  in Figure~\ref{fig:idea_optcs_msel}. 
\revise{A crucial design here is the term $\ind\{\hat{V}_{n+j}(k)\leq \hat{V}_{n+\ell}(k)\}$ in $\tilde{p}_\ell^{(j)}$ to achieve the permutation invariance of $\cV$ and $\cR$.}

Intuitively, OptCS-MSel improves power because each $\cS_j(k)$ approximates the SCS selection set $\cS(k)$ (Section~\ref{subsec:challenges}) since $\{\tilde{p}_\ell^{(j)}(k)\}$ only slightly differ from the SCS conformal p-values with the $k$-th model.  As such, we expect  $\hat{k}_j \approx \argmax_{k} |\cS(k)|$ (or at least a model with strong performance). Accordingly, each p-value $p_j$ in~\eqref{form:conf_p} based on the selected models (e.g., $\hat{V}_{n+\ell}^{(j)}=\hat{V}_{n+\ell}(\hat{k}_j)$)  tends to be more powerful in distinguishing the non-null from the null and improves the selection performance. Since \(\cS_j(k)\) is a proxy for the aggregate downstream selection performance of the
\(k\)-th fixed score, we expect OptCS-MSel to approximate the best fixed-score SCS in terms of selection set size. 
\frevise{Accordingly, the selected index \(\hat k_j\) should not be
interpreted as the locally optimal model for unit \(j\); instead, this arises from the mechanism used to preserve validity when performing model selection.}

We establish OptCS-MSel's FDR control below, with proof in Appendix~\ref{proof:modsel_valid}.

\begin{prop}
\label{lemma:modsel_valid}
    The score-generating functional $\cV$ and auxiliary selection function $\cR$ above obey Definitions~\ref{def:monotone},~\ref{def:permu_equiv} and~\ref{def:R_mon_invar}. Therefore, the output of \textnormal{OptCS-MSel} obeys $\fdr\leq q$.
\end{prop}

While finite-sample power analysis of OptCS-MSel is challenging, we show that under mild regularity conditions, OptCS-MSel asymptotically matches the best fixed-score SCS procedure among the \(K\) candidate scores, where performance is measured by the size of the selection set. We define the power of a selection set  $\cS\subseteq [m]$  as the random variable $\text{Power}(\cS) =  \sum_{j\in \cS} \ind\{Y_{n+j}>c_{n+j}\}/m$. The following proposition shows that under standard assumptions, the selection set of OptCS-MSel is as powerful as the largest selection set among the outputs of SCS with $K$ candidate conformity scores. Its proof is deferred to Appendix~\ref{app:subsec_proof_power_msel}.

\begin{prop}\label{prop:power_msel}
    Suppose $n_2=n$, the data $\{(X_i,Y_i,c_i)\}_{i=1}^{n+m}$ are i.i.d., and consider $K$ fixed, pre-trained conformity score functions $\{v(\cdot,\cdot;k)\}_{k=1}^K$. Fix a nominal FDR level $q\in (0,1)$. Let $\cS_{k,\textnormal{SCS}}$ be the output of conformal selection at nominal level $q$ with the $k$-th score. Let $\cS_{\textnormal{OptCS}}$ be the output of \textnormal{OptCS-MSel} at nominal level $q$ with homogeneous pruning. We define $F(s;k) = \PP(v(X_i,Y_i;k)\leq s)$ for any $s\in \RR$ and $k\in [K]$, and 
    $
    G(t;k) = \PP ( F( v(X_i,c_i;k) ;k) \leq t  ) 
    $ for any $t\in [0,1]$. 
    Define  $t^*_k = \sup\{t\colon t/G(t;k)\leq q\}$, $k\in [K]$. Assume the following conditions hold:
    \begin{enumerate}[label=(\roman*)]
        \item For each $k\in [K]$, $G(t;k)$ is continuous near $t_k^*$, and there exists $t'\in (t_k^*-\epsilon, t_k^*)$ with $t'/G(t';k)<q$ for any sufficiently small $\epsilon>0$. 
        \item There exists a unique maximizer $k^*\in [K]$ of $\{G(t_k^*;k)\}_{k\in[K]}$. 
    \end{enumerate}
    Let $\hat{k} = \argmax_{k\in [K]}|\cS_{k,\textnormal{SCS}}|$ be the score function that leads to the largest selection set with SCS. Then as $m,n\to \infty$, it holds that $\PP(\hat{k} = k^*) \to 1$, and $|\textnormal{Power}(\cS_{\textnormal{OptCS}}) - \textnormal{Power}(\cS_{\hat{k},\textnormal{SCS}})|\stackrel{P}{\to}0$. 
\end{prop}

\begin{remark} 
\rrevise{Given the candidate scores,  OptCS-MSel does not involve $\{(X_i, Y_i)\}_{i=1}^{n_1}$. Consequently, its validity requires only exchangeability of the calibration/test samples. The candidate models may be trained over arbitrary external data (viewed as $\{(X_i, Y_i)\}_{i=1}^{n_1}$ here), and additional processing steps, such as fine-tuning or adapting for distribution shift, can be done on the preparatory data before applying OptCS-MSel, or one may use  OptCS-Full-MSel in Section~\ref{subsec:loo_sel} if they are affordable.} 
\end{remark}

\subsection{OptCS-Full: Conformal selection without data splitting}
\label{subsec:loo}

In this part, we introduce OptCS-Full, a procedure that avoids sample splitting when model training with labeled data is needed. Setting aside the model selection issue, we consider a fixed model class: the prediction model will be trained through a given process and wrapped by a conformity score  in a given way. 
The goal is to improve power by training a more accurate model with all labeled data. %

The training process is an algorithm $\cA$ that inputs a set of labeled data and outputs a trained model, e.g.,
$
\cA\colon \cup_{N\geq 0}(\cX\times\cY)^N \to \{\textrm{measurable functions}~\hat\mu\colon \cX\to \cY\}.
$ 
We require that $\cA$ treats the input data symmetrically:\footnote{If $\cA$ is randomized, this only needs to hold in a distributional sense. Thus, any algorithm can be made symmetric by randomly permuting the input data before feeding them into the algorithm.} 
{\setlength{\belowdisplayskip}{4pt}
\setlength{\abovedisplayskip}{4pt}\@\label{eq:training_symmetry}
\cA\big( (x_1,y_1),\dots,(x_N,y_N) \big) = \cA\big( (x_{\pi(1)},y_{\pi(1)}),\dots,(x_{\pi(N)},y_{\pi(N)}) \big)
\@}
for any permutation $\pi$ of $[N]$. 
Here $\hat\mu$ refers to the trained output.\footnote{With slight abuse of notation, $\hat\mu$ may  represent general models other than a point predictor (e.g., quantile/mean regression, conditional c.d.f.~models).} 
The score function wraps around $\hat\mu \colon \cX\to \cY$ in a given way, which can be written as
$v(\cdot,\cdot\given \hat\mu) \colon \cX\times\cY\to \RR,$
whose functional dependence on $\hat\mu$ is fixed. 

\frevise{To preserve permutation invariance in training with labeled data,}
we binarize the labels to $\cY=\{0,1\}$ and $c=0$, which can be done by setting $\tilde{Y}_i=\ind\{Y_i > c_i\}$. 
\frevise{This preserves the oracle score ranking even with imputed labels (see our discussion after Proposition~\ref{lemma:loo_full_valid} and Appendix~\ref{app:subsec_null_impute_discuss}), but may discard magnitude information in finite samples in regression problems. 
If this is a concern, a general transformation without interrupting validity is $\tilde{Y}_i = T(Y_i,c_i)$ for any function $T(\cdot,\cdot)$ obeying $T(y,c)=T(c,c)$ whenever $y\leq c$.\footnote{\frevise{For example, $T(y,c)=(y-c)_+$ satisfies this condition and retains the magnitude of $y$ above the threshold, which might be a useful alternative in regression problems. Appendix~\ref{app:subsec_encoding} empirically compares binary and positive-part encoding in regression for OptCS-Full and OptCS-Full-MSel.}} 
In general, we denote the imputed test label $\tilde{c}_{n+j}=T(c_{n+j},c_{n+j})$, the imputed test point $\hat{Z}_{n+j}^T=(X_{n+j},\tilde{c}_{n+j})$, and the full observation $Z_{i}^T=(X_{i},\tilde{Y}_{i})$, with the superscript indicating data for training.}

The main idea of OptCS-Full is to involve all the labeled data $\{Z_i\}_{i=1}^n$ in training \rrevise{while obeying the conditions in OptCS. Following Remark~\ref{rmk:general_V_idea}, the key is to ensure the $j$-th training process treats $\{Z_i\}_{i=n_1+1}^n$ and $\hat{Z}_{n+j}$ symmetrically. When the $j$-th test data with imputed null label enters training, the model may be biased towards predicting $Y_{n+j}$ as null and weaken the score. To enhance power, we employ  leave-one-out (LOO)  where the $j$-th point is evaluated by a model that does not use it in training.}

We now present a procedure following the intuitions above.
Inspired by Full Conformal Prediction~\citep{vovk2005algorithmic}, %
we train a model $\hat\mu_\ell := \cA(\mathbf{D}_{-\ell})$ with
{\setlength{\belowdisplayskip}{4pt}
\setlength{\abovedisplayskip}{4pt}\@\label{eq:loo_train_data}
\mathbf{D}_{-\ell}  := \{Z_i^T\}_{i=1}^{n_1} \cup \{Z_i^T\}_{i=n_1+1}^n \cup \{\hat{Z}^T_{n+j}\}_{j=1}^m \setminus \{\tilde{Z}_{\ell}^T \}
\@}
for each $\ell\in\{n_1+1,\dots,m+n\}$, 
where we denote $\tilde{Z}^T_\ell=Z^T_\ell$ when $\ell \leq n$ and $\tilde{Z}^T_\ell =\hat{Z}^T_\ell$  otherwise.  
Then, we define the score-generating functional $\cV$ via $\cV^{(j)} \equiv  (V_{n_1+1},\dots,V_{n},\hat{V}_{n+1},\dots,\hat{V}_{n+m})$, where (recall $Z_i=(X_i,Y_i)$ and $\hat{Z}_i=(X_i,c_i)$)
{\setlength{\belowdisplayskip}{4pt}
\setlength{\abovedisplayskip}{4pt}\@\label{eq:def_full_score}
V_i = v(Z_{i } \given \hat\mu_{i }) \text{ for } n_1<i \leq n, \quad \text{and} \quad
\hat{V}_{n+j} = v(\hat{Z}_{n+j} \given \hat\mu_{n+j}) \text{ for } j \in [m].
\@} 
In other words, the score of each calibration and test point  is evaluated by a model trained on all data except that specific point.\footnote{Note that OptCS-Full needs to train $(n_2+m)$ many models. While we set $n_1=0$ in our experiments, one may increase $n_1$ to reduce computation complexity at a price of coarsened p-values (the training process still uses $n$ labeled data, yet only $n_2$ samples are used in the p-values).} 
\rrevise{Here, without LOO training, the $\hat\mu_i$'s can be replaced by a model trained on $\{Z_i^T\}_{i=1}^n\cup\{\hat{Z}^T_{n+j}\}_{j=1}^m$, which remains valid and reduces computational costs, but may empirically lose power; see Appendix~\ref{app:comp_other_full}.}

\begin{remark}[Design rationale]
Given the $j$-th null imputation and LOO training ideas stated earlier, in~\eqref{eq:def_full_score} we use all-test null imputation and include all of $\{\hat{Z}_{n+j}^T\}_{j=1}^m$ to share the models across all the test points to reduce computation complexity. An alternative method that performs LOO training over $\{Z_i^T\}_{i=1}^n\cup\{\hat{Z}_{n+j}^T\}$ only is introduced in Appendix~\ref{app:subsec_discuss_loo}, which can be useful when label dilution is a concern.
\end{remark}

The main intuition for power improvement with OptCS-Full is that the model $\hat\mu_i$ tends to be more accurate than that in SCS since it is trained on more labeled data. Indeed, if the algorithm $\cA$ is stable to leave-one-out training, one could imagine $\hat\mu_i\approx \hat\mu$, where $\hat\mu$ is the ``ideal'' model trained on $n$ labeled samples. Thus, OptCS-Full mimics a procedure where the $n$ labeled samples are involved in both training and p-value construction, yet in a way that preserves the FDR control.

It turns out that applying BH to the p-values~\eqref{form:conf_p} with these scores suffices for FDR control without pruning.  
Proposition~\ref{lemma:loo_full_valid} formalizes this idea and proves the validity of OptCS-Full, whose proof is in Appendix~\ref{proof:loo_full_valid}. Algorithm~\ref{algo:optcs-full} summarizes OptCS-Full. An alternative analysis shows the oracle p-values are PRDS by adapting~\cite{marandon2024adaptive}, hence the FDR control also follows similar to~\cite{jin2023selection}. We defer the details to Appendix \ref{app:subsec_alt_full_proof}. %

\begin{algorithm}
    \small
    \captionsetup{font=small}
    \caption{OptCS-Full}
    \label{algo:optcs-full}
\begin{algorithmic}[1]
    \REQUIRE{Labeled data $\{(X_i, Y_i)\}_{i=1}^n$, test data $\{X_{n+j}\}_{j=1}^m$, thresholds $\{c_{n+j}\}_{j=1}^m$, FDR target $q \in (0,1)$, symmetric training algorithm $\cA$, outcome encoding function $T(\cdot)$.}  
    \STATE (Optional) Split  preparatory data $\{(X_i,Y_i)\}_{i=1}^{n_1}$ and calibration data $\{(X_i,Y_i)\}_{i=n_1+1}^{n}$.
    \STATE Compute $Z_i^T=(X_i,T(Y_i,c_i))$, $i\in[n]$, and $\hat Z_{n+j}^{T} = (X_{n+j},T(c_{n+j},c_{n+j}))$, $j\in[m]$. 
    \FOR{$\ell=n_1+1,\dots,n+m$}  
        \STATE Train model $\hat\mu_\ell = \cA(\{Z_i^T\}_{i=1}^n \cup \{\hat{Z}_{n+j}^T\}_{j=1}^m \backslash \{\tilde{Z}_\ell^T\})$. \hfill{\textcolor{gray}{\texttt{// leave-one-out training}}}
    \ENDFOR 
    \STATE Compute $V_i=v(X_i,Y_i\given \hat\mu_i)$, $n_1<i\leq n$, and $\hat{V}_{n+j}=v(X_{n+j},c_{n+j}\given \hat\mu_{n+j})$, $j\in[m]$. 
    \FOR{$j=1,\dots,m$}  
        \STATE Compute $p_j$ via~\eqref{form:conf_p} with $V_i^{(j)}\equiv V_i$, $i=n_1+1,\dots,n$, and $\hat{V}_{n+\ell}^{(j)}\equiv \hat{V}_{n+\ell}$, $\ell\in [m]$.
    \ENDFOR 
    \STATE Compute selection set $\cS $ by applying the BH procedure to $\{p_j\}$ at level $q$.  
    \ENSURE{Selection set $\cS$.}
\end{algorithmic}
\end{algorithm}

\begin{prop}
\label{lemma:loo_full_valid}
    Suppose the training process $\cA$ of  $\{\hat\mu_\ell\}_{\ell=n_1+1}^{n+m}$ obeys~\eqref{eq:training_symmetry}. Let $\cS$ be the output of Algorithm~\ref{algo:optcs} and  $\cS_\bh$ be that  of  \textnormal{BH} applied to p-values~\eqref{form:conf_p} with scores~\eqref{eq:def_full_score}. 
    Then, the score-generating functional $\cV$ in \textnormal{OptCS-Full} by~\eqref{eq:def_full_score} obeys Definitions~\ref{def:monotone} and~\ref{def:permu_equiv}. Also, setting $\hat{R}_j \equiv |\cS_{\bh}|$ in Algorithm~\ref{algo:optcs} yields  $\cS = \cS_{\bh}$, and  $\hat{R}_j$ obeys a relaxed version of Def.~\ref{def:R_mon_invar} \frevise{(which requires $\hat{R}_j=\bar{R}_j$ on the event $\{p_j\leq s_j,Y_{n+j}\leq c_{n+j}\}$ for some  $\bar{R}_j$ obeying Definition~\ref{def:R_mon_invar})} that yields finite-sample FDR control.  
\end{prop}

\rrrevise{We now discuss the computation--power tradeoff induced by all-test null imputation and
leave-one-out training in OptCS-Full.  
After imputation, the augmented training data~\eqref{eq:loo_train_data} follows a mixture of the distribution of $(X_i,Y_i)$'s and the imputed $(X_{n+\ell},0)$'s. 
This induces class imbalance which can be addressed via a 
standard and effective oversampling strategy~\citep{mohammed2020machine}. Although alternatives like class-weighting  can be used, oversampling usually dominates; see Appendix~\ref{app:comp_other_full} for a comparison.} 
Interestingly, however, the null imputation \frevise{after binarization} rescales the \emph{oracle} score~\citep{jin2023selection} by a constant factor and preserves its ranking; thus, the rank-based oracle p-values (and asymptotic power under consistent learning) are unchanged (see discussion in Appendix~\ref{app:subsec_null_impute_discuss}). 

In finite samples, null imputation
can dilute the supervised signal while the LOO training increases computation. This leads to a natural
computation--power spectrum. 
Removing LOO from OptCS-Full involves only one model fit and remains valid \frevise{(see Appendix~\ref{app:nonLOO} for a runtime comparison of the analogous Non-LOO variant of OptCS-Full-MSel)} but may lose power due to self-influence. 
At the more computationally intensive end, we present a per-test variant (OptCS-Full-sep) in Appendix~\ref{app:subsec_discuss_loo} which ensures that each training set contains at most one
null-imputed test point, at the cost of \(m(n_2+1)\) fits. Appendix~\ref{app:subsec_discuss_loo} also presents an intermediate \(K_0\)-fold variant with $m+K_0n_2$ model fits. Appendix~\ref{app:subsec_null_impute_experiment} empirically
studies the impact of null imputation and the computation--power tradeoff. We observe robust
power of OptCS-Full in moderate regimes where $m$ is comparable to $n$. 
OptCS-Full-sep is preferred when $m$ is large relative to $n$.

Data reuse in model training and selection is also considered in recent
novelty detection~\citep{marandon2024adaptive,liang2024integrative,lee2025full} and  adaptive CS methods~\citep{gui2025acs}, although the techniques differ.
A comprehensive comparison in Appendix~\ref{app:subsec_exp_od} shows the power of OptCS-Full via full-data training.

\subsection{OptCS-Full-MSel: Model training and selection without data splitting}
\label{subsec:loo_sel}

Finally, we combine the preceding two ideas to propose OptCS-Full-MSel, a variant of Algorithm~\ref{algo:optcs} that allows \rrrevise{score} training, selection, and FDR-controlling conformal selection with all labeled data. %
Figure~\ref{fig:idea_optcs-full-msel} visualizes the pipeline of OptCS-Full-MSel.

\begin{figure} 
    \centering
    \includegraphics[width=\linewidth]{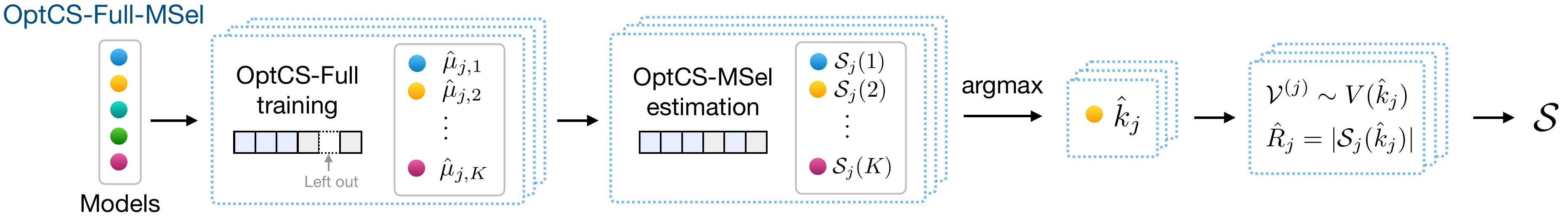}
    \caption{{\small OptCS-Full-MSel uses all labeled data for model training (similar to OptCS-Full),  model selection (similar to OptCS-MSel), and final multiple testing.}}
    \label{fig:idea_optcs-full-msel}
\end{figure}

We follow Section~\ref{subsec:loo} and assume $K$ training processes $\{\cA_k\}$ obeying~\eqref{eq:training_symmetry} which lead to $K$ candidate score functions. Likewise, we encapsulate their dependence on fitted models by 
$
v(\cdot,\cdot\given \hat\mu_k,k) \colon \cX\times\cY\to \RR,
$
where  $\hat\mu_k$ is a trained prediction model.
\frevise{Each $\cA_k$ may use its own training label encoding $T_k(\cdot,\cdot)$ obeying $T_k(y,c)=T_k(c,c)$ whenever $y\leq c$.
We write $Z_i^{T_k}$ and $\widehat Z_{n+j}^{T_k}$ for the correspondingly transformed records, and $\mathbf{D}^{(k)}_{-\ell}$ for the leave-one-out training set defined as in~\eqref{eq:loo_train_data} with $T$ replaced by $T_k$, $\ell=n_1+1,\dots,n+m$.}

In model training, like OptCS-Full, we compute $\hat\mu_{\ell,k} = \cA_k\big( \mathbf{D}_{-\ell}^{(k)}  \big)$, and  define $V_i(k) = v(X_i,Y_i \given \hat\mu_{i,k}, k)$, $\hat{V}_{n+j}(k) = v(X_{n+j}, {c}_{n+j}\given \hat\mu_{n+j,k},k)$.
Then, we follow OptCS-MSel to estimate the selection performance by $\cS_j(k)$, defined as the output of BH applied to  $\{\tilde{p}_\ell^{(j)}(k)\}_{\ell \neq j}\cup \{0\}$ at nominal level $q$, where %
{\setlength{\belowdisplayskip}{6pt}
\setlength{\abovedisplayskip}{6pt}\$
\tilde{p}_\ell^{(j)}(k) =  \frac{\ind\{\hat{V}_{n+j}(k)\leq \hat{V}_{n+\ell}(k)\} +\sum_{i=n_1+1}^n \ind\{V_i(k)\leq \hat{V}_{n+\ell}(k)\}}{n_2+1}, \quad \ell \neq j.
\$}

Finally, we complete the definition of score generating function $\cV$ in~\eqref{eq:def_V_full_msel} (Line~\ref{line:optcs_score} of Algorithm~\ref{algo:optcs}) and auxiliary selection function $\cR$ (Line~\ref{line:optcs_R} of Algorithm~\ref{algo:optcs}) via 
\@
&\cV^{(j)} = \big( V_{n_1+1}(\hat{k}_j),\dots,V_{n}(\hat{k}_j), \hat{V}_{n+1}(\hat{k}_j),\dots,\hat{V}_{n+m}(\hat{k}_j)\big),\label{eq:def_V_full_msel} \\ 
&\textrm{where} \quad 
 \hat{k}_j =  \text{argmax}_{k\in[K]} \big| \cS_j(k) \big|,\quad \text{and}\quad  \hat{R}_j = \big|\cS_j(\hat{k}_j)\big|. \label{eq:def_R_full_msel}
\@  

\revise{As OptCS-Full-MSel combines the preceding two ideas, the intuitions for its power improvement are twofold. First, each model $\hat\mu_{i,k}$ trained with all the labeled data is often more powerful \rrrevise{than the one trained with a subset of labeled data}, so each $\hat{V}_{n+j}(k)$ improves upon the split version. In addition, $\hat{V}_{n+j}(\hat{k}_j)$ with the adaptively selected model $\hat{k}_j$ further boosts the power of the p-values.}
Proposition~\ref{prop:validity_loosel} guarantees the validity of OptCS-Full-MSel, whose proof is in Appendix~\ref{proof:validity_loosel}.

\begin{prop}
\label{prop:validity_loosel}
Suppose $\cA_k$ obeys~\eqref{eq:training_symmetry} for all $k\in[K]$. Then, the functional $\cV$  specified by~\eqref{eq:def_V_full_msel} obeys Definitions~\ref{def:monotone} and~\ref{def:permu_equiv}, and the function $\cR$ specified by~\eqref{eq:def_R_full_msel} obeys Definition~\ref{def:R_mon_invar}. Therefore, the output of \textnormal{OptCS-Full-MSel} obeys $\fdr\leq q$. 
\end{prop}

\noindent\textbf{Computation--power tradeoffs.}
As discussed in the introduction (Table~\ref{tab:optcs-features}), the OptCS variants suit diverse computational regimes. The default LOO implementations of OptCS-Full and OptCS-Full-MSel require \((n_2+m)\) and \(K(n_2+m)\) model fits, respectively. They are most appropriate when the score model is lightweight enough for repeated refitting (such as  post-hoc predictors built from fixed LLM outputs in Section~\ref{subsec:app_LLM}). OptCS-MSel selects among pre-trained candidate scores without retraining.   
\frevise{For full-data training, the Non-LOO variants fit
one model per candidate on all labeled and null-imputed test
observations, requiring only $1$ and $K$ fits for OptCS-Full
and OptCS-Full-MSel, respectively, with the subsequent
score-selection and testing steps unchanged.
They retain finite-sample FDR control and provide practical
alternatives when repeated fitting is costly (see Appendix~\ref{app:comp_other_full} for empirical comparisons). The variants in Appendix~\ref{app:subsec_discuss_loo} reduce test-label imputation to improve power at additional computational cost.} 

% !TEX root = main.tex
\section{Simulation studies} \label{sec:simu}

\subsection{Conformity score selection with pre-trained models}
\label{subsec:simu_score_sel}
 
We first evaluate  OptCS-MSel in selecting pre-trained conformity scores while maintaining FDR control. 
The competing methods include:
\begin{enumerate}[label=(\roman*).]
    \item \texttt{Greedy}: 
    Select the candidate conformity score function which leads to the largest selection set in SCS. %
    \item \texttt{OptCS-MSel\_homo}, \texttt{OptCS-MSel\_hete}, and \texttt{OptCS-MSel\_dtm}: Our method with homogeneous, heterogeneous, and deterministic pruning.
    \item \texttt{Base\_random}: Randomly pick a conformity score and use it in SCS.
    \item \texttt{Base\_cal\_split}: Randomly split the calibration set into three folds: $\cD_{\text{calib\_sel}}$ (25\%), $\cD_{\text{test\_sel}}$ (25\%), and $\cD_{\text{calib}}'$ (50\%). We select the score which leads to the largest selection set in \basename with $\cD_{\text{calib\_sel}}$ as the calibration set and $\cD_{\text{test\_sel}}$ as the test set. We then run \basename with calibration set $\cD_{\text{calib}}'$. %
    \item \texttt{Base\_tr\_split}: Similar to (iv),  but we split $\cD_\train$ into $\cD_{\text{calib\_sel}}$ (25\%), $\cD_{\text{test\_sel}}$ (25\%), and $\cD_{\text{train}}'$ (50\%). After training models on  $\cD_{\text{train}}'$, we use $\cD_{\text{calib\_sel}}$ (25\%), $\cD_{\text{test\_sel}}$ (25\%) to select a score. 
\end{enumerate}
We construct 8 data generating processes, adapted from \cite{liang2024conformal} (with 11 models) and  \cite{jin2023selection} (with 24 models), all with $c_{n+j}\equiv 0$. %
\revise{We use conformity scores  $v(X_i, Y_i) = M \cdot \ind\{Y_i > c_i\} - \hat\mu(X_i)/\hat\sigma(X_i)$, where $\hat\mu$ and $\hat\sigma$ depend on model choices.} The details are in Appendix~\ref{subsec:simulation_setup}. 

\begin{figure}[htbp]
    \centering
    \captionsetup{font=small}
    \includegraphics[width=0.8\linewidth]{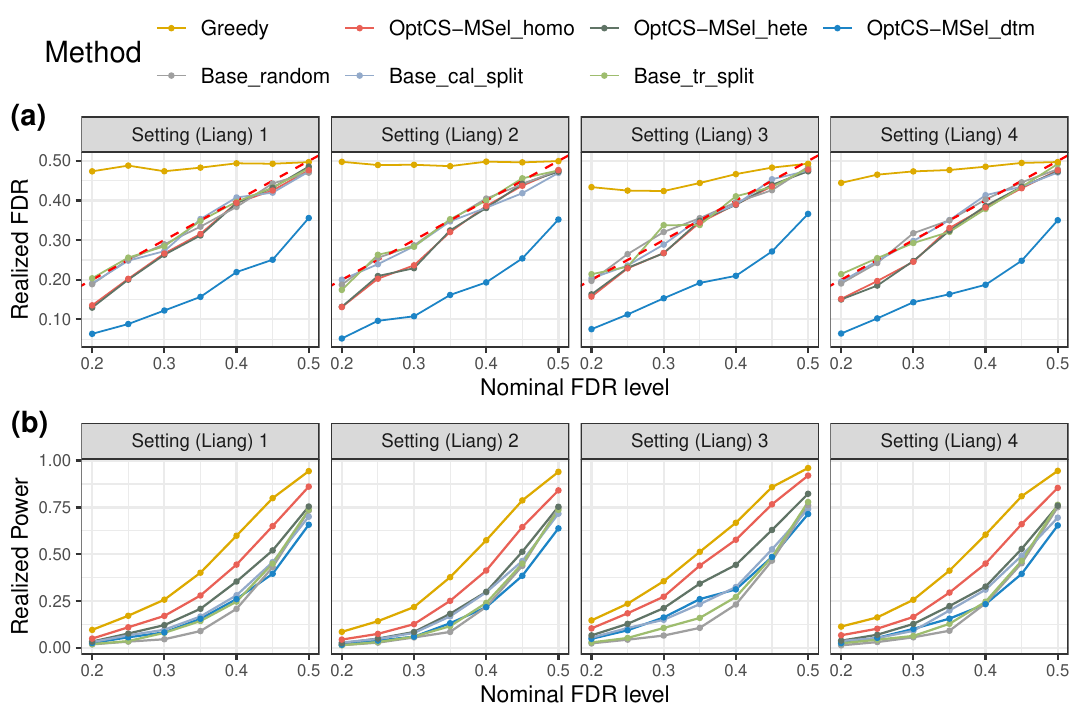}
    \caption{{\small Realized FDR (a) and power (b) for varying nominal FDR levels under Liang's settings for conformity score selection. Each \frevise{column} corresponds to one data generating process. The selection power differs significantly across individual conformity scores.}}
    \label{fig:simu_score_sel_Liang}
\end{figure}

\begin{figure}[htbp]
    \centering
    \captionsetup{font=small}
    \includegraphics[width=0.85\linewidth]{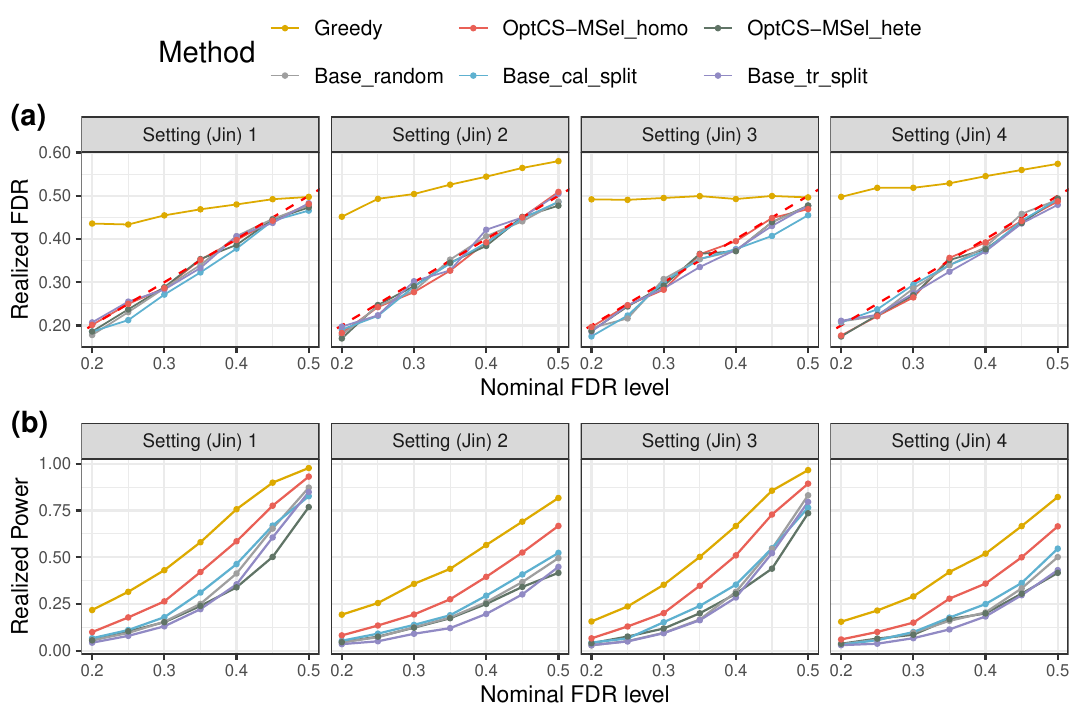}
    \caption{Realized FDR (a) and power (b) for varying nominal FDR levels under Jin's settings. Each subplot corresponds to one data generating process. In these settings, the selection power with each individual conformity score is similar to each other.}
    \label{fig:simu_score_sel_Jin}
\end{figure}

\noindent\textbf{Simulation results.} 
Figure~\ref{fig:simu_score_sel_Liang} reports the empirical FDR and power over $N=500$ independent runs for Liang's settings with $n_1=n_2 = m = 100$. Results in Jin's settings are in Figure~\ref{fig:simu_score_sel_Jin}.
The empirical FDR is the mean of $|\cS\cap \cH_0|/(1 \vee |\cS|)$, and the power is the mean of $|\cS\cap \cH_1|/|\cH_1|$, where $\cH_0 = \{j\in [m]\colon Y_{n+j}\leq c_{n+j} := 0\}$ and $\cH_1 = \{j\in [m]\colon Y_{n+j}>c_{n+j}\}$.  
We observe drastic violation of FDR with the \texttt{Greedy} method due to the double-dipping bias.
For all 8 settings and across all nominal FDR levels, \texttt{OptCS-MSel\_homo} consistently outperforms all competing methods that maintain valid FDR control; \texttt{OptCS-MSel\_hete} outperforms the baselines in Liang's settings, where the candidate scores differ substantially in quality. 
Consistent with~\cite{jin2023model}, deterministic pruning has lower power and does not exhaust the FDR budget. We thus omit it in later results.

\noindent\textbf{Additional studies.} \rrrevise{Results in Appendix~\ref{subsec:more_simu} with varied $(n_1,n_2)$ confirm consistent power of OptCS-MSel across various sample sizes. Appendix~\ref{app:subsec_simu_std}  assesses the  stability of these methods, where, among the three pruning options, \texttt{OptCS-MSel\_homo} generally leads to the lowest standard deviation and \texttt{OptCS-MSel\_dtm} tends to have the highest. The random pruning tends to smooth the selection, improving both power and stability. 
}

\subsection{Conformal selection without data splitting}
\label{subsec:simu_loo}

Next, we evaluate OptCS-Full in Section~\ref{subsec:loo} against data splitting baselines. The competing methods include: (i) \texttt{OptCS-Full\_os}: Algorithm~\ref{algo:optcs-full} with over-sampling. 
(ii) \texttt{OptCS-Full\_sep}: The variant in Appendix~\ref{app:subsec_discuss_loo} that reduces null test samples in training. (iii-v) \texttt{Base\_split\_0.75/0.50/0.25}: Randomly split labeled data into  $\cD_{\train}$ (75\%/50\%/25\%) and $\cD_\calib$ (25\%/50\%/75\%), then apply \basename with $\cD_\train$ as training data and $\cD_\calib$ as calibration data. Detailed setup is in Appendix~\ref{subsec:loo_setup}.

\begin{figure}[htb]
	\centering
	\includegraphics[width=0.75\linewidth]{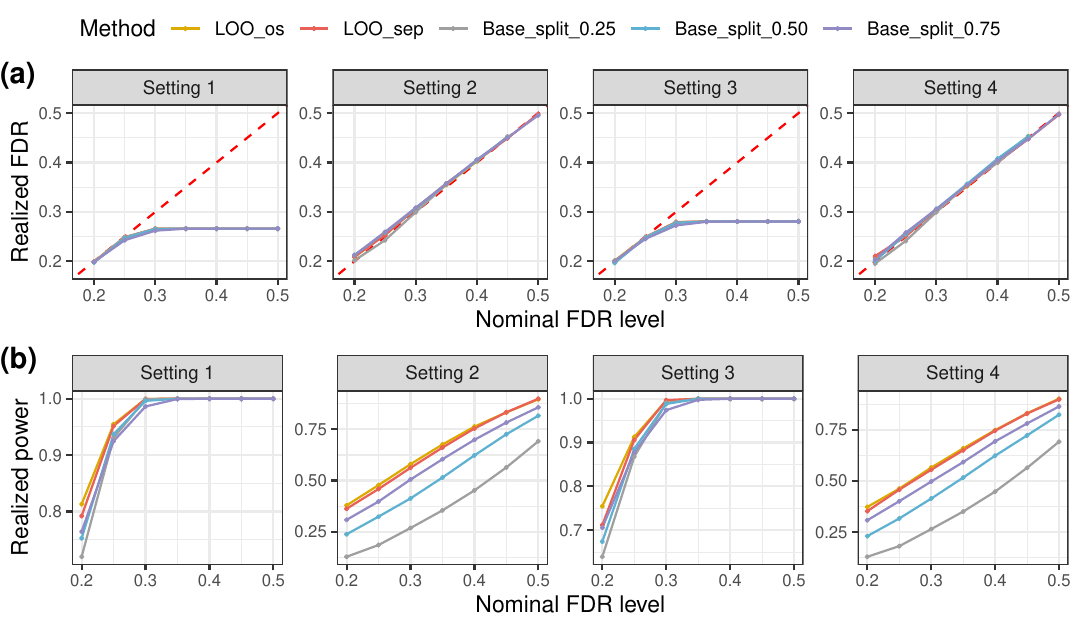}
	\caption{{\small Empirical FDR (a) and power (b) for OptCS-Full and \basename using SVR at various nominal FDR levels with $n=500$, $m=100$, averaged over $N=500$ runs.}} %
\label{fig:simu_LOO_onemodel}
\end{figure}

\noindent\textbf{Simulation results.} %
Figure~\ref{fig:simu_LOO_onemodel} shows (a) the tight FDR control  of all  methods and (b) the superior power of \texttt{OptCS-Full\_os} and \texttt{OptCS-Full\_sep} over \basename baselines with the same amount of labeled data. The most powerful sample split scheme in \basename varies, hence it is hard to determine the split ratio  in practice. By contrast, utilizing all labeled data in training, \texttt{OptCS-Full} consistently outperforms. 

\rrevise{
\noindent\textbf{Additional studies.} In Appendix~\ref{app:comp_other_full}, we evaluate OptCS-Full with other full-data training strategies, which shows the efficacy of LOO training and oversampling. In Appendix~\ref{app:subsec_null_impute_experiment}, we systematically evaluate the impact of all-test null imputation in OptCS-Full by varying the test sample size  and fraction of non-nulls; our results show that OptCS-Full performs well under various configurations, and OptCS-Full-sep is a  powerful option when signal dilution due to null imputation is a concern. In Appendix~\ref{app:subsec_exp_od}, OptCS-Full yields superior power over outlier detection and adaptive conformal selection methods, showing the benefit of full-data training in OptCS-Full.}

\subsection{Model training and selection with full data}
\label{subsec:simu_full_msel}

Finally, we   evaluate OptCS-Full-MSel, which uses all labeled data in  score training, model selection, and multiple testing. 
Competing methods include (i) \texttt{OptCS-Full-MSel}: OptCS-Full-MSel with over-sampling in training. (ii-iii) \texttt{Base\_random\_0.25/0.75}: Randomly split labeled data into  $\cD_{\train}$ (25\%/75\%) and $\cD_\calib$ (75\%/25\%), randomly select a model, train it on $\cD_\train$, and apply \basename with $\cD_\calib$. %
(iv-vii) \texttt{Base\_split\_112/\allowbreak 121/\allowbreak 211/\allowbreak 111}: Randomly split labeled data into $\cD_\train$, $\cD_\sel$, $\cD_\calib$ with ratio 1:1:2 (or 1:2:1, 2:1:1, 1:1:1), use  $\cD_\train$ to train all the models, $\cD_\sel$ to select the model via SCS, then run \basename with $\cD_\calib$ and the selected model. %
To study the contributions of   full-data training  and  model selection, we also evaluate two ``partial'' baselines: (a) \texttt{OptCS-Full\_random}: only full data training. (b) \texttt{OptCS-Full\_split}: OptCS-Full after sample splitting to select a model. 
Details are in Appendix~\ref{subsec:loo_setup}. %
Model details are in Appendix~\ref{app:subsec_simu_full_msel}.

\noindent\textbf{Simulation results.} The empirical FDR and power in Liang's settings are in Figure~\ref{fig:simu_FM_liang}; results for Jin's settings are in Figure~\ref{fig:simu_FM_jin}. 
OptCS-Full-MSel  controls the FDR, consistently outperforming
all  non-OptCS baselines  %
and the two partial baselines, and both modules contribute  to power improvement. 

\rrevise{
\noindent\textbf{Additional studies.} Similar to Section~\ref{subsec:simu_loo}, we include additional evaluations with (i) other training strategies in Appendix~\ref{app:comp_other_full}, (ii) the impact of all-test null imputation in Appendix~\ref{app:subsec_null_impute_experiment}, and (iii) other outlier detection and adaptive conformal selection methods in Appendix~\ref{app:subsec_exp_od}, with similar messages as those for OptCS-Full.}

\begin{figure} 
    \centering
    \includegraphics[width=0.9\linewidth]{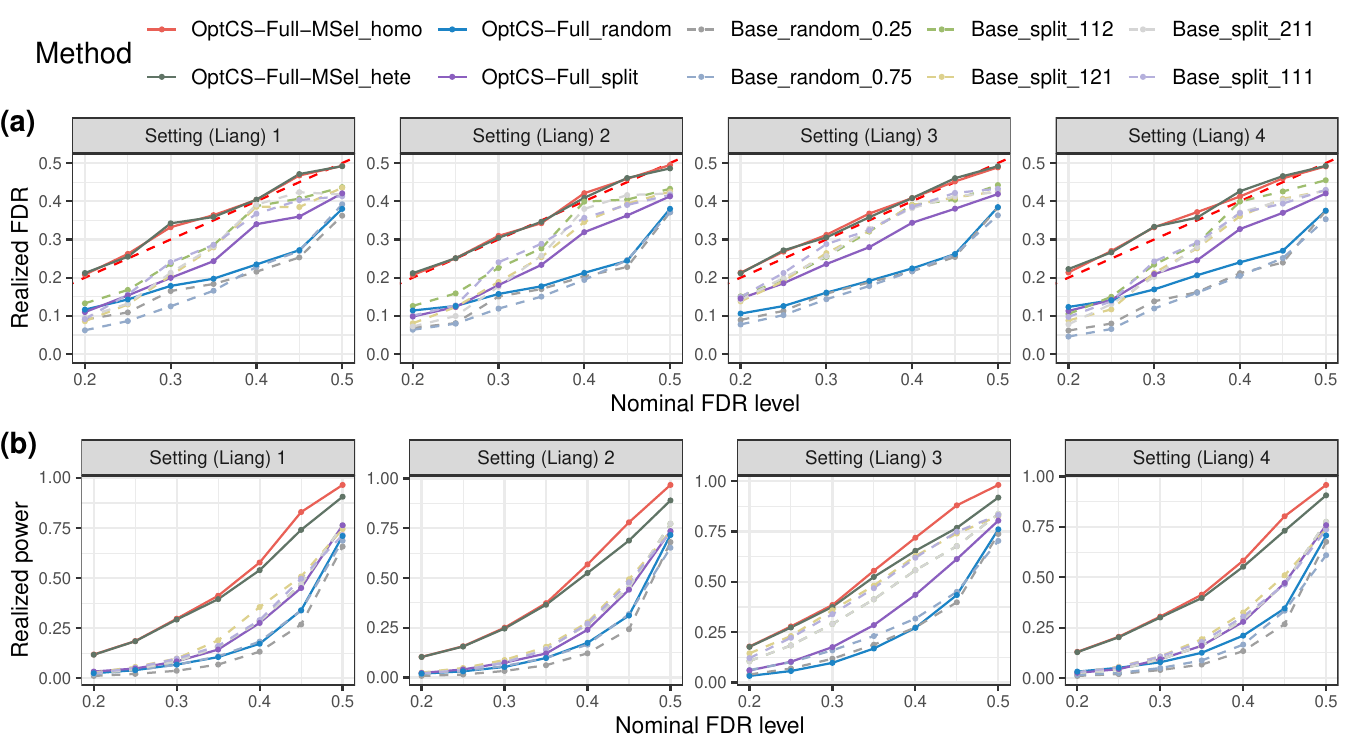}
    \caption{{\small Empirical FDR (a) and power (b) for OptCS-Full-MSel and baselines in Liang's settings with $n=500$ and $m=100$, averaged over $N=500$ runs.}}
    \label{fig:simu_FM_liang}
\end{figure}

\begin{figure}[htbp]
    \centering
    \includegraphics[width=0.9\linewidth]{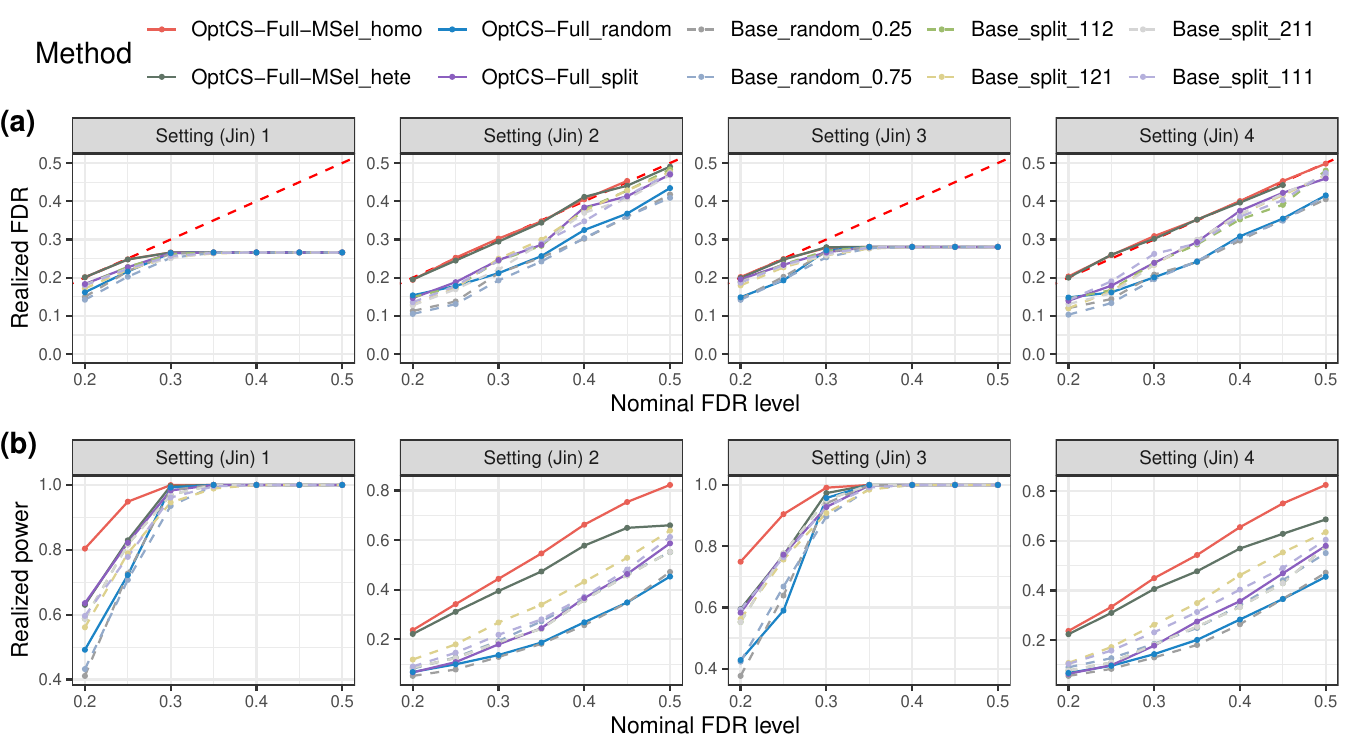}
    \caption{{\small Empirical FDR (a) and power (b) for OptCS-Full-MSel and baselines in Jin's settings with $n=500$ and $m=100$. Details are otherwise the same as Figure~\ref{fig:simu_FM_liang}.}}
    \label{fig:simu_FM_jin}
\end{figure}

% !TEX root = main.tex

\FloatBarrier
\section{Real data applications} \label{sec:real}

We apply OptCS to two representative applications: drug discovery (Section~\ref{subsec:app_drug}) and large language models in radiology report generation (Section~\ref{subsec:app_LLM}).  

\subsection{Drug discovery with model selection}
\label{subsec:app_drug}

We first apply OptCS to drug discovery tasks for selecting drug candidates with favorable biological properties. In this task, controlling the FDR of the selection set ensures that subsequent investigations, such as wet-lab validation of their properties, are resource-efficient. 

Concretely, in this problem, each sample is a drug candidate (such as a small molecule, an antibiotic, or a protein), whose physical and chemical features are encoded in $X\in \cX$, and we are interested in a certain biological property $Y\in \cY\subseteq \RR$. While the ground-truth knowledge of $Y$ typically needs to be evaluated via processes such as high-throughput screening (HTS) or wet-lab experiments~\citep{htsarticle,macarron2011impact}, machine learning models are increasingly used to predict the properties to identify potentially viable candidates before such costly evaluations~\citep{huang2007drug}. 
Given test drugs $\{X_{n+j}\}_{j=1}^m$, we aim to select promising ones with $Y_{n+j}>c_{n+j}$ for pre-determined thresholds $c_{n+j}$ while controlling the FDR. 

For drug property prediction, state-of-the-art models provide a wide variety of pretrained molecule embeddings~\citep{xiong2019pushing,dgllife,Landrum2016RDKit2016_09_4}. These embeddings can be leveraged to train a predictor (e.g., a shallow neural network) on local datasets for a specific downstream task. 
With the many choices, it is desirable to build the selection set with the best model. 
Among the three procedures, OptCS-MSel is the most suitable for this setting, where the models are typically costly to train. 

We apply OptCS-MSel to five drug discovery datasets~\citep{huang2021therapeutics}: \revise{\texttt{caco2} (902 drug candidates in total), \texttt{hepat} (1,020), \texttt{micro} (1,102), \texttt{vdss} (1,130), and \texttt{davis} (17,180)}, each with up to 17 deep learning models. See Appendix~\ref{app:subsec_drug_model} for the details on models and choices of $c_{n+j}$. 

Figure~\ref{fig:real_drug_class} presents the performance of five methods in Section~\ref{subsec:simu_score_sel} (except \texttt{Base\_tr\_split} which is infeasible for pre-trained models). While \texttt{Greedy} drastically violates the FDR, OptCS-MSel with homogeneous pruning outperforms competing FDR-controlling methods,  effectively selecting  models to improve selection power with rigorous FDR control.

\begin{figure} 
    \centering
    \captionsetup{font=small}
    \includegraphics[width=0.8\linewidth]{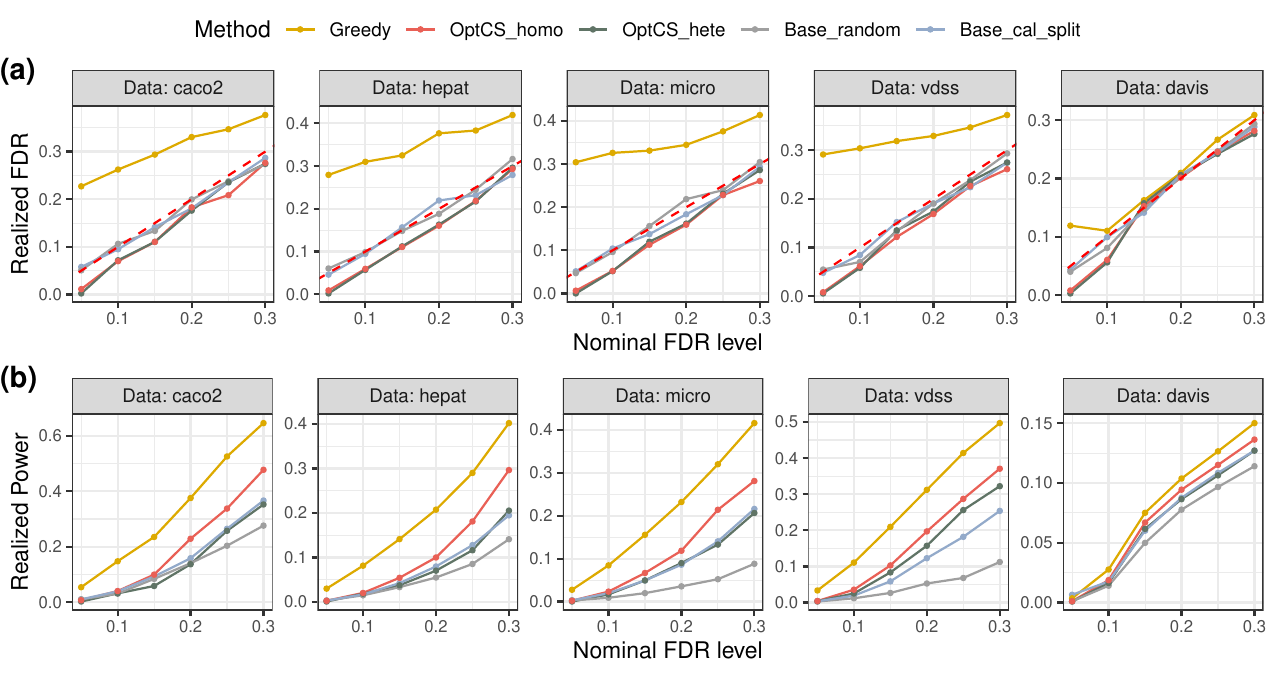}
    \caption{{\small Empirical FDR (a) and power (b) for five drug discovery tasks over $N=500$ runs.}}
    \label{fig:real_drug_class}
\end{figure}

\subsection{Boosting LLM Alignment}
\label{subsec:app_LLM}

Finally, we apply OptCS to aligning large language models for radiology report generation, identifying radiology images for which the LLM-generated reports meet a certain alignment criterion, where FDR control implies that the fraction of selected high-quality outputs  is sufficiently high. 

Concretely, each sample is a ``prompt'' $X\in \cX$, e.g., a radiology image. A vision-to-language model $f\colon \cX\to \cL$ generates a report summarizing the findings from the image, where $\cL$ is the space of reports. To address potential factual errors and other biases in $f(X)$, Conformal Alignment~\citep{gui2024conformal} defines the alignment status via an indicator $Y = \cA(f(X),E)\in \{0,1\}$, where $E\in \cE$ is gold-standard reference information (e.g.~a report written by human experts), and $\cA\colon \cL\times\cE\to \{0,1\}$ is a user-specified alignment function. With labeled data $\{(X_i,E_i)\}_{i=1}^n$, Conformal Alignment follows \basename to split the data into two folds. With one fold, it trains a predictor $g\colon \cX\to [0,1]$ for the alignment indicator. The second fold is used as calibration data in \basename to select new images $\{X_{n+j}\}_{j=1}^m$ with ``aligned'' reports ($Y_{n+j}>0$ should a reference report be acquired).
Given the scarcity of high-quality reference data, effective use of limited labeled data for model optimization is crucial. In this part, we focus on full-sample training variants which are particularly suitable for the lightweight training processes often used in Conformal Alignment.

We use a subset of the MIMIC-CXR dataset~\citep{johnson2019mimic}, and the vision-to-language model follows fine-tuning in~\cite{gui2024conformal}; see Appendix~\ref{app:subsec_data_LLM} for details. The features used are detailed in Appendix~\ref{app:subsec_detail_LLM}. 

\rrrevise{We note that in this application, OptCS does not refit the LLM; all LLM outputs and uncertainty features are kept fixed, and fitting is only applied to lightweight post-hoc predictors. In our implementation, with
\((n_2,m)=(200,500)\), OptCS-Full completes within a few minutes using depth-5 random
forests and in about 30--40 seconds using SVMs on standard CPUs with
\texttt{scikit-learn}.}

\noindent\textbf{Full-sample training with OptCS-Full.} 
We first apply OptCS-Full to utilize all labeled data in training the alignment predictor given a model class. 
We train three classifiers (logistic regression, random forests, and XGBoost), wrapped by the same conformity score $v(x,y) = My - g(x)$ for $M=1000$. 
We fix $n_\text{labeled}=400$ and $m=100$.
We compare OptCS-Full with \basename (the method in~\cite{gui2024conformal}) with the recommended $|\cD_\train| = 250$ and $|\cD_\calib| = 150$. The results are in Figure~\ref{fig:real_llm_full}. For each model class, the power improvement from full-sample training is significant, suggesting the utility of OptCS in settings with limited labeled data. 

\noindent\textbf{Combining model selection and full-sample training.} 
We then apply OptCS-Full-MSel to maximize power improvement with many modeling choices.    
We design two setups with different model choices; details are in Appendix~\ref{app:subsec_detail_LLM}.
We fix $n_{\text{labeled}} = 400$ and $m=100$ with the same data split as before. 
We compare OptCS-Full-MSel with sample splitting baselines similar to those in Section~\ref{subsec:simu_full_msel}, each with $N=500$ independent runs.
In Figure~\ref{fig:real_llm_full_msel}, not surprisingly, \texttt{Greedy} drastically violates FDR in both setups. OptCS-Full-MSel with homogeneous pruning outperforms the split baselines in both setups, suggesting the power gain from both full-sample training and model selection. 
In the first setup where individual model power differs significantly, the benefit of model selection is significant. In setup 2 where the model  qualities are similar, heterogeneous pruning  is less powerful, yet  homogeneous pruning remains powerful. 
Finally, OptCS-Full-MSel is slightly less powerful than  OptCS-Full in Figure~\ref{fig:real_llm_full} because the model classes there are  near-optimal, thereby eliminating the difficulty of model selection. %

\begin{figure} 
    \centering
    \includegraphics[width=\linewidth]{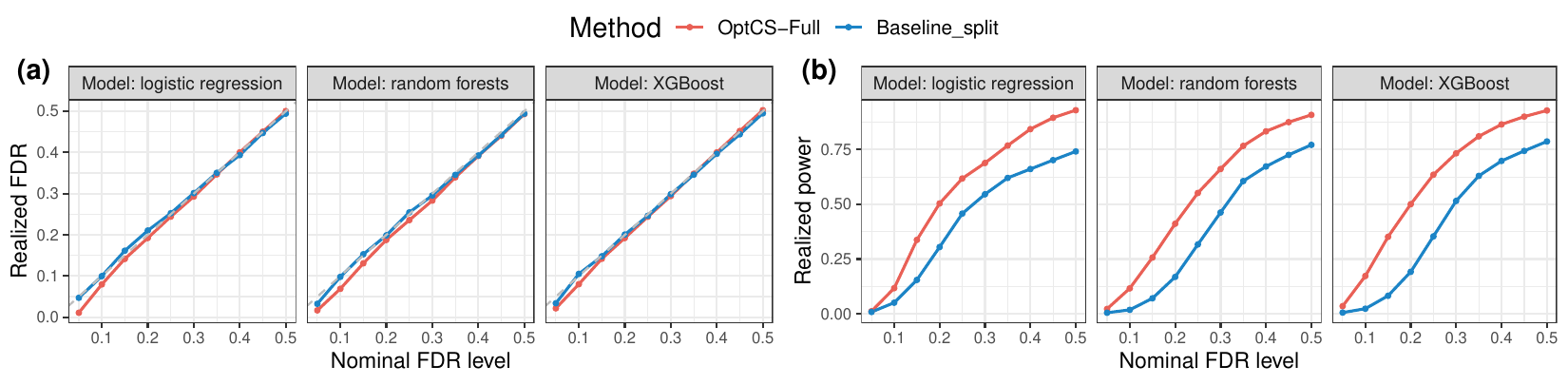}
    \caption{{\small Empirical FDR (a) and power (b) of OptCS-Full and \basename for aligning large language models with various nominal  levels, averaged over $N=500$ runs.}}
    \label{fig:real_llm_full}
\end{figure}

\begin{figure} 
    \centering
    \includegraphics[width=\linewidth]{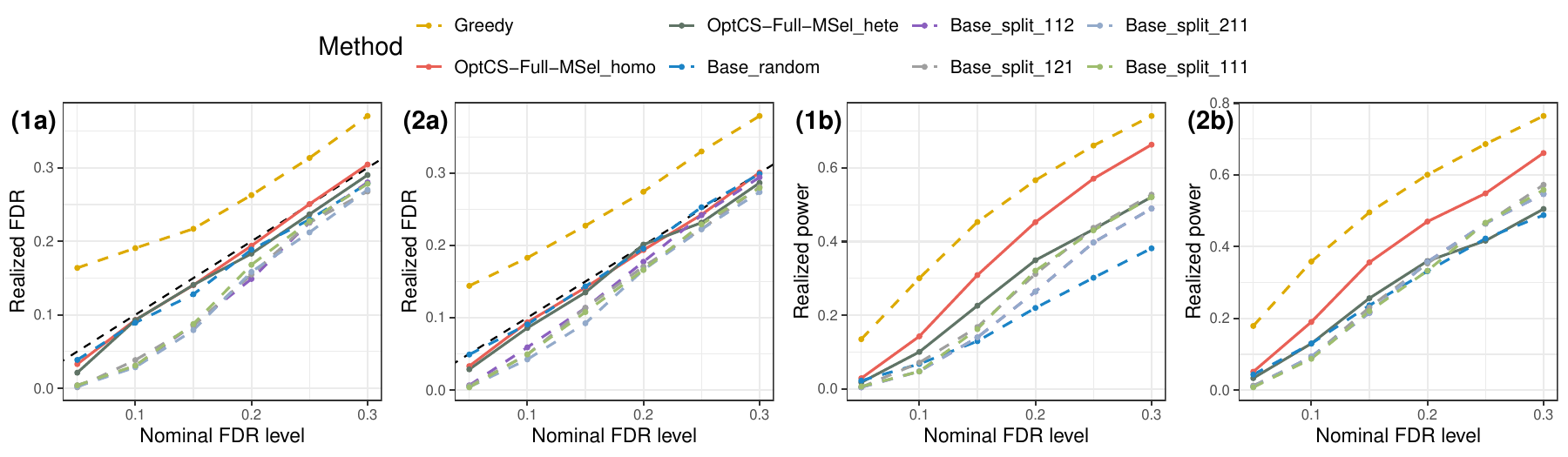}
    \caption{{\small Empirical FDR (setup 1 in (1a) and setup 2 in (2a)) and power (setup 1 in (1b) and setup 2 in (2b)) in LLM alignment under  various nominal FDR levels.}}
    \label{fig:real_llm_full_msel}
\end{figure}

\section*{Data availability}
The datasets used in the numerical experiments are publicly available from~\cite{huang2020deeppurpose} and the MIMIC-CXR project \url{https://physionet.org/content/mimic-cxr/2.0.0/} under the PhysioNet Credentialed Health Data License 1.5.0. Reproducibility code for the simulation and real data experiments is available at \url{https://github.com/Tian-Bai/OptCS}.

\section*{Acknowledgments}
Y.J.~is partially supported by the National Science Foundation (NSF: \#2610282).

\bibliographystyle{plainnat}%
\bibliography{reference}

\begin{thebibliography}{47}
\providecommand{\natexlab}[1]{#1}
\providecommand{\url}[1]{\texttt{#1}}
\expandafter\ifx\csname urlstyle\endcsname\relax
  \providecommand{\doi}[1]{doi: #1}\else
  \providecommand{\doi}{doi: \begingroup \urlstyle{rm}\Url}\fi

\bibitem[Bai et~al.(2025)Bai, Tang, Xu, Svetnik, Yang, Khalili, Yu, and
  Yang]{bai2024conformal}
Tian Bai, Peng Tang, Yuting Xu, Vladimir Svetnik, Bingjia Yang, Abbas Khalili,
  Xiang Yu, and Archer~Y. Yang.
\newblock Conformal selection for efficient and accurate compound screening in
  drug discovery.
\newblock \emph{Journal of Chemical Information and Modeling}, 65\penalty0
  (24):\penalty0 13070--13085, 2025.
\newblock \doi{10.1021/acs.jcim.5c01633}.

\bibitem[Bates et~al.(2023)Bates, Cand{\`e}s, Lei, Romano, and
  Sesia]{bates2023testing}
Stephen Bates, Emmanuel Cand{\`e}s, Lihua Lei, Yaniv Romano, and Matteo Sesia.
\newblock Testing for outliers with conformal p-values.
\newblock \emph{The Annals of Statistics}, 51\penalty0 (1):\penalty0 149--178,
  2023.

\bibitem[Benjamini and Hochberg(1995)]{benjamini1995controlling}
Yoav Benjamini and Yosef Hochberg.
\newblock Controlling the false discovery rate: a practical and powerful
  approach to multiple testing.
\newblock \emph{Journal of the Royal statistical society: series B
  (Methodological)}, 57\penalty0 (1):\penalty0 289--300, 1995.

\bibitem[Benjamini and Yekutieli(2001)]{benjamini2001control}
Yoav Benjamini and Daniel Yekutieli.
\newblock {The control of the false discovery rate in multiple testing under
  dependency}.
\newblock \emph{The Annals of Statistics}, 29\penalty0 (4):\penalty0 1165 --
  1188, 2001.

\bibitem[Carracedo-Reboredo et~al.(2021)Carracedo-Reboredo, Li{\~n}ares-Blanco,
  Rodr{\'\i}guez-Fern{\'a}ndez, Cedr{\'o}n, Novoa, Carballal, Maojo, Pazos, and
  Fernandez-Lozano]{carracedo2021review}
Paula Carracedo-Reboredo, Jose Li{\~n}ares-Blanco, Nereida
  Rodr{\'\i}guez-Fern{\'a}ndez, Francisco Cedr{\'o}n, Francisco~J Novoa, Adrian
  Carballal, Victor Maojo, Alejandro Pazos, and Carlos Fernandez-Lozano.
\newblock A review on machine learning approaches and trends in drug discovery.
\newblock \emph{Computational and structural biotechnology journal},
  19:\penalty0 4538--4558, 2021.

\bibitem[Chernozhukov et~al.(2021)Chernozhukov, W{\"u}thrich, and
  Zhu]{chernozhukov2021distributional}
Victor Chernozhukov, Kaspar W{\"u}thrich, and Yinchu Zhu.
\newblock Distributional conformal prediction.
\newblock \emph{Proceedings of the National Academy of Sciences}, 118\penalty0
  (48):\penalty0 e2107794118, 2021.

\bibitem[Fithian and Lei(2022)]{fithian2022conditional}
William Fithian and Lihua Lei.
\newblock Conditional calibration for false discovery rate control under
  dependence.
\newblock \emph{The Annals of Statistics}, 50\penalty0 (6):\penalty0
  3091--3118, 2022.

\bibitem[Gazin et~al.(2025)Gazin, Heller, Marandon, and
  Roquain]{gazin2025selecting}
Ulysse Gazin, Ruth Heller, Ariane Marandon, and Etienne Roquain.
\newblock Selecting informative conformal prediction sets with false coverage
  rate control.
\newblock \emph{Journal of the Royal Statistical Society Series B: Statistical
  Methodology}, page qkae120, 2025.

\bibitem[Gui et~al.(2024)Gui, Jin, and Ren]{gui2024conformal}
Yu~Gui, Ying Jin, and Zhimei Ren.
\newblock {Conformal Alignment}: Knowing when to trust foundation models with
  guarantees.
\newblock In \emph{Advances in Neural Information Processing Systems},
  volume~37, 2024.

\bibitem[Gui et~al.(2025)Gui, Jin, Nair, and Ren]{gui2025acs}
Yu~Gui, Ying Jin, Yash Nair, and Zhimei Ren.
\newblock {ACS}: An interactive framework for conformal selection.
\newblock \emph{arXiv preprint arXiv:2507.15825}, 2025.

\bibitem[He et~al.(2020)He, Liu, Gao, and Chen]{he2020deberta}
Pengcheng He, Xiaodong Liu, Jianfeng Gao, and Weizhu Chen.
\newblock {DeBERTa}: Decoding-enhanced {BERT} with disentangled attention.
\newblock \emph{arXiv preprint arXiv:2006.03654}, 2020.

\bibitem[Hu and Lei(2024)]{hu2024two}
Xiaoyu Hu and Jing Lei.
\newblock A two-sample conditional distribution test using conformal prediction
  and weighted rank sum.
\newblock \emph{Journal of the American Statistical Association}, 119\penalty0
  (546):\penalty0 1136--1154, 2024.

\bibitem[Huang et~al.(2020)Huang, Fu, Glass, Zitnik, Xiao, and
  Sun]{huang2020deeppurpose}
Kexin Huang, Tianfan Fu, Lucas~M Glass, Marinka Zitnik, Cao Xiao, and Jimeng
  Sun.
\newblock {DeepPurpose}: A deep learning library for drug-target interaction
  prediction.
\newblock \emph{Bioinformatics}, 2020.

\bibitem[Huang et~al.(2021)Huang, Fu, Gao, Zhao, Roohani, Leskovec, Coley,
  Xiao, Sun, and Zitnik]{huang2021therapeutics}
Kexin Huang, Tianfan Fu, Wenhao Gao, Yue Zhao, Yusuf Roohani, Jure Leskovec,
  Connor~W Coley, Cao Xiao, Jimeng Sun, and Marinka Zitnik.
\newblock Therapeutics data commons: Machine learning datasets and tasks for
  drug discovery and development.
\newblock \emph{arXiv preprint arXiv:2102.09548}, 2021.

\bibitem[Huang et~al.(2023{\natexlab{a}})Huang, Jin, Candes, and
  Leskovec]{huang2024uncertainty}
Kexin Huang, Ying Jin, Emmanuel Candes, and Jure Leskovec.
\newblock Uncertainty quantification over graph with conformalized graph neural
  networks.
\newblock \emph{Advances in Neural Information Processing Systems}, 36,
  2023{\natexlab{a}}.

\bibitem[Huang et~al.(2023{\natexlab{b}})Huang, Yu, Ma, Zhong, Feng, Wang,
  Chen, Peng, Feng, Qin, et~al.]{huang2023survey}
Lei Huang, Weijiang Yu, Weitao Ma, Weihong Zhong, Zhangyin Feng, Haotian Wang,
  Qianglong Chen, Weihua Peng, Xiaocheng Feng, Bing Qin, et~al.
\newblock A survey on hallucination in large language models: Principles,
  taxonomy, challenges, and open questions.
\newblock \emph{arXiv preprint arXiv:2311.05232}, 2023{\natexlab{b}}.

\bibitem[Huang(2007)]{huang2007drug}
Ziwei Huang.
\newblock \emph{Drug discovery research: new frontiers in the post-genomic
  era}.
\newblock John Wiley \& Sons, 2007.

\bibitem[Huo et~al.(2024)Huo, Lu, Ren, and Zou]{huoreal}
Yuyang Huo, Lin Lu, Haojie Ren, and Changliang Zou.
\newblock Real-time selection under general constraints via predictive
  inference.
\newblock In \emph{The Thirty-eighth Annual Conference on Neural Information
  Processing Systems}, 2024.

\bibitem[Jin and Cand{\`e}s(2023{\natexlab{a}})]{jin2023model}
Ying Jin and Emmanuel~J Cand{\`e}s.
\newblock Model-free selective inference under covariate shift via weighted
  conformal p-values.
\newblock \emph{arXiv preprint arXiv:2307.09291}, 2023{\natexlab{a}}.

\bibitem[Jin and Cand{\`e}s(2023{\natexlab{b}})]{jin2023selection}
Ying Jin and Emmanuel~J Cand{\`e}s.
\newblock Selection by prediction with conformal p-values.
\newblock \emph{Journal of Machine Learning Research}, 24\penalty0
  (244):\penalty0 1--41, 2023{\natexlab{b}}.

\bibitem[Johnson et~al.(2019)Johnson, Pollard, Berkowitz, Greenbaum, Lungren,
  Deng, Mark, and Horng]{johnson2019mimic}
Alistair~EW Johnson, Tom~J Pollard, Seth~J Berkowitz, Nathaniel~R Greenbaum,
  Matthew~P Lungren, Chih-ying Deng, Roger~G Mark, and Steven Horng.
\newblock {MIMIC-CXR}, a de-identified publicly available database of chest
  radiographs with free-text reports.
\newblock \emph{Scientific data}, 6\penalty0 (1):\penalty0 317, 2019.

\bibitem[Kuhn et~al.(2023)Kuhn, Gal, and Farquhar]{kuhn2023semantic}
Lorenz Kuhn, Yarin Gal, and Sebastian Farquhar.
\newblock Semantic uncertainty: Linguistic invariances for uncertainty
  estimation in natural language generation.
\newblock \emph{arXiv preprint arXiv:2302.09664}, 2023.

\bibitem[Landrum(2016)]{Landrum2016RDKit2016_09_4}
Greg Landrum.
\newblock {RDKit}: Open-source cheminformatics software, 2016.
\newblock URL
  \url{https://github.com/rdkit/rdkit/releases/tag/Release_2016_09_4}.

\bibitem[Lee et~al.(2025)Lee, Popov, and Ren]{lee2025full}
Junu Lee, Ilia Popov, and Zhimei Ren.
\newblock Full-conformal novelty detection: A powerful and non-random approach.
\newblock \emph{arXiv preprint arXiv:2501.02703}, 2025.

\bibitem[Lei et~al.(2018)Lei, G’Sell, Rinaldo, Tibshirani, and
  Wasserman]{lei2018distribution}
Jing Lei, Max G’Sell, Alessandro Rinaldo, Ryan~J Tibshirani, and Larry
  Wasserman.
\newblock Distribution-free predictive inference for regression.
\newblock \emph{Journal of the American Statistical Association}, 113\penalty0
  (523):\penalty0 1094--1111, 2018.

\bibitem[Li et~al.(2021)Li, Zhou, Hu, Fan, Zhang, Gu, and Karypis]{dgllife}
Mufei Li, Jinjing Zhou, Jiajing Hu, Wenxuan Fan, Yangkang Zhang, Yaxin Gu, and
  George Karypis.
\newblock {DGL-LifeSci}: An open-source toolkit for deep learning on graphs in
  life science.
\newblock \emph{ACS Omega}, 2021.

\bibitem[Liang et~al.(2024{\natexlab{a}})Liang, Zhu, and
  Barber]{liang2024conformal}
Ruiting Liang, Wanrong Zhu, and Rina~Foygel Barber.
\newblock Conformal prediction after efficiency-oriented model selection.
\newblock \emph{arXiv preprint arXiv:2408.07066}, 2024{\natexlab{a}}.

\bibitem[Liang et~al.(2024{\natexlab{b}})Liang, Sesia, and
  Sun]{liang2024integrative}
Ziyi Liang, Matteo Sesia, and Wenguang Sun.
\newblock Integrative conformal p-values for out-of-distribution testing with
  labelled outliers.
\newblock \emph{Journal of the Royal Statistical Society Series B: Statistical
  Methodology}, 86\penalty0 (3):\penalty0 671--693, 2024{\natexlab{b}}.

\bibitem[Lin et~al.(2023)Lin, Trivedi, and Sun]{lin2023generating}
Zhen Lin, Shubhendu Trivedi, and Jimeng Sun.
\newblock Generating with confidence: Uncertainty quantification for black-box
  large language models.
\newblock \emph{arXiv preprint arXiv:2305.19187}, 2023.

\bibitem[Lloyd()]{htsarticle}
Mathew Lloyd.
\newblock High-throughput screening as a method for discovering new drugs.
\newblock URL
  \url{https://www.drugtargetreview.com/article/61883/high-throughput-screening-as-a-method-for-discovering-new-drugs/}.

\bibitem[Macarron et~al.(2011)Macarron, Banks, Bojanic, Burns, Cirovic,
  Garyantes, Green, Hertzberg, Janzen, Paslay, et~al.]{macarron2011impact}
Ricardo Macarron, Martyn~N Banks, Dejan Bojanic, David~J Burns, Dragan~A
  Cirovic, Tina Garyantes, Darren~VS Green, Robert~P Hertzberg, William~P
  Janzen, Jeff~W Paslay, et~al.
\newblock Impact of high-throughput screening in biomedical research.
\newblock \emph{Nature reviews Drug discovery}, 10\penalty0 (3):\penalty0
  188--195, 2011.

\bibitem[Marandon et~al.(2024)Marandon, Lei, Mary, and
  Roquain]{marandon2024adaptive}
Ariane Marandon, Lihua Lei, David Mary, and Etienne Roquain.
\newblock Adaptive novelty detection with false discovery rate guarantee.
\newblock \emph{The Annals of Statistics}, 52\penalty0 (1):\penalty0 157--183,
  2024.

\bibitem[Mohammed et~al.(2020)Mohammed, Rawashdeh, and
  Abdullah]{mohammed2020machine}
Roweida Mohammed, Jumanah Rawashdeh, and Malak Abdullah.
\newblock Machine learning with oversampling and undersampling techniques:
  overview study and experimental results.
\newblock In \emph{2020 11th international conference on information and
  communication systems (ICICS)}, pages 243--248. IEEE, 2020.

\bibitem[Nair et~al.(2025)Nair, Jin, Yang, and Candes]{nair2025diversifying}
Yash Nair, Ying Jin, James Yang, and Emmanuel Candes.
\newblock Diversifying conformal selections.
\newblock \emph{arXiv preprint arXiv:2506.16229}, 2025.

\bibitem[Romano et~al.(2019)Romano, Patterson, and
  Candes]{romano2019conformalized}
Yaniv Romano, Evan Patterson, and Emmanuel Candes.
\newblock Conformalized quantile regression.
\newblock \emph{Advances in neural information processing systems}, 32, 2019.

\bibitem[Storey et~al.(2004)Storey, Taylor, and Siegmund]{storey2004strong}
John~D Storey, Jonathan~E Taylor, and David Siegmund.
\newblock Strong control, conservative point estimation and simultaneous
  conservative consistency of false discovery rates: a unified approach.
\newblock \emph{Journal of the Royal Statistical Society Series B: Statistical
  Methodology}, 66\penalty0 (1):\penalty0 187--205, 2004.

\bibitem[Stutz et~al.(2022)Stutz, Dvijotham, Cemgil, and
  Doucet]{stutz2021learning}
David Stutz, Krishnamurthy Dvijotham, Ali~Taylan Cemgil, and Arnaud Doucet.
\newblock Learning optimal conformal classifiers.
\newblock In \emph{International Conference on Learning Representations}, 2022.

\bibitem[Szyma{\'n}ski et~al.(2011)Szyma{\'n}ski, Markowicz, and
  Mikiciuk-Olasik]{szymanski2011adaptation}
Pawe{\l} Szyma{\'n}ski, Magdalena Markowicz, and El{\.z}bieta Mikiciuk-Olasik.
\newblock Adaptation of high-throughput screening in drug discovery
  toxicological screening.
\newblock \emph{International journal of molecular sciences}, 13\penalty0
  (1):\penalty0 427--452, 2011.

\bibitem[Vovk et~al.(2005)Vovk, Gammerman, and Shafer]{vovk2005algorithmic}
Vladimir Vovk, Alexander Gammerman, and Glenn Shafer.
\newblock \emph{Algorithmic learning in a random world}.
\newblock Springer, 2005.

\bibitem[Wang et~al.(2024)Wang, Huo, Peng, and Zou]{wangconformalized}
Xiaoning Wang, Yuyang Huo, Liuhua Peng, and Changliang Zou.
\newblock Conformalized multiple testing after data-dependent selection.
\newblock In \emph{Thirty-eighth Annual Conference on Neural Information
  Processing Systems}, 2024.

\bibitem[Weidinger et~al.(2021)Weidinger, Mellor, Rauh, Griffin, Uesato, Huang,
  Cheng, Glaese, Balle, Kasirzadeh, et~al.]{weidinger2021ethical}
Laura Weidinger, John Mellor, Maribeth Rauh, Conor Griffin, Jonathan Uesato,
  Po-Sen Huang, Myra Cheng, Mia Glaese, Borja Balle, Atoosa Kasirzadeh, et~al.
\newblock Ethical and social risks of harm from language models.
\newblock \emph{arXiv preprint arXiv:2112.04359}, 2021.

\bibitem[Wu et~al.(2023)Wu, Huo, Ren, and Zou]{wu2023optimal}
Xiaoyang Wu, Yuyang Huo, Haojie Ren, and Changliang Zou.
\newblock Optimal subsampling via predictive inference.
\newblock \emph{Journal of the American Statistical Association}, pages 1--13,
  2023.

\bibitem[Xie et~al.(2024)Xie, Barber, and Cand{\`e}s]{xie2024boosted}
Ran Xie, Rina~Foygel Barber, and Emmanuel~J Cand{\`e}s.
\newblock Boosted conformal prediction intervals.
\newblock In \emph{Advances in Neural Information Processing Systems},
  volume~37, 2024.

\bibitem[Xiong et~al.(2019)Xiong, Wang, Liu, Zhong, Wan, Li, Li, Luo, Chen,
  Jiang, et~al.]{xiong2019pushing}
Zhaoping Xiong, Dingyan Wang, Xiaohong Liu, Feisheng Zhong, Xiaozhe Wan, Xutong
  Li, Zhaojun Li, Xiaomin Luo, Kaixian Chen, Hualiang Jiang, et~al.
\newblock Pushing the boundaries of molecular representation for drug discovery
  with the graph attention mechanism.
\newblock \emph{Journal of medicinal chemistry}, 63\penalty0 (16):\penalty0
  8749--8760, 2019.

\bibitem[Xu and Ramdas(2024)]{xu2024online}
Ziyu Xu and Aaditya Ramdas.
\newblock Online multiple testing with e-values.
\newblock In \emph{International Conference on Artificial Intelligence and
  Statistics}, pages 3997--4005. PMLR, 2024.

\bibitem[Yang and Kuchibhotla(2024)]{yang2024selection}
Yachong Yang and Arun~Kumar Kuchibhotla.
\newblock Selection and aggregation of conformal prediction sets.
\newblock \emph{Journal of the American Statistical Association}, pages 1--13,
  2024.

\bibitem[Zhang et~al.(2022)Zhang, Jiang, Ren, Zou, and Dou]{zhang2022automs}
Yifan Zhang, Haiyan Jiang, Haojie Ren, Changliang Zou, and Dejing Dou.
\newblock {AutoMS}: automatic model selection for novelty detection with error
  rate control.
\newblock In \emph{Advances in Neural Information Processing Systems},
  volume~35, pages 19917--19929, 2022.

\end{thebibliography}

\newpage 
\appendix 
% !TEX root = main.tex

\section{Deferred discussion}
 
\subsection{Detailed related works}
\label{app:subsec_related_work}

\subsubsection{Connection to conditional calibration}
\label{app:subsub_cc}

A key technique for maintaining FDR control in our procedures is to use an individual threshold $s_j = q\hat{R}_j/m$ for each test p-value to prune the final selection set. 
This is largely inspired by the conditional calibration method of~\cite{fithian2022conditional}, whose technique is also extended in~\cite{jin2023model} to handle unknown p-value dependency in conformal selection under covariate shift. We discuss the connections and distinctions here.

\citet{fithian2022conditional} study the classical multiple testing setting where $m$ deterministic hypotheses $\{H_j\}_{j=1}^m$ are tested with dependent p-values $\{p_j\}_{j=1}^m$. To achieve FDR control, they assume there exist sufficient statistics $\{S_j\}_{j=1}^m$ such that for each $j$, the distribution of data conditional on $S_j$ is known. They then use Monte Carlo to find individual thresholds $\hat\tau_j$ and estimator $\hat{R}_j$ of the BH selection set, such that $\EE[\ind\{p_j \leq \hat\tau_j\}/\hat{R}_j]\leq q/m$. Finally, they use the same pruning method as ours with heterogeneous $\{\xi_j\}$ with $\hat\tau_j$ in the place of our $s_j$. 
Our method can be viewed as a generalization of this approach, where $\hat{R}_j$ is chosen as a permutation invariant function of the calibration and the $j$-th test data, and we take $\hat\tau_j = q \hat{R}_j /m$ without resorting to Monte Carlo simulation. This is because we leverage the conditional validity of the p-value to theoretically ensure that $\EE[\ind\{p_j\leq q\hat{R}_j/m\}/\hat{R}_j] \leq q/m$ (omitting the dependence on the null event here for the clarity of comparison). As such, the problem structure leveraged in our procedures is quite different from theirs. 
Apart from this key distinction, we also differ in terms of the hypothesis testing setting (since we address random hypotheses) and pruning methods (since we also introduce deterministic and homogeneous pruning). 

Such a generalization  is also used in \cite{jin2023model} who calibrates $s_j = q\hat{R}_j/m$ to address dependent weighted conformal p-values under covariate shift. However, the strategy there relies on a distinct conditional independence and permutation invariance structure than ours. Our setting requires delicate design of $\hat{R}_j$ and $p_j$ to accommodate complicated dependence due to data re-use and model selection.

\subsubsection{Connection to adaptive outlier detection}
\label{app:subsub_adadetect}

The general ideas of reusing data while maintaining inference validity for outlier detection are connected to several recent works, which we discuss in more detail here. 

\cite{zhang2022automs} propose AutoMS for automatic model selection in novelty detection with FDR control. Given a reference sample of inliers, AutoMS uses leave-one-out Jackknife scores to estimate the null score distribution for each candidate detector, constructs an FDR-based detection threshold, and selects the detector yielding the largest number of discoveries. This is conceptually related to OptCS-MSel and OptCS-Full-MSel, which likewise select among candidate scores according to downstream selection performance while allowing substantial data reuse. However, the validity theory of AutoMS relies on learning stability and uniform convergence of the Jackknife estimate, yielding a non-asymptotic high-probability bound on the FDP and asymptotic FDR control under density, stability, and rate conditions.

\cite{marandon2024adaptive} proposes AdaDetect which uses calibration and test data for better outlier detection model training. In the context of outlier detection where the calibration data consists of inliers and the test data is a mixture of inliers and outliers with unknown distributions, AdaDetect splits a set of calibration data as a ``null'' class to create a two-class structure between this subfold and the test data for model training, while using the remaining calibration data for both model training and   conformal p-value construction. 
First, their training step is close to our OptCS-Full variant in that both include test samples for model training. However, since the problem settings differ, the strategies are quite distinct: we do not need to set aside a subset of calibration data and can use all of them in model training; we need to impute the unknown test labels without compromising validity, while the data are fully observed in outlier detection. 
Second, apart from OptCS-Full, the other two variants of OptCS also reuse calibration data  for model choice and p-value construction. In their work, the validity is guaranteed by ensuring (global) permutation invariance across the entire set of (second fold of) calibration and test data.  However, since we need to evaluate the performance of each model, the role of test and calibration data is intrinsically asymmetric, and it is difficult to maintain global permutation invariance. In contrast, we only require certain permutation invariance among the calibration data and the $j$-th test point, thereby allowing the use of flexible model selection and training processes. As such, our multiple testing step and theoretical analysis differ significantly from theirs. 
\rrevise{Finally, we remark that the alternative proof of FDR control for OptCS-Full (Appendix~\ref{app:subsec_alt_full_proof}) relies on a similar PRDS property of the conformal p-values here, similar to AdaDetect.}

Similar ideas of data splitting are discussed in \cite{gazin2025selecting} in the context of selecting informative conformal prediction sets, which involves both selection and prediction set construction afterwards. 
Since the problems addressed are different, our procedures differ naturally. 
Although their InfoSP procedure recovers (split) conformal selection in the selection step, their selective-conditional method (InfoSCOP) diverges significantly from OptCS. In addition, in the construction of p-values, InfoSCOP only uses the preselected calibration subset, whereas all of our procedures can leverage the entire calibration dataset (so all of them are used both in training and calibration). 

\rrevise{
Finally, \cite{liang2024integrative} also study model optimization and selection in outlier detection. Their method improves upon one-class classifiers as conformity scores in conformal outlier detection by leveraging labeled outliers. Their main method splits the inliers and outliers into two folds; the first training fold is used to train the models, and related to our methods, they use median difference to select the models in a way that is permutation invariant to the calibration data and each test point. The second calibration fold is used to construct two preliminary p-values, which are integrated into a single  p-value. In a similar spirit to ours, they use a conditional calibration strategy to ensure FDR control. 
Our work differs mainly in scope and in the type of data-dependent optimization being enabled. OptCS is developed for  conformal selection problems with potentially heterogeneous thresholds $c_{n+j}$ 
 and general score/model optimization steps (including full-sample training and/or model selection) given any score function, whereas the training strategies, such as the proposal of integrated p-values, are highly specialized to outlier detection settings. 
That said, the two works are conceptually related in the shared motivation of improving power by model selection with data reuse. Beyond sample splitting,~\cite{liang2024integrative} discuss a more expensive transductive/CV$+$-type variant to reduce the loss and instability introduced by
sample splitting, at the cost of additional model fits (which also renders conditional calibration prohibitively expensive). In our work, Appendix~\ref{app:subsec_discuss_loo}  discusses a more
computationally intensive variant of OptCS-Full similar in spirit; it is designed to mitigate the effect of training with many
imputed null test samples (our OptCS-Full and OptCS-Full-MSel reduce computation costs by simultaneous null imputation of all test points). Both developments illustrate that
pursuing stronger data reuse often requires leave-out-style constructions with additional
computation. In general, our problem setting, specific training and model selection strategies, and the technical framework (with general permutation-based conditions) are sharply distinct.}
 
\subsubsection{Connection to full-conformal novelty detection}
\label{app:fcnd_connection}

We discuss the connection between OptCS and the full-conformal novelty detection
(FCND) framework of~\cite{lee2025full}. FCND is an independent and concurrent work that
develops full-conformal e-values for novelty detection relative to a null/inlier reference
distribution, and then applies e-BH-type procedures for FDR control. 
If the binary reduction is applied, FCND becomes a close e-value-based comparator (similar to AdaDetect and IC methods). However,
the connection is limited to the binarized full-data-training branch of OptCS; the general OptCS
framework also covers model/score selection and score constructions based on regression,
quantile, or clipped scores without reducing to binary
classification.

One key difference between OptCS and FCND is in data usage. 
FCND  constructs full-conformal e-values relative to a \emph{null/inlier reference sample}. In contrast, OptCS-Full remains in the supervised 
setting: it uses both labeled positives and labeled negatives to train a supervised score, and
imputes test labels as null to preserve validity. Thus, the direct FCND adaptation uses labeled
nulls as a reference sample, while OptCS-Full exploits labeled positive information at the
cost of possible signal dilution from null-imputed test samples.

The two methods also use different validity architectures. FCND follows an e-value/e-BH
route, whereas OptCS follows a p-value/conditional-calibration route: the score-generating
map \(\mathcal V\) produces valid data-dependent conformal p-values, and the auxiliary map
\(\mathcal R\) calibrates the dependent p-values for FDR control. The default \(K=1\) FCND
construction and the non-LOO version of OptCS-Full are structurally analogous one-fit
procedures; both can be made more robust through LOO-style constructions at additional
computational cost. Hence, the similarities reflect a common computation--power tradeoff
in full-data conformal methods, while the data usage and evidence objects remain different.
Targeted numerical comparisons with the official \texttt{fcnd} implementation are reported
in Appendix~\ref{app:subsec_compare_fcnd_empirical}.

\subsection{Design of auxiliary selection function}
\label{app:subsec_discuss_R}

In this part, we formalize the discussion in Remark~\ref{rm:Rj_choice}. For any p-values $\{p_j\}$ and auxiliary selection sizes $\{\hat{R}_j\}$, define the selection set with heterogeneous, homogeneous, and deterministic pruning as  
    \$
    \cS_{\hete} = \big\{ j\colon  p_j \leq s_j, ~ \xi_j  \hat{R}_{j}   \leq  r_{\hete}^*  \big\}, \quad &  r_{\hete}^* := \max\Big\{ r\colon  \sum_{j=1}^m \ind {\{  p_j\leq s_j, \, \xi_j  \hat{R}_j \leq r \}}   \geq r  \Big\},  \\
    \cS_{\homo} = \big\{ j\colon  p_j \leq s_j, ~ \xi \hat{R}_{j}   \leq  r_{\homo}^*  \big\}, \quad &  r_{\homo}^* := \max\Big\{ r\colon  \sum_{j=1}^m \ind {\{  p_j\leq s_j, \, \xi  \hat{R}_j \leq r \}}   \geq r  \Big\}, \\
    \cS_{\dtm} = \big\{ j\colon  p_j \leq s_j, ~   \hat{R}_{j}   \leq  r_{\dtm}^*  \big\}, \quad &  r_{\dtm}^* := \max\Big\{ r\colon  \sum_{j=1}^m \ind {\{  p_j\leq s_j, \,   \hat{R}_j \leq r \}}   \geq r  \Big\}.
    \$
    Here, $s_j  =q\hat{R}_j/m$, and  $ \xi_j \iid \textnormal{Unif}([0,1])$ and $ \xi\sim \textnormal{Unif}([0,1])$ are all independent of the data. 
    Also, define the output of the BH procedure applied to $\{p_j\}$ at level $q\in (0,1)$:
    \$
\cS_{\bh} = \big\{ j\colon  p_j \leq  q r_{\bh}^*/m   \big\}, \quad &  r_{\bh}^* := \max\Big\{ r\colon  \sum_{j=1}^m \ind {\{  p_j\leq qr/m  \}}   \geq r  \Big\}.
    \$

The following proposition formalizes Remark~\ref{rm:Rj_choice}. 

\begin{prop}
\label{prop:Rj}
    For any values of p-values $\{p_j\}$ and any values $\{\hat{R}_j\}$, the following statements hold:
    \begin{enumerate}[label=(\roman*)]
        \item  $\cS_{\dtm}\subseteq \cS_\bh$.
        \item If $\hat{R}_j \equiv |\cS_{\bh}|$, then $\cS_\hete = \cS_\homo = \cS_\dtm = \cS_\bh$. 
        \item $\cS_{\dtm}\subseteq  \cS_\hete$ and $\cS_\dtm \subseteq \cS_\homo$. 
    \end{enumerate}
\end{prop}

\begin{proof}[Proof of Proposition~\ref{prop:Rj}.]
    Fix any values $\{\lambda_j\}\leq 1$. Given any $\{p_j\}$, $\{\hat{R}_j\}$, and $s_j = q\hat{R}_j/m$,  we denote 
    \$
    \cS_1(r;\vec{\lambda} ) &= \{j \in[m] \colon p_j\leq s_j, \lambda_j \hat{R}_j \leq r\} \\ &= \Big\{ j \in[m]\colon p_j \leq q \hat{R}_j/m~\text{and}  ~ \lambda_j \cdot q \hat{R}_j/m \leq qr/m\Big\},
    \$
    and 
    \$
    \cS_2(r) = \{ j \in[m]\colon p_j \leq qr/m\}.
    \$

For any fixed $\vec\lambda\in[0,1]^m$, both
$S_1(r;\vec\lambda)$ and $S_2(r)$ are increasing in $r$.
Define
\[
r_1^*(\vec\lambda)
=
\max\{r\in\{0,\ldots,m\}:|S_1(r;\vec\lambda)|\ge r\},
\qquad
r_2^*
=
\max\{r\in\{0,\ldots,m\}:|S_2(r)|\ge r\}.
\]
Under deterministic pruning, $\vec\lambda=\mathbf1_m$.
If $j\in S_1(r;\mathbf1_m)$, then
\[
p_j\le q\hat R_j/m\le qr/m,
\]
so $S_1(r;\mathbf1_m)\subseteq S_2(r)$ for every $r$.
Consequently, $r_1^*(\mathbf1_m)\le r_2^*$, and
\[
\cS_{\rm dtm}
=
S_1(r_1^*(\mathbf1_m);\mathbf1_m)
\subseteq S_1(r_2^*;\mathbf1_m)
\subseteq S_2(r_2^*)
=
\cS_{\rm BH}.
\]
This proves (i).
For later use, monotonicity and maximality also imply
\[
|S_1(r_1^*(\vec\lambda);\vec\lambda)|
=
r_1^*(\vec\lambda)
\qquad\text{for every }\vec\lambda\in[0,1]^m.
\]
Indeed, if the cardinality were strictly larger, that
cardinality would itself be a larger admissible value
in the definition of $r_1^*(\vec\lambda)$, a contradiction.

    We then prove (iii) using similar ideas. Consider any $\vec{\lambda}  \leq \vec{\eta} $ meaning that $\lambda_j\leq \eta_j$ for all $j\in[m]$. By the definition of $\cS_1(r;\vec{\lambda})$, we know $\cS_1(r;\vec{\eta}) \subseteq \cS_1(r;\vec{\lambda})$ for any $r\in \NN$. Further, by the definition of $r_1^*(\vec{\lambda})$, we have $r_1^*(\vec{\eta})\leq r_1^*(\vec{\lambda})$, and therefore
    \$
    \cS_1(r_1^*(\vec{\eta});\vec{\eta})\subseteq \cS_1(r_1^*(\vec{\eta});\vec{\lambda}) \subseteq \cS_1(r_1^*(\vec{\lambda});\vec{\lambda}).
    \$
    Thus, since $\vec{\xi} \leq \mathbf{1}_m$, we know $\cS_\hete = \cS_1(r_1^*(\vec{\xi});\vec{\xi}) \supseteq  \cS_1(r_1^*( \mathbf{1}_m);\mathbf{1}_m) = \cS_\dtm$.  
    Since $\xi \mathbf{1}_m \leq \mathbf{1}_m$, we know $\cS_\homo = \cS_1(r_1^*(\xi \mathbf{1}_m);\xi \mathbf{1}_m) \supseteq  \cS_1(r_1^*( \mathbf{1}_m);\mathbf{1}_m) = \cS_\dtm$.  

    Finally, we prove (ii). Note that $r_2^* \leq |\cS_2(r_2^*)|$ by the definition of $r_2^*$. On the other hand, it must hold that $r_2^* \geq |\cS_2(r_2^*)|$, as otherwise   $|\cS_2(r_2^*+1)| \geq |\cS_2(r_2^*)| \geq r_2^*+1$, contradicting the definition of $r_2^*$. These two facts imply that $r_2^* = |\cS_2(r_2^*)| = |\cS_\bh|$. For any $\{\lambda_j\}\leq \mathbf{1}_m$, setting $\hat{R}_j \equiv r_2^*$ implies that for any $r\geq r_2^*$, 
    \@\label{eq:S12_equiv}
    \cS_1(r;\vec{\lambda} ) &= \Big\{ j \in[m]\colon p_j \leq q r_2^*/m~\text{and}  ~ \lambda_j \cdot q r_2^*/m \leq qr/m\Big\} \nonumber \\ & = \Big\{ j \in[m]\colon p_j \leq q r_2^*/m\Big\} \nonumber \\& = \cS_2(r_2^*),
    \@
    while $\cS_1(r;\vec{\lambda})\subseteq \cS_2(r_2^*)$ for any $r< r_2^*$. This implies $r_1^*(\vec{\lambda}) = r_2^*$ for any $\vec{\lambda}\leq \mathbf{1}_m$. 
    Therefore, using~\eqref{eq:S12_equiv} and the definitions of the three selection sets, we have 
    $
    \cS_\homo = \cS_1(r_2^*, \xi\mathbf{1}_m) = \cS_2(r_2^*) = \cS_\bh
    $, $ 
    \cS_\hete = \cS_1(r_2^*, \vec{\xi}) = \cS_2(r_2^*)= \cS_\bh$, and $\cS_\dtm = \cS_1(r_2^*,\mathbf{1}_m ) = \cS_2(r_2^*)= \cS_\bh$. This completes the proof of (ii). 
\end{proof}

\subsection{An alternative proof for OptCS-Full}
\label{app:subsec_alt_full_proof}

We provide an alternative proof for the validity of OptCS-Full variant under slightly stronger assumptions, based on the PRDS property of conformal p-values~\citep{benjamini2001control,bates2023testing,jin2023selection}. 
All distributional assumptions below are
understood conditionally on the preparatory data $Z_{1:n_1}$,
if any. This property is characterized by the following result:

\begin{prop} \label{prop:optcs_full_prds}
    Assume that the calibration and test data $\{Z_i\}_{i=n_1+1}^{n+m}$ are i.i.d., and the data in $\{Z_i\}_{i=n_1+1}^{n} \cup \{\hat{Z}_{n+\ell}\}_{\ell \neq j} \cup \{Z_{n+j}\}$ are mutually independent for any $j \in [m]$.
    Let $p_1, \dots, p_m$ be the conformal p-values defined in~\eqref{form:conf_p}, with the \textnormal{OptCS-Full} scores defined in~\eqref{eq:def_full_score}. For each fixed $j$, let $\hat\mu_i^{*,j}$ be the
    oracle LOO model obtained from $\mathbf D_{-i}$ by replacing $\hat Z_{n+j}^{T}$ with $Z_{n+j}^{T}$ in the augmented training sample before leaving out
    observation $i$. Define the oracle scores 
    \[
    V_i^{*,j}=
    \begin{cases}
    v(Z_i\given\hat\mu_i^{*,j}),
    & i\in\{n_1+1,\ldots,n,n+j\},\\
    v(\hat Z_i\given\hat\mu_i^{*,j}),
    & i\in\{n+1,\ldots,n+m\}\setminus\{n+j\}.
    \end{cases}
    \]
Assume these oracle scores have no ties almost surely for every $j\in[m]$. Define the oracle conformal p-values
    \begin{align} \label{eq:oracle_p}
        p_\ell^{*,j}
        =
        \frac{
            1+\sum_{i=n_1+1}^{n}
            \ind\{V_i^{*,j}\le V_{n+\ell}^{*,j}\}
        }{n_2+1},
        \qquad \ell\in[m],
    \end{align}
    and write $p_j^*=p_j^{*,j}$.
    Then, conditionally on the preparatory data,
    the vector $(p_1^{*,j},\dots,p_m^{*,j})$ is PRDS on $p_j^*$, and $p_j^*$ is super-uniform, for every $j\in[m]$.
\end{prop}

Assuming the above proposition, the validity of OptCS-Full
can be directly established following the technique of
Theorem~2.3 in~\cite{jin2023selection}, since the output
of the pruning step is identical to that of directly
running BH on the p-values by Proposition~\ref{prop:Rj}.
Specifically, denote by $\cS$ the OptCS-Full output and
by $\cS_{j\to *}$ the BH selection set obtained from
$(p_1^{*,j},\dots,p_m^{*,j})$.
On $\{Y_{n+j}\le c_{n+j}\}$, null compatibility gives
$Z_{n+j}^{T}=\hat Z_{n+j}^{T}$, so all observed and
oracle fitted models coincide. By the clipped-score
property, the entire score vectors also coincide.
Consequently,
\[
    (p_1^{*,j},\dots,p_m^{*,j})=(p_1,\dots,p_m),
    \qquad
    \cS_{j\to *}=\cS
    \quad\text{on }\{Y_{n+j}\le c_{n+j}\}.
\]
It follows that
\begin{align*}
    \fdr
    &=
    \sum_{j=1}^m
    \EE\bigg[
        \frac{
            \ind\{p_j\le q|\cS|/m,\,
                  Y_{n+j}\le c_{n+j}\}
        }{|\cS|\vee1}
    \bigg]\\
    &=
    \sum_{j=1}^m
    \EE\bigg[
        \frac{
            \ind\{p_j^*\le q|\cS_{j\to *}|/m,\,
                  Y_{n+j}\le c_{n+j}\}
        }{|\cS_{j\to *}|\vee1}
    \bigg]\\
    &\le
    \sum_{j=1}^m
    \EE\bigg[
        \frac{
            \ind\{p_j^*\le q|\cS_{j\to *}|/m\}
        }{|\cS_{j\to *}|\vee1}
    \bigg]
    \le
    \sum_{j=1}^m\frac qm
    =q.
\end{align*}
The last inequality follows from the standard
BH bound under PRDS and super-uniformity
at coordinate $j$, applied conditionally on the
preparatory data and then averaged. 

We now return to the proof of Proposition~\ref{prop:optcs_full_prds}, which is based on Theorem 3.3 of \cite{marandon2024adaptive}.
\cite{marandon2024adaptive} considers an outlier detection setting where the calibration data are identically distributed from some distribution $\mathbb{P}$, and the goal is to identify test data drawn from different distributions. The score-exchangeability argument in that theorem yields the following coordinatewise result, with the score direction reversed to match our notation:

\begin{lemma} \label{lemma:outlier_prds}
    Let $W_{n_1+1},\dots,W_n\sim\mathbb P$ be the
    calibration data and $W_{n+1},\dots,W_{n+m}$
    be the test samples. For each $j\in[m]$, the null hypothesis is $H_j^*:W_{n+j}\sim\mathbb P$.
    We denote the sets of inliers and outliers by
    $\cH_0=\{j:W_{n+j}\sim\mathbb P\}$ and
    $\cH_1=\{j:W_{n+j}\not\sim\mathbb P\}$, respectively.
    Given any mechanism generating scores
    $\{V_i\}_{i=n_1+1}^{n+m}$ with no ties almost surely,
    construct
    \[
        \tilde p_\ell^*
        =
        \frac{
            1+\sum_{i=n_1+1}^{n}
            \ind\{V_i\le V_{n+\ell}\}
        }{n_2+1},
        \qquad \ell\in[m].
    \]
    If, for a given $j\in\cH_0$,
    $\{V_i\}_{i=n_1+1}^{n}\cup\{V_{n+j}\}$
    is exchangeable conditionally on
    $\{V_{n+\ell}\}_{\ell\ne j}$, then
    $(\tilde p_1^*,\dots,\tilde p_m^*)$ is PRDS on
    $\tilde p_j^*$, and $\tilde p_j^*$ is super-uniform.
    In particular, if this condition holds for every
    $j\in\cH_0$, the vector is PRDS on $\cH_0$.
\end{lemma}

Lemma~\ref{lemma:outlier_prds} allows any construction
of the scores satisfying the stated conditional
exchangeability and no-ties conditions.
We note that this is the case for the oracle OptCS-Full
scores. Fix $j\in[m]$ and instantiate the construction
with
\[
    W_i=Z_i\quad(i=n_1+1,\dots,n),\qquad
    W_{n+j}=Z_{n+j},\qquad
    W_{n+\ell}=\hat Z_{n+\ell}\quad(\ell\ne j).
\]
Here, $j$ is an inlier coordinate; no restriction on the
distributions of the remaining test observations is
needed for the coordinatewise conclusion.

By the assumptions, the calibration observations and
$W_{n+j}$ are exchangeable conditionally on the
remaining inputs. The fixed observationwise
transformation $T$ commutes with permutations, and
the symmetry of $\cA$ makes the corresponding LOO
scores permutation equivariant. Moreover, permuting
the calibration observations and $W_{n+j}$ leaves
the scores at all other test indices unchanged.
Thus,  $\{V_i^{*,j}\}_{i=n_1+1}^{n}\cup\{V_{n+j}^{*,j}\}$  is exchangeable conditionally on $\{V_{n+\ell}^{*,j}\}_{\ell\ne j}$. 
Applying Lemma~\ref{lemma:outlier_prds}, we obtain
    $(\tilde p_1^*,\dots,\tilde p_m^*) = (p_1^{*,j},\dots,p_m^{*,j})$, 
which is PRDS on
$\tilde p_j^*=p_j^{*,j}=p_j^*$, with $p_j^*$
super-uniform. This concludes the proof of
Proposition~\ref{prop:optcs_full_prds}.

\subsection{A second variant of OptCS-Full(-MSel)}
\label{app:subsec_discuss_loo}

In this part, we present two additional variants of OptCS that avoid using too many imputed null samples in the leave-one-out training idea of OptCS-Full and OptCS-Full-MSel. We call these variants OptCS-Full\_sep and OptCS-Full-MSel\_sep.

\paragraph{Method with fixed model class.} We begin by the setting of Section~\ref{subsec:loo} with a fixed model class. We assume a training process $\cA$, and write the conformity score function via  
$v(x,y \given \hat\mu)\colon \cX\times\cY\to \RR$, 
where the functional dependence on a trained model is encapsulated in $\hat\mu$ representing a general model.
We introduce the procedure of OptCS-Full\_sep by specifying the two subroutines in Lines~\ref{line:optcs_score} and~\ref{line:optcs_R} in Algorithm~\ref{algo:optcs}. 
First, for each $j\in [m]$, we define the  leave-one-out training set (again without splitting the labeled data)
\$
\mathbf{D}_{-\ell}^{(j)} = \{Z_i^T\}_{i=1}^n \cup\{\hat{Z}_{n+j}^T\} \backslash \{\tilde{Z}^T_\ell\},
\$
where $\tilde{Z}^T_\ell = Z^T_\ell$ if $\ell \leq n$ and $\tilde{Z}^T_{\ell} = \hat{Z}^T_{n+j}$ if $\ell = n+j$. 
Then, for each $\ell\in [n]\cup\{n+j\}$, we train a model  
$\hat\mu_{\ell}^{(j)} = \cA( \mathbf{D}_{-\ell}^{(j)})$ 
using the leave-one-out data. Note that each $\mathbf{D}_{-\ell}^{(j)}$ contains at most one imputed null data point. 
We then specify the score generating functional in Line~\ref{line:optcs_score} of Algorithm~\ref{algo:optcs} via 
\$
&\cV^{(j)} = (V_{n_1+1}^{(j)},\dots,V_n^{(j)}, \hat{V}_{n+1}^{(j)},\dots, \hat{V}_{n+m}^{(j)}), \\ 
&\textrm{where}\quad V_{i}^{(j)} = v(X_i,Y_i\given \hat\mu_{i}^{(j)}),\quad \hat{V}_{n+\ell}^{(j)} = v(X_{n+\ell},c_{n+\ell}\given \hat\mu_{n+j}^{(j)}).
\$
In our experiments, we find that running the BH procedure using p-values~\eqref{form:conf_p} based on the above scores always controls the empirical FDR. Therefore, we recommend using such a relaxed procedure for practical use. 
However, if one aims for rigorous theoretical control of FDR in finite samples, we specify the auxiliary selection function $\cR(\cdot)$ in Line~\ref{line:optcs_R} of Algorithm~\ref{algo:optcs} as follows. We define the auxiliary p-values 
\$
\tilde{p}_{\ell}^{(j)} = \frac{\sum_{i=n_1+1}^n \ind\{ V_{i}^{(j)}\leq v(X_{n+\ell},c_{n+\ell}\given \hat\mu_{i}^{(j)}) \} + \ind\{ \hat{V}_{n+j}^{(j)}\leq v(X_{n+\ell},c_{n+\ell}\given \hat\mu_{n+j}^{(j)}) \}}{n_2+1},
\$
and set $\hat{R}_j$ as the size of  the BH procedure output applied to $\{\tilde{p}_{\ell}^{(j)}\}_{\ell \neq j}\cup \{0\}$. Similar to the proof ideas of Proposition~\ref{lemma:loo_full_valid}, we can show that these definitions obey the properties of $\cV(\cdot)$ and $\cR(\cdot)$ required for finite-sample FDR control in Theorem~\ref{thm:optcs_fdr_control}. 

\paragraph{Method with model selection.} The variant with model selection follows a similar idea as above. Following Section~\ref{subsec:loo_sel}, we consider $K$ training processes $\{\cA_k\}_{k=1}^K$ and   $K$ conformity scores whose dependence on a trained model is represented by 
$v(\cdot,\cdot\given \hat\mu_k,k)$. 
With the same idea as before, we obtain the leave-one-out model (using the $T_k$ encoding)
\$
\hat\mu_{\ell,k}^{(j)} = \cA_k(\mathbf{D}_{-\ell}^{(j,k)}),\quad \ell \in [n]\cup\{n+j\}.
\$
Then, we use ideas similar to OptCS-MSel to select an optimal model for the $j$-th test point in a permutation-invariant fashion. Specifically, we define the scores 
$V_{i,k}^{(j)} = v(X_i,Y_i\given \hat\mu_{i,k}^{(j)},k)$ for $i\in[n_1+1:n]$, and  $\hat{V}_{n+\ell,i,k}^{(j)}=v(X_{n+\ell},c_{n+\ell}\given \hat\mu_{i,k}^{(j)},k)$ for $\ell\in [m]$ and $i\in [n_1+1:n]\cup\{n+j\}$. We define the auxiliary p-values 
\$
\tilde{p}_{\ell,k}^{(j)} = \frac{\ind\{\hat{V}_{n+j,n+j,k}^{(j)} \leq \hat{V}_{n+\ell,n+j,k}^{(j)}\} + \sum_{i=n_1+1}^{n} \ind\{V_{i,k}^{(j)} \leq \hat{V}_{n+\ell,i,k}^{(j)}\} }{1+n_2},
\$
and let $\tilde{S}_k^{(j)}$ be the selection set of the BH procedure applied to $\{0\}\cup\{\tilde{p}_{\ell,k}^{(j)}\}_{\ell\neq j}$. The selected model index is then 
$\hat{k}_j = \argmax_{k\in [K]}|\tilde{S}_k^{(j)}|$. 
Finally, we 
specify the score generating functional in Line~\ref{line:optcs_score} of Algorithm~\ref{algo:optcs} via 
\$
V_{i}^{(j)} = v(X_i,Y_i\given \hat\mu_{i,\hat{k}_j}^{(j)},\hat{k}_j),\quad \hat{V}_{n+\ell}^{(j)} = v(X_{n+\ell},c_{n+\ell}\given \hat\mu_{n+j,\hat{k}_j}^{(j)},\hat{k}_j),
\$
and the auxiliary selection sizes are given by $\hat{R}_j = |\tilde{S}_{\hat{k}_j}^{(j)}|$. Following similar ideas as before, one can show that the conditions needed in Theorem~\ref{thm:optcs_fdr_control} are fulfilled by this construction, and applying OptCS calibration yields valid FDR control. 

\paragraph{Intermediate solution: $K_0$-fold split.}
The above two methods are suitable in settings where practitioners are willing to commit intensive computation efforts to fully exploit the labeled data. Finally, we discuss an intermediate solution based on $K_0$-fold splitting that sits between the all-test null imputation in the main text and the intensive variant here. For brevity, we introduce the method with a fixed model class, following the setup in Section~\ref{subsec:loo}. 

This variant still preserves all the labeled data in training, but reduces the number of imputed null test samples compared with OptCS-Full. Let $K_0\in \NN^+$ be the number of folds (such as 5 or 10). We randomly split the \emph{test} data into $K_0$ folds, $I_1,\dots,I_{K_0}\subseteq [m]$. 
For each $j\in I_{k_0}$, $k_0\in [K_0]$, we define the leave-one-out datasets
\$
\mathbf{D}_{-\ell}^{(j)} = \{Z_i^T\}_{i=1}^n \cup \{\hat{Z}_{n+j'}^T\}_{j'\in I_{k_0}} \backslash \{\tilde{Z}_{ \ell}^T\},\quad \ell \in \{n_1+1,\dots,n\}\cup\{n+j'\}_{j'\in I_{k_0}}.
\$
The corresponding trained model is then 
\$
\hat\mu_{\ell}^{(j)} = \cA(\mathbf{D}_{-\ell}^{(j)}),\quad \ell \in \{n_1+1,\dots,n\}\cup\{n+j'\}_{j'\in I_{k_0}}.
\$
It incur $|I_{k_0}|+n_2$ times of model fitting for each $k_0$, totaling $m+n_2\cdot K_0$ times of model fitting. 
We  define the score-generating functional in Line~\ref{line:optcs_score} of Algorithm~\ref{algo:optcs} via 
\$
V_{i}^{(j)} = v(X_i,Y_i\given \hat\mu_{i }^{(j)} ),\quad \hat{V}_{n+\ell}^{(j)} = v(X_{n+\ell},c_{n+\ell}\given \hat\mu_{n+\ell}^{(j)} ),\quad \ell\in I_{k_0}.
\$
Note that   $V_i^{(j)}=V_i^{(j')}$ and  $\hat{V}_{n+\ell}^{(j)}=\hat{V}_{n+\ell}^{(j')}$ for any $j,j'\in I_{k_0}$ by definition. 
We define the p-values for $j\in I_{k_0}$ via 
\$
p_j = \frac{1+\sum_{i=n_1+1}^n\ind\{V_i^{(j)}\leq \hat{V}_{n+j}^{(j)}\}}{n_2+1}.
\$
It then suffices to set $\hat{R}_j$ as the size of the BH selection set applied to $\{p_j\}_{j\in I_{k_0}}$ at level $q$. One can show that this construction satisfies the conditions needed for Theorem~\ref{thm:optcs_fdr_control} (using exactly the same proof idea as that of OptCS-Full). 

\subsection{Conceptual justification of (binarized) null imputation}
\label{app:subsec_null_impute_discuss}

This appendix provides a theoretical intuition for why \emph{all-test null imputation} in OptCS-Full
need not hurt \emph{asymptotic} power, even when a non-negligible fraction of test points are truly non-null.
The key observation is that, under the standard setup of OptCS (i.e., calibration and test
covariates share the same marginal distribution), the optimal score under the ``mixed'' training
distribution induced by null imputation becomes a \emph{global rescaling} of the optimal score under the
correct data distribution. Since the p-values depend on the \emph{relative ordering}
(ranks) of scores rather than their absolute values, such a rescaling does not affect the limiting behavior.
We emphasize that this argument is asymptotic and abstracts away finite-sample training error; the latter is
addressed empirically in Appendix~\ref{app:subsec_null_impute_experiment}.

\subsubsection{Setup and the induced mixed distribution}

We present the argument under the binary reduction used in OptCS-Full.
Let $Y \in \{0,1\}$ denote the reduced label, where $Y=1$ corresponds to ``non-null'' (e.g., $Y>c$ in the original
problem) and $Y=0$ corresponds to ``null''. Let $P$ denote the joint distribution of a labeled pair $(X,Y)$, and
write the true conditional probability 
\[
p(x) := P(Y=1 \mid X=x).
\]
OptCS assumes that calibration and test covariates share the same marginal distribution.

OptCS-Full (with all-test null imputation) trains its predictive model on the augmented dataset consisting of
(i) labeled samples $(X,Y)\sim P$ and (ii) imputed test samples $(X,0)$ whose covariates follow the same marginal
$P_X$. At the population level, this corresponds to a \emph{mixture} distribution $\widetilde{P}$ over $(X,Y)$:
\begin{equation}
\widetilde{P}
\;:=\;
\frac{n}{n+m}\,P \;+\; \frac{m}{n+m}\,(P_X \otimes \delta_0),
\label{eq:mix-dist}
\end{equation}
where $n$ is the number of labeled samples used for training and $m$ is the number of test samples imputed as
null. (Under LOO fitting, leaving out one point only changes the mixture weights by $o(1)$ and does not affect
the asymptotic discussion.)

\subsubsection{Optimal score under $\widetilde{P}$ is a global rescaling}

According to~\cite{jin2023selection}, the ``oracle'' score under the true distribution $P$ that leads to the most powerful selection is based on  $p(x)=P(Y=1\given X=x)$, namely, the clipped score $V(x,y) = My-p(x)$. 
A direct calculation shows that under $\widetilde{P}$, the oracle score is a
\emph{global scaling} of $p(x)$.
Indeed, the $X$-marginal under $\widetilde{P}$ is still $P_X$, since both components in
\eqref{eq:mix-dist} share the same $X$-marginal. 
Moreover, the joint mass of $(X,Y=1)$ under $\widetilde{P}$ only comes from the labeled component: 
\begin{equation}
\widetilde{p}(x)
\;:=\;
\widetilde{P}(Y=1 \mid X=x)
\;=\;
\frac{n}{n+m}\, P(Y=1 \mid X=x)
\;=\;
\frac{n}{n+m}\, p(x).
\label{eq:p-tilde-rescale}
\end{equation}
Consequently, the oracle score under the mixed distribution is simply a constant multiple of
the oracle score under the correct labeled distribution.
In particular, the \emph{ordering} of the score function is preserved. 

\subsubsection{Unchanged asymptotic behavior for conformal p-values}

Conformal selection and the OptCS p-values depend on score comparisons of the form
$\mathbb{I}\{V_i^{(j)} \le V_{n+j}^{(j)}\}$ (see~\eqref{form:conf_p} and the definitions in Section~\ref{sec:method}).
Thus, the \emph{rank} of the test score among calibration scores is the key object.

To make the order-preserving argument transparent, consider the clipped score commonly used for the binary
reduction~\citep{jin2023selection}
\[
v(x,y \mid \mu) := M y - \mu(x),
\]
with $M$ chosen large enough so that $v(x,1\mid \mu) > v(x',0\mid \mu)$ for all $x,x'$ (this ensures that only
null calibration points contribute to the relevant comparisons).
For a null-labeled point ($y=0$), the score reduces to $v(x,0\mid \mu)=-\mu(x)$, so the ordering of scores among
null points is exactly the reverse ordering of $\mu(x)$.

\paragraph{Oracle-level statement.}
If we could plug in the oracle score, then by \eqref{eq:p-tilde-rescale},
\[
v(x,0 \mid \widetilde{p}) = -\widetilde{p}(x) = -\frac{n}{n+m}p(x),
\]
which differs from $-p(x)$ by a positive constant factor. Therefore,   the induced ranking (hence the resulting conformal p-values) is identical to that obtained from the oracle
score under $P$. 
Concretely, the conformal p-value for test point $j$ can be written as
\begin{equation}
p_j
=
\frac{1 + \sum_{i =n_1+1}^n\mathbb{I}\!\left\{ v(X_i,Y_i \mid \mu_i) \le v(X_{n+j},0\mid \mu_{n+j}) \right\}}
{1 + n_2},
\label{eq:pj-general}
\end{equation}
where $\mu_i$ denotes the fitted model used for scoring point $i$ (LOO fits in OptCS-Full). 
If $\mu_i$ were replaced by the
oracle $\widetilde{p}$ (or $p$), the indicator event depends only on the ordering of $\widetilde{p}(X_i)$ and
$\widetilde{p}(X_{n+j})$, hence equals the corresponding event with $p(X_i)$ and $p(X_{n+j})$. Thus, the oracle p-values coincide.

\paragraph{Plug-in  statement.}
Suppose the learning algorithm used in OptCS-Full is \emph{consistent} for the mixed distribution in the sense
that, under LOO fitting, $\mu_i(\cdot)$ converges (in an appropriate mode, e.g., in probability in $L_2(P_X)$ or
uniformly on compact sets) to $\widetilde{p}(\cdot)$.
Then the score $v(\cdot,0\mid \mu_i)$ converges to $v(\cdot,0\mid \widetilde{p})$. Since \eqref{eq:pj-general} depends only on ranks,
the resulting conformal p-values converge to the same limiting p-values as those obtained from the correct
labeled distribution. Consequently, with similar ideas as in the asymptotic analysis of, e.g., Appendix~\ref{app:subsec_proof_power_msel}, we can show the power converges to the asymptotic power. 

\subsubsection{Remarks and limitations}

\paragraph{Other score mappings.}
The argument above does not rely on the specific clipped score; any score construction in OptCS-Full that is
a fixed (strictly) monotone transform of a consistent estimate of $\widetilde{p}(x)$ yields the same order
preservation and hence the same asymptotic implication.

\paragraph{Class-imbalance corrections.}
Many standard imbalance corrections (e.g., constant class weights or random over/under-sampling) effectively
change class fractions while keeping in-class distributions unchanged. One can show that at the population level, the
resulting \emph{oracle score} remains a monotone function of $p(x)$, and thus also preserves ordering; this is
why such corrections do not alter the conceptual asymptotic picture, while potentially improving finite-sample
training.

\paragraph{Finite-sample caveat.}
While $\widetilde{p}(x)$ preserves ordering, the learning problem can become harder when $m$ is large due to
severe class imbalance and signal dilution, which may degrade finite-sample training quality. This motivates
our use of imbalance correction (e.g., over-sampling) and also motivates the per-test variant in Appendix~\ref{app:subsec_discuss_loo}.
We provide targeted simulations and practical guidance for this computation--power tradeoff in Appendix~\ref{app:subsec_null_impute_experiment}.

\section{Technical proofs}

\subsection{Proof of Theorem~\ref{thm:optcs_fdr_control}}

\begin{proof}[Proof of Theorem~\ref{thm:optcs_fdr_control}.]
\label{proof:optcs_fdr_control}
    Fix a $j \in [m]$ that corresponds to a null hypothesis, i.e. $Y_{n+j} \leq c_{n+j}$. When computing the scores as the output of $\cV$, consider substituting the $j$-th test sample $\hat{Z}_{n+j}$ with $Z_{n+j}$ which use the true response $Y_{n+j}$. We denote 
    \$
        &\big( \bar{V}^{(j)}_{n_1:n}, \bar{V}_{n+1:n+j-1}^{(j)}, \bar{V}_{n+j}^{(j)}, \bar{V}_{n+j+1:n+m}^{(j)} \big) \\ &:= \cV^{(j)}(Z_{1:n_1}, Z_{n_1+1:n}, \hat{Z}_{n+1:n+j-1}, Z_{n+j}, \hat{Z}_{n+j+1:n+m}).
    \$ 
    We will compare this new set of scores with~\eqref{eq:scores}. Note that these scores are unobservable, and are used only for the purpose of analysis.
    Since $\cV$ is monotone for the null (Definition~\ref{def:monotone}), on the null event that $Y_{n+j}\leq c_{n+j}$, we have
    \$
        \bar{V}^{(j)}_{i} = {V}^{(j)}_{i} \text{ for } i \leq n, \quad \bar{V}^{(j)}_{n+i} = \hat{V}^{(j)}_{n+i} \text{ for } i \leq m, i \neq j, \quad \text{and} \quad \bar{V}^{(j)}_{n+j} \leq \hat{V}^{(j)}_{n+j}.
    \$
    We now define the oracle p-value as
    \$
    p_j^* := \frac{1 + \sum_{i=n_1+1}^n \ind\{\bar{V}_i^{(j)} \leq \bar{V}^{(j)}_{n+j}\}}{n_2+1}.
    \$
    By definition, $p_j^* \leq p_j$ on the null event that $Y_{n+j}\leq c_{n+j}$. In the next, we seek to prove the subuniformity of $p_j^*$ by an exchangeability argument, and as a result, $p_j$ will be conservative on the event $\{Y_{n+j}\leq c_{n+j}\}$.

    Denote $\cZ_j =(Z_{n_1+1}, \dots, Z_n, Z_{n+j})$, and also $[\cZ_j] = [Z_{n_1+1}, \dots, Z_n, Z_{n+j}]$ as the unordered set of the calibration data and the $j$-th test point (with true label). 
    We will show that the set of scores $\{\bar{V}_i^{(j)}\}_{i=n_1+1}^n \cup \{\bar{V}^{(j)}_{n+j}\}$ are exchangeable conditional on $[\cZ_j]\cup \{\hat{Z}_{n+\ell}\}_{\ell\neq j}$. 
    Technically, we are to show that  for any values $(v_{n_1+1},\dots,v_n, v_{n+j})$, and any permutation $\pi$ of $\{n_1+1,\dots,n,n+j\}$, 
    \@\label{eq:exch_score}
& \PP\Big(  \bar{V}_{n_1+1}^{(j)} = v_{n_1+1}, \dots,  \bar{V}_{n}^{(j)} = v_n, \bar{V}_{n+j}^{(j)} = v_{n+j} \Biggiven [\cZ_j]\cup \{\hat{Z}_{n+\ell}\}_{\ell\neq j} \Big)   \notag   \\ 
&=\PP\Big(  \bar{V}_{n_1+1}^{(j)} = v_{\pi(n_1+1)}, \dots,  \bar{V}_{n}^{(j)} = v_{\pi(n)}, \bar{V}_{n+j}^{(j)} = v_{\pi(n+j)} \Biggiven [\cZ_j]\cup \{\hat{Z}_{n+\ell}\}_{\ell\neq j}\Big) ,
    \@
    even though they are produced by a highly data-dependent process. For clarity, we write any realized value of the unordered set $[\cZ_j]$ as $[z]=[z_{n_1+1},\dots,z_n,z_{n+j}]$, any realized value of $\{\hat{Z}_{n+\ell}\}_{\ell\neq j}$ as  $\hat{\mathbf{z}}_{-j}$, and condition on the event  $\{ [\cZ_j] = [z],~\{\hat{Z}_{n+\ell}\}_{\ell\neq j} = \hat{\mathbf{z}}_{-j}\}$ in what follows. 

    The permutation equivariance of $\cV$ (Definition~\ref{def:permu_equiv}) ensures the collection of output scores are not affected by the ordering of data in $\cZ_j$, and as a result, the unordered set of the scores $[\{\bar{V}_i^{(j)}\}_{i=n_1+1}^n \cup \{\bar{V}^{(j)}_{n+j}\}]$ is fully determined by $[\cZ_j]$ and $\{\hat{Z}_{n+\ell}\}_{\ell \neq j}$. Thus, the probabilities on both sides of~\eqref{eq:exch_score} are nonzero if and only if $[v_{n_1+1},\dots,v_n,v_{n+j}]$ equals the unordered set of $\cV^{(j)}(Z_{n_1+1},\dots,Z_n,\hat{Z}_{n+1:n+j-1},Z_{n+j},\hat{Z}_{n+j+1:n+m})$. We thus restrict our attention to this case, and without loss of generality, let 
    \$
    (v_{n_1+1},\dots,v_n, v_{n+j}) = \cV^{(j)} (z_{1},\dots,z_{n} , \hat{z}_{n+1:n+j-1}, z_{n+j}, \hat{z}_{n+j+1:n+m}).
    \$
    We also assume $(v_{n_1+1},\dots,v_n, v_{n+j})$ are distinct; otherwise we conceptually add an ordering of them to distinguish two elements of the same value. 
    
    Given such an event,  the only randomness is in which value in $[z_1,\dots,z_n,z_{n+j}]$ corresponds to the data points $Z_1,\dots,Z_n,Z_{n+j}$. 
    Since $(Z_{n_1+1},\dots,Z_n,Z_{n+j})$ is exchangeable given $\{\hat{Z}_{n+\ell}\}_{\ell \neq j}$, we know that for any permutation $\pi'$ of $\{n_1+1,\dots,n,n+j\}$, 
    \$
    &\PP\Big( Z_{n_1+1}=z_{\pi'(n_1+1)},\dots, Z_n =z_{\pi'(n)}, Z_{n+j} = z_{\pi'(n+j)} \Biggiven [\cZ_j] = [z],~\{\hat{Z}_{n+\ell}\}_{\ell\neq j} = \hat{\mathbf{z}}_{-j} \Big) 
    \\&= \frac{1}{(1+n_2)!}.
    \$
    For any permutation $\pi$ of $\{n_1+1,\dots,n,n+j\}$, due to permutation equivariance in Definition~\ref{def:permu_equiv}, 
    \$
    &\PP\Big(  \bar{V}_{n_1+1}^{(j)} = v_{\pi(n_1+1)}, \dots,  \bar{V}_{n}^{(j)} = v_{\pi(n)}, \bar{V}_{n+j}^{(j)} = v_{\pi(n+j)} \Biggiven [\cZ_j]\cup \{\hat{Z}_{n+\ell}\}_{\ell\neq j}\Big) \\ 
    &= \PP\Big(  Z_{n_1+1}  = z_{\pi(n_1+1)}, \dots,  Z_{n}  = z_{\pi(n)}, Z_{n+j}  = z_{\pi(n+j)} \Biggiven [\cZ_j]\cup \{\hat{Z}_{n+\ell}\}_{\ell\neq j}\Big)  = \frac{1}{(1+n_2)!}.
    \$
    This proves~\eqref{eq:exch_score}. 
    Similar to the standard result in conformal inference,~\eqref{eq:exch_score} implies 
    \$
        p_j^* \biggiven [\cZ_j] \cup \{\hat{Z}_{n+\ell}\}_{\ell \neq j} \quad \sim \quad 
        \textrm{Unif}\big( \{ 1/(n_2+1),\dots,n_2/(n_2+1),1\} \big).
    \$
    as $p_j^*$ depends on $[\cZ_j] \cup \{\hat{Z}_{n+\ell}\}_{\ell \neq j}$ only through the scores.
    In other words, for any $t\in [0,1]$ that is measurable with respect to $[\cZ_j]\cup\{\hat{Z}_{n+\ell}\}_{\ell \neq j}$, we have 
    \$
        \PP\big( p_j^* \leq t \biggiven [\cZ_j]\cup\{\hat{Z}_{n+\ell}\}_{\ell \neq j} \big)  \leq t. \tag{$*$}
    \$
    This is the conditional subuniformity  of the ``oracle'' conformal p-value $p_j^*$.

    The proof will now rely on Lemma~\ref{lem:fdr_decomp}, which is an application of Lemma C.1 from~\cite{jin2023model}, to bound the FDR for all three pruning methods. We will then show FDR control through a series of inequalities. From the lemma, we have
    \$
        \fdr \leq \sum_{j=1}^m \EE\bigg[  \frac{\ind\{p_j \leq q \hat{R}_j/m, ~ Y_{n+j}\leq c_{n+j}\}}{1\vee \hat{R}_j}   \bigg],
    \$
    and we bound each summand individually below $q/m$. 
    For any $j \in [m]$, by the permutation invariance of $\cR$ under the $j$-th null (Definition~\ref{def:R_mon_invar}), we have
    \$
        \cR( Z_{1:n_1}, Z_{n_1+1:n}, \hat{Z}_{n+1:n+j-1}, Z_{n+j}, \hat{Z}_{n+j+1:n+m})_j =: \bar{R}_j = \hat{R}_j,
    \$
    and that the auxiliary selection size $\bar{R}_j$ is measurable given $[\cZ_j]\cup\{\hat{Z}_{n+\ell}\}_{\ell \neq j}$ since it is permutation invariant.
    Combined with $p_j^* \leq p_j$, this implies
    \@ \label{eq:permu_invar_relax}
        \EE\Bigg[  \frac{\ind\{p_j \leq q\hat{R}_j/m, Y_{n+j}\leq c_{n+j}\}}{1\vee \hat{R}_j}  \Bigg] &\leq  \EE\Bigg[  \frac{\ind\{p_j^* \leq q\hat{R}_j/m, Y_{n+j}\leq c_{n+j}\}}{1\vee \bar{R}_j}  \Bigg]  \\ \nonumber
         &\leq  \EE\Bigg[  \frac{\ind\{p_j^* \leq q\bar{R}_j/m\}}{1\vee \bar{R}_j}  \Bigg].
    \@
    Since $q\bar{R}_j/m$ is measurable given $[\cZ_j]\cup\{\hat{Z}_{n+\ell}\}_{\ell \neq j}$,
    \$
        \EE\Bigg[  \frac{\ind\{p_j^* \leq q\bar{R}_j/m\}}{1\vee \bar{R}_j} \Bigggiven [\cZ_j]\cup\{\hat{Z}_{n+\ell}\}_{\ell \neq j} \Bigg]
        &\leq \frac{ \PP( p_j^* \leq q\bar{R}_j/m \given[\cZ_j]\cup\{\hat{Z}_{n+\ell}\}_{\ell \neq j}  ) }{1\vee \bar{R}_j} \\
        &\leq q/m.
    \$
    Finally, taking the expectation over $[\cZ_j]\cup\{\hat{Z}_{n+\ell}\}_{\ell \neq j}$ completes the proof of Theorem~\ref{thm:optcs_fdr_control}.
\end{proof}

\subsection{Proof of Proposition~\ref{lemma:modsel_valid}}
\label{proof:modsel_valid}

\begin{proof}[Proof of Proposition~\ref{lemma:modsel_valid}.]
We begin by proving that the proposed $\hat{R}_j$ satisfies  Definition~\ref{def:R_mon_invar}.
Fix any $j \in [m]$, and we specify the $j$-th output of ``oracle'' auxiliary selection size functional here as that of $\cR$ depicted in Section~\ref{subsec:modsel} with $Z_{n+j}$ used in place of $\hat{Z}_{n+j}$. 
Put another way, the $j$-th element of the oracle functional, which we denote as $\bar{R}_j$, is defined by the size of the output of the BH procedure applied to $\{\bar{p}_\ell^{(j)} (k)\}$, the ``oracle'' counterparts of the modified p-values:
\$
\bar{p}_\ell^{(j)} (k) = \frac{\sum_{i\in \cI_\calib} \ind\{V_i(k) \leq \hat{V}_{n+\ell}(k)\}+ \ind\{{V}_{n+j}(k)\leq \hat{V}_{n+\ell}(k)\}}{n_2+1},\quad \ell\in[m],~\ell\neq j.
\$
Recall that we use the same definitions of $V_i(k)$ and $\hat{V}_{n+\ell}(k)$. 
In the above expression, $V_{n+j}(k)$ is used in place of $\hat{V}_{n+j}(k)$ compared with the modified p-values.
We first show that the ``oracle'' modified p-values $\{\bar{p}_\ell^{(j)}(k)\}_{\ell \neq j}$ are invariant to any permutation on the calibration and test samples $\cZ_j := \{Z_{n_1+1}, \dots, Z_n, Z_{n+j}\}$, for any $k \in [K]$ and $\ell \neq  j$. By construction,
\$
    \bar{p}_\ell^{(j)} (k) &= \frac{\sum_{i\in \cI_\calib} \ind\{V_i(k) \leq \hat{V}_{n+\ell}(k)\}+ \ind\{ {V}_{n+j}(k)\leq \hat{V}_{n+\ell}(k)\}}{n_2+1} \\
    &= \frac{\sum_{i \in \cI_\calib \cup \{n+j\}} \ind\{V_i(k) \leq \hat{V}_{n+\ell}(k)\}}{n_2+1} 
    = \frac{\sum_{z \in \cZ_j} \ind\{V(z; k) \leq V(X_{n+\ell}, c_{n+\ell}; k)\}}{n_2+1}
\$
is invariant of any permutation of $\cZ_j$. Denote $\bar\cS_j(k)$ as the output of BH with the set of oracle modified p-values. 
Let $\bar{k}_j$ be the oracle version of $\hat{k}_j$, i.e., $\bar{k}_j = \argmax_{k\in[K]}|\bar\cS_j(k)|$. 
Since $\bar\cS_j(k)$ is a deterministic function of the modified p-values, it is also permutation invariant to $\cZ_j$ for any $k$. Hence, $\bar{k}_j$ is also permutation invariant to $\cZ_j$.\footnote{A random tie-breaking may be involved here. However, as noted earlier, such randomness will not compromise the validity of our approach given that the random variables $\{U_j\}$ are independent of the data.}
This means the oracle auxiliary selection size $\bar{R}_j = |\bar\cS_j(\bar{k}_j)|$ will also be permutation invariant, establishing the second condition  in Definition~\ref{def:R_mon_invar}. We then proceed to prove the first condition that $\hat{R}_j$ and $\bar{R}_j$ coincide in the $j$-th element on the null event. Since we use clipped scores, it holds that $V_{n+j}(k) = \hat{V}_{n+j}(k)$ on the event $\{y_{n+j} \leq c_{n+j}\}$ for all $k\in[K]$.
This further implies $\bar{p}_\ell^{(j)} (k) = \tilde{p}_\ell^{(j)} (k)$ for all $\ell \neq j$ and all $k\in [K]$, and therefore  $|\bar\cS_j(k)| = |\cS_j(k)|$ for all $k\in[K]$. As a consequence, for any $y_{n+j}\leq c_{n+j}$, we have $\hat{k}_j = \bar{k}_j$ and $|\bar\cS_j(\bar{k}_j)|=|\cS_j(\hat{k}_j)|$, which implies that $\hat{R}_j$ satisfies the first condition  in Definition~\ref{def:R_mon_invar}.

We now show the permutation equivariance property (Definition~\ref{def:permu_equiv}) of the score-generating functional $\cV$. %
We will achieve this by showing the permutation invariance of $\hat{k}_j$  together with the definition of $\cV$. 
For notational clarity, here we consider any hypothesized values $z_{1:n_1}'$ for $Z_{1:n_1}$, $z_{1:n_2}$ for $Z_{(n_1+1):n}$, and ${z}_{(n+1):(n+m)}$ for $\hat{Z}_{(n+1):(n+m)}$. Following our definition of OptCS-MSel, 
for these hypothesized input values, we can write 
\$
   & \cV^{(j)}(z_{1:n_1}',z_{ 1 },\dots, z_{ n_2 }, {z}_{(n+1):(n+j-1)},  {z}_{ n+j }, z_{(n+j+1):(n+m)}) \\ 
   & := \Big( v(x_1,y_1;\hat{k}_j),  \dots, v(x_{n_2},y_{n_2};\hat{k}_j), v(x_{n+1},y_{n+1};\hat{k}_j), \dots, {v}(x_{n+m},y_{n+m};\hat{k}_j) \Big),
\$ 
where $\hat{k}_j = \hat{k}_j(z_{1:n_1}',z_{1:n_2}, {z}_{(n+1):(n+m)})$ is the selected model index for each $j\in [m]$. 

Using the same arguments as the preceding part,  we can show that the modified p-values $\{\tilde{p}_{\ell}^{(j)}\}_{\ell\neq j}$, and thus the selected model index $\hat{k}_j$, are invariant to any permutation of the calibration data and the $j$-th test point. 
That is, for any permutation $\pi\colon \{1,\dots,n_2,n+j\}\to \{1,\dots,n_2,n+j\}$, it holds that  
\$
&\hat{k}_j\big(z_{1:n_1}',z_{\pi(1)},\dots, z_{\pi(n_2)}, {z}_{(n+1):(n+j-1)},  {z}_{\pi(n+j)}, z_{(n+j+1):(n+m)}\big) \\ 
&= \hat{k}_j\big(z_{1:n_1}', z_{ 1 },\dots, z_{ n_2},  {z}_{(n+1):(n+j-1)},  {z}_{n+j}, z_{(n+j+1):(n+m)}\big). 
\$
This implies 
\$
   & \cV^{(j)}\big(z_{1:n_1}',z_{\pi(1)},\dots, z_{\pi(n_2)}, {z}_{(n+1):(n+j-1)},  {z}_{\pi(n+j)}, z_{(n+j+1):(n+m)}\big)_i \\
   & = v\Big( x_{\pi(i)}, y_{\pi(i)}; \hat{k}_j (z_{1:n_1}',z_{\pi(1)},\dots, z_{\pi(n_2)}, {z}_{(n+1):(n+j-1)},  {z}_{\pi(n+j)}, z_{(n+j+1):(n+m)})\Big) \\ 
   &= v\Big( x_{\pi(i)}, y_{\pi(i)}; \hat{k}_j (z_{1:n_1}',z_{1},\dots, z_{ n_2 }, {z}_{(n+1):(n+j-1)},  {z}_{ n+j }, z_{(n+j+1):(n+m)})\Big).
\$
Also, again by the definition of $\cV^{(j)}$, we have 
\$
  & \cV^{(j)}\big(z_{1:n_1}',z_{ 1 },\dots, z_{ n_2 }, {z}_{(n+1):(n+j-1)},  {z}_{ n+j }, z_{(n+j+1):(n+m)}\big)_{\pi(i)} \\
     &= v\Big( x_{\pi(i)}, y_{\pi(i)}; \hat{k}_j (z_{1:n_1}',z_{1},\dots, z_{ n_2 }, {z}_{(n+1):(n+j-1)},  {z}_{ n+j }, z_{(n+j+1):(n+m)})\Big).
\$
Comparing the two equations above proves the permutation equivariance property (Definition~\ref{def:permu_equiv}). 

Finally, we prove the monotonicity for the null (Definition~\ref{def:monotone}).  
For $y_{n+j}\leq c_{n+j}$, with slight abuse of notations we let $\hat{k}_j$ be the selected model index using $c_{n+j}$ and $\bar{k}_j$ be the selected model index using $y_{n+j}$, while keeping other parts of the inputs intact. 
By the definition of $\cV^{(j)}$, we have 
\$
    \cV^{(j)}(\textcolor{Gray}{z_{1:n_1}',z_{1:n_2}, \hat{z}_{n+1:n+j-1},} \,\hat{z}_{n+j}, \textcolor{Gray}{\hat{z}_{n+j+1:n+m}})_{n_2+j} = v(x_{n+j}, c_{n+j};\hat{k}_j)
\$
and 
\$
\cV^{(j)}(\textcolor{Gray}{z_{1:n_1}',z_{1:n_2}, \hat{z}_{n+1:n+j-1},} \, {z}_{n+j},\textcolor{Gray}{\hat{z}_{n+j+1:n+m}})_{n_2+j} = v(x_{n+j}, y_{n+j};\bar{k}_j).
\$
Since we use the clipped score, using similar arguments as before, we have $\hat{k}_j=\bar{k}_j$ when $y_{n+j}\leq c_{n+j}$. Then since each individual score is monotone, we have $v(x_{n+j}, c_{n+j};\hat{k}_j) \geq v(x_{n+j}, y_{n+j};\hat{k}_j) =v(x_{n+j}, y_{n+j};\bar{k}_j)$, that is, 
\$
    &\cV^{(j)}(\textcolor{Gray}{z_{1:n_1}',z_{1:n_2}, \hat{z}_{n+1:n+j-1},} \,\hat{z}_{n+j}, \textcolor{Gray}{\hat{z}_{n+j+1:n+m}})_{n+j} \\ &\geq \cV^{(j)}(\textcolor{Gray}{z_{1:n_1}',z_{1:n_2}, \hat{z}_{n+1:n+j-1},} \, {z}_{n+j},\textcolor{Gray}{\hat{z}_{n+j+1:n+m}})_{n+j} 
\$
as desired. In addition, for other entries, since $\hat{k}_j= \bar{k}_j$,  we have 
\$
&\cV^{(j)}(\textcolor{Gray}{z_{1:n_1}',z_{1:n_2}, \hat{z}_{n+1:n+j-1},} \,\hat{z}_{n+j}, \textcolor{Gray}{\hat{z}_{n+j+1:n+m}})_{ \ell} \\&= \cV^{(j)}(\textcolor{Gray}{z_{1:n_1}',z_{1:n_2}, \hat{z}_{n+1:n+j-1},} \, {z}_{n+j},\textcolor{Gray}{\hat{z}_{n +j+1:n+m}})_{ \ell} 
\$ 
for any other index $\ell \neq n+j$.
Therefore, we conclude the proof of Proposition~\ref{lemma:modsel_valid}.
\end{proof}

\subsection{Proof of Proposition~\ref{lemma:loo_full_valid}}
\label{proof:loo_full_valid}

\begin{proof}[Proof of Proposition~\ref{lemma:loo_full_valid}.]
    \frevise{We prove the proposition for the transformed training observations $Z_i^T$ and $\hat{Z}_{n+j}^T$ using any function $T$ obeying $T(y,c)=T(c,c)$ whenever $y\leq c$.} 
    
    Fix $j \in [m]$ throughout. We first show the permutation equivariance of $\cV$ (Definition~\ref{def:permu_equiv}) given any input values. 
    We consider hypothesized inputs $z_{1:n_1}'$ for $Z_{1:n_1}$, $z_{1:n_2}$ for $Z_{(n_1+1):n}$, and  $z_{(n+1):(n+m)}$ for $\hat{Z}_{(n+1):(n+m)}$. 
    By construction of the procedure, for any permutation $\pi$
    of $\{1,\cdots,n_2,n+j\}$, extended by the identity to
    the remaining local indices, and any $i\in\cL=\{1,\dots,n_2,n+1,\dots,n+m\}$,
    \$
        &\cV^{(j)}(z_{1:n_1}', z_{\pi(1:n_2)},
        z_{n_2+1}, \dots, z_{\pi(n+j)}, \dots, z_{n+m})_i
        = v\Big(x_{\pi(i)},y_{\pi(i)}
        \given \hat\mu_{\pi(i)}\Big),
        \quad \text{and} \\
        &\cV^{(j)}(z_{1:n_1}', z_{1:n_2},
        z_{n+1}, \dots, z_{n+j}, \dots, z_{n+m})_{\pi(i)}
        = v\Big(x_{\pi(i)},y_{\pi(i)}
        \given \hat\mu'_{\pi(i)}\Big).
    \$
    Here, $\hat\mu_{\pi(i)}$ is trained through the algorithm
    $\cA$ on the dataset
    \$
        (z_{1:n_1}^{\prime,T}, z_{\pi(1:n_2)}^T,
        z_{n+1}^T, \dots, z_{\pi(n+j)}^T,
        \dots, z_{n+m}^T)
        \setminus z_{\pi(i)}^T,
    \$
    while $\hat\mu_{\pi(i)}'$ is trained through $\cA$
    with the same collection of data but in a different order,
    i.e.,
    \$
        (z_{1:n_1}^{\prime,T}, z_{1:n_2}^T,
        z_{n+1}^T, \dots, z_{n+j}^T,
        \dots, z_{n+m}^T)
        \setminus z_{\pi(i)}^T.
    \$
    Since the fixed transformation $T$ commutes with permutations,
    the symmetry~\eqref{eq:training_symmetry} of $\cA$ guarantees $\hat\mu_{\pi(i)}=\hat\mu_{\pi(i)}'$,
    which shows the permutation equivariance of $\cV$ (Definition~\ref{def:permu_equiv}) given any input values.
    
    The first condition of monotonicity
    (Definition~\ref{def:monotone}) directly follows from
    the monotonicity of candidate scores and the fact that
    each observation is scored using a model trained without
    that observation.
    We now formally prove the second condition in
    Definition~\ref{def:monotone}. To this end, we consider
    hypothetical inputs $z_{1:n_1}'$ for $Z_{1:n_1}$,
    $z_{1:n_2}$ for $Z_{(n_1+1):n}$,
    $\hat z_{(n+1):(n+m)}$ for
    $\hat Z_{(n+1):(n+m)}$, and $z_{n+j}$ for $Z_{n+j}$.
    Overriding the notations above, we let $\hat\mu_i$ be
    the $i$-th leave-one-out trained model with inputs
    $(z_{1:n_1}^{\prime,T},z_{1:n_2}^T,
    \hat z_{n+1}^T,\dots,\hat z_{n+j}^T,
    \dots,\hat z_{n+m}^T)$ except the $i$-th observation,
    and $\hat\mu_i^*$ be the corresponding model with
    $\hat z_{n+j}^T$ replaced by $z_{n+j}^T$.
    By the condition on $T$, when $y_{n+j}\le c_{n+j}$, we have 
    \$
        T(y_{n+j},c_{n+j})
        =T(c_{n+j},c_{n+j}),
        \qquad
        z_{n+j}^T=\hat z_{n+j}^T.
    \$
    As such, we have $\hat\mu_i=\hat\mu_i^*$ for all
    $i$, using the same fitting randomness in
    the two constructions if $\cA$ is randomized.
    Since all other scoring records are unchanged, this
    proves the second condition in Definition~\ref{def:monotone}.
    Putting the above two facts together, we prove statement (i).

    We return to the original observation indices below.
    The first fact in (ii) is an implication of
    Proposition~\ref{prop:Rj}. Now, we relax
    Definition~\ref{def:R_mon_invar} by only requiring
    the first condition to hold when
    $p_j\le q\hat R_j/m$ and $Y_{n+j}\le c_{n+j}$.
    Such relaxation will not compromise the validity of
    Theorem~\ref{thm:optcs_fdr_control}, as
    \eqref{eq:permu_invar_relax} in its proof only uses
    $\bar R_j=\hat R_j$ on
    $\{p_j\le q\hat R_j/m,\ Y_{n+j}\le c_{n+j}\}$.
    
    In the following, we specify the $j$-th oracle auxiliary
    selection size $\bar R_j$.
    We define the set of oracle modified p-values
    $\{\bar p_\ell^{(j)}\}_{\ell\ne j}$ similarly to
    the proof of Proposition~\ref{lemma:modsel_valid};
    that is,
    \$
        \bar p_\ell^{(j)}
        =
        \frac{
        \sum_{i\in\cI_\calib}
        \ind\{V_i^*\le\hat V_{n+\ell}^*\}
        +
        \ind\{V_{n+j}^*\le\hat V_{n+\ell}^*\}
        }{n_2+1},
        \quad \ell\in[m],~\ell\ne j.
    \$
    Here we define the oracle leave-one-out scores
    \$
        V_i^*=v(X_i,Y_i\given\hat\mu_i^*),
        \quad i\in\cI_\calib\cup\{n+j\},
        \qquad
        \hat V_{n+\ell}^*
        =v(X_{n+\ell},c_{n+\ell}
        \given\hat\mu_{n+\ell}^*),
        \quad \ell\in[m],
    \$
    where $\hat\mu_\ell^*=\cA(\mathbf D_{-\ell}^*)$ and
    \$
        \mathbf D_{-\ell}^*
        :=
        (Z_1^T,\dots,Z_n^T,
        \hat Z_{n+1}^T,\dots,\hat Z_{n+j-1}^T,
        Z_{n+j}^T,
        \hat Z_{n+j+1}^T,\dots,\hat Z_{n+m}^T)
        \setminus\tilde Z_\ell^{*,T},
    \$
    with $\tilde Z_\ell^{*,T}=Z_\ell^T$ if
    $\ell\in\{n_1+1,\dots,n,n+j\}$ and
    $\tilde Z_\ell^{*,T}=\hat Z_\ell^T$ otherwise.
    That is, these are the counterparts of the scores in OptCS-Full when $\hat Z_{n+j}$ is replaced by $Z_{n+j}$, with transformed records used for training
    and original records used for scoring.

    We then set $\bar R_j$ as the selection size of BH applied to $\{\bar p_\ell^{(j)}\}_{\ell\ne j}\cup\{0\}$ at nominal level $q$. Note that we view the quantities defined above (and those in OptCS-Full) implicitly as functions of hypothesized inputs
    \$
        &\hat\cZ_j
        :=(Z_{1:n},\hat Z_{(n+1):(n+m)})
        =\hat{\mathbf z}_j
        :=(z_{1:n_1}',z_{n_1+1:n},
        \hat z_{(n+1):(n+m)}),\\
        &\cZ_j
        :=(Z_{1:n},\hat Z_{(n+1):(n+j-1)},
        Z_{n+j},\hat Z_{(n+j+1):(n+m)})\\
        &\qquad
        =\mathbf z_j
        :=(z_{1:n_1}',z_{n_1+1:n},
        \hat z_{(n+1):(n+j-1)},z_{n+j},
        \hat z_{(n+j+1):(n+m)}).
    \$

    We now show that $\bar R_j$ obeys the relaxed version
    of Definition~\ref{def:R_mon_invar}.
    Specifically, we will show that (a) $\bar R_j$ is
    permutation invariant with respect to $\mathbf z_j$,
    and (b) $\bar R_j=|\cS_\bh|$ when
    $p_j\le s_j=q\hat R_j/m$ and $Y_{n+j}\le c_{n+j}$.

    To show (a), we note that each $\bar p_\ell^{(j)}$
    is invariant to permutations of the calibration
    observations and the $j$-th observation in $\mathbf z_j$.
    Formally, for any permutation $\pi$ of
    $\cI_\calib\cup\{n+j\}$, extended by the identity
    elsewhere, denote the corresponding oracle modified
    p-value after permutation by $\bar p_\ell^{(j)}(\pi)$.
    Then by definition,
    \$
        \bar p_\ell^{(j)}(\pi)
        =
        \frac{
        \sum_{i\in\cI_\calib\cup\{n+j\}}
        \ind\{V_{\pi(i)}^*\le\hat V_{n+\ell}^*\}
        }{n_2+1}
        =
        \bar p_\ell^{(j)},
    \$
    where we use the permutation equivariance of the LOO
    scores and the invariance of $\hat V_{n+\ell}^*$
    for $\ell\ne j$, which follow from the observationwise
    transformation and the symmetry of $\cA$.
    Since each $\bar p_\ell^{(j)}$ is invariant to
    permutations of $\{Z_{n_1+1},\dots,Z_n,Z_{n+j}\}$,
    so is the output of BH applied to them, and so is
    $\bar R_j$. This proves (a). 
    
    We now show the relaxed condition (b).
    On the null event $\{Y_{n+j}\le c_{n+j}\}$,
    null compatibility gives
    $Z_{n+j}^T=\hat Z_{n+j}^T$.
    Consequently, the training samples and fitted models
    coincide; by the clipped-score property,
    \$
        \mathbf D_{-\ell}^*=\mathbf D_{-\ell},
        \quad
        \hat\mu_\ell^*=\hat\mu_\ell,
        \quad
        V_i^*=V_i \quad (i\in\cI_\calib),
        \quad
        \hat V_{n+\ell}^*=\hat V_{n+\ell},
        \quad \text{and}\quad
        V_{n+j}^*=\hat V_{n+j}.
    \$
    That is, $Y_{n+j}\le c_{n+j}$ implies
    \$
        \bar p_\ell^{(j)}
        =
        \frac{
        \sum_{i\in\cI_\calib}
        \ind\{V_i\le\hat V_{n+\ell}\}
        +
        \ind\{\hat V_{n+j}\le\hat V_{n+\ell}\}
        }{n_2+1}
        =:p_\ell^{(j)}.
    \$
    Thus, $\bar R_j$ is the output size of BH at nominal
    level $q$ applied to
    $\{p_\ell^{(j)}\}_{\ell\ne j}\cup\{0\}$ on the
    null event, and the $p_\ell^{(j)}$ are expressed
    in terms of observed inputs.
    Now assuming $Y_{n+j}\le c_{n+j}$ and
    $p_j\le s_j:=q\hat R_j/m$, we show the relaxed
    condition using $p_\ell^{(j)}$ in place of
    $\bar p_\ell^{(j)}$. We consider the following
    two cases:
    \begin{enumerate}[label=(\roman*).]
        \item
        If $\hat V_{n+j}\le\hat V_{n+\ell}$, then
        $V_{n+j}^*=\hat V_{n+j}\le\hat V_{n+\ell}$
        by the clipped-score property on the null event.
        In this case,
        $\bar p_\ell^{(j)}=p_\ell\ge p_j$.

        \item
        If $\hat V_{n+j}>\hat V_{n+\ell}$, then
        $p_\ell\le p_j$; it also holds that
        \$
            p_\ell^{(j)}
            \le
            \frac{
            \sum_{i\in\cI_\calib}
            \ind\{V_i\le\hat V_{n+\ell}\}+1
            }{n_2+1}
            \le
            \frac{
            \sum_{i\in\cI_\calib}
            \ind\{V_i\le\hat V_{n+j}\}+1
            }{n_2+1}
            =p_j.
        \$
    \end{enumerate}
    The inequality
    $p_j\le s_j=q\hat R_j/m=q|\cS_{\textnormal{BH}}|/m$
    indicates $j\in\cS_{\textnormal{BH}}$, i.e.,
    the $j$-th test sample is falsely rejected.
    If we change the input p-values to BH from
    \$
        p_1,\dots,p_{j-1},p_j,p_{j+1},\dots,p_m
        \qquad\text{to}\qquad
        p_1^{(j)},\dots,p_{j-1}^{(j)},0,
        p_{j+1}^{(j)},\dots,p_m^{(j)},
    \$
    all p-values larger than $p_j$ remain unchanged
    by (i), and only already rejected p-values can
    decrease. Hence no new rejections can be made.
    In addition, all p-values $p_\ell\le p_j$
    remain no greater than $p_j$ after the replacement
    by (i)--(ii), so all original rejections are retained,
    leading to a rejection set of size at least
    $|\cS_{\textnormal{BH}}|=\hat R_j$.
    The two facts imply that $\bar R_j=\hat R_j$
    when $Y_{n+j}\le c_{n+j}$ and $p_j\le s_j$,
    demonstrating the first condition and hence the
    relaxed permutation invariance of $\cR$.
    We thus conclude the proof of
    Proposition~\ref{lemma:loo_full_valid}.
\end{proof}

\subsection{Proof of Proposition~\ref{prop:validity_loosel}}

\begin{proof}[Proof of Proposition~\ref{prop:validity_loosel}.]
\label{proof:validity_loosel} 
    Fix any $j\in [m]$ throughout. 
    We first show the permutation equivariance of the
    score functional (Definition~\ref{def:permu_equiv}).
    Namely, we consider $\cV$ with any hypothesized inputs
    $\mathbf z=(z_{1:n_1}',z_{1:n_2},z_{n+1:n+m})$ and any permutation
    $\pi$ of $\cI_j:=\{1,\dots,n_2,n+j\}$.
    Extend $\pi$ by the identity outside $\cI_j$.
    As such, all quantities defined in OptCS-Full-MSel
    can be implicitly viewed as functions of $\mathbf z$.
    The permuted input is denoted as
    \$
        \pi(\mathbf z)
        :=
        (z_{1:n_1}',z_{\pi(1:n_2)},
        z_{n+1},\dots,z_{\pi(n+j)},\dots,z_{n+m}).
    \$

    Similar to the proof of
    Proposition~\ref{lemma:loo_full_valid}, each fixed
    transformation $T_k$ commutes with permutations.
    Denote by $\hat\mu_{i,k}^{\pi}$ the model trained via
    $\cA_k$ on the $T_k$-transformed records after leaving
    out the $i$-th element of the permuted input
    $\pi(\mathbf z)$.
    Due to the symmetry of $\cA_k$, we have 
        $\hat\mu_{i,k}^{\pi}=\hat\mu_{\pi(i),k}$
    for  $i\in \cL=[n_2]\cup\{n+1,\dots,n+m\}$. 
    Denote the modified p-values $p_\ell^{(j)}(k)$
    obtained after permutation by $p_\ell^{(j)}(k,\pi)$.
    For $\ell\in[m]\setminus\{j\}$, this implies
    \$
        p_\ell^{(j)}(k,\pi)
        &=
        \frac{
        \sum_{i\in\cI_j}
        \ind\{
            v(x_{\pi(i)},y_{\pi(i)}
              \given\hat\mu_{i,k}^{\pi},k)
            \le
            v(x_{n+\ell},y_{n+\ell}
              \given\hat\mu_{n+\ell,k}^{\pi},k)
        \}
        }{n_2+1}\\
        &=
        \frac{
        \sum_{i\in\cI_j}
        \ind\{
            v(x_i,y_i\given\hat\mu_{i,k},k)
            \le
            v(x_{n+\ell},y_{n+\ell}
              \given\hat\mu_{n+\ell,k},k)
        \}
        }{n_2+1}
        =
        p_\ell^{(j)}(k).
    \$
    Here, the second equality uses reindexing over
    $\cI_j$ and $\pi(n+\ell)=n+\ell$ for $\ell\ne j$. 

    That is, all the modified p-values, hence
    $|\cS_j(k)|$ and $\hat k_j$, are permutation invariant
    with respect to $(z_1,\dots,z_{n_2},z_{n+j})$,
    using the same independent tie-breaking randomness.
    Therefore, for any permutation $\pi$ of $\cI_j$,
    \$
        &\cV^{(j)}(z_{1:n_1}',z_{1:n_2}, z_{n+1:n+m})_{\pi(i)}
        =
        v(z_{\pi(i)}
          \given\hat\mu_{\pi(i),\hat k_j},\hat k_j),
        \quad\text{and}\\
        &\cV^{(j)}(\pi(\mathbf z))_i
        =
        v(z_{\pi(i)}
          \given\hat\mu_{i,\hat k_j}^{\pi},\hat k_j)
        =
        v(z_{\pi(i)}
          \given\hat\mu_{\pi(i),\hat k_j},\hat k_j).
    \$
    This proves Definition~\ref{def:permu_equiv}.

    We then show the monotonicity property in
    Definition~\ref{def:monotone}.
    Overriding the notations, we consider any hypothesized
    inputs 
        $\hat{\mathbf z} := (z_{1:n_1}',z_{1:n_2}, \hat z_{(n+1):(n+m)})$ 
    and their ``oracle'' version 
        $\mathbf z
        :=
        (z_{1:n_1}',z_{1:n_2},\allowbreak
        \hat z_{(n+1):(n+j-1)},\allowbreak z_{n+j},\allowbreak
        \hat z_{(n+j+1):(n+m)})$, 
    with $\hat z_{n+j}=(x_{n+j},c_{n+j})$,
    $z_{n+j}=(x_{n+j},y_{n+j})$, and
    $y_{n+j}\le c_{n+j}$.
    Null compatibility gives
    \$
        T_k(y_{n+j},c_{n+j})
        =
        T_k(c_{n+j},c_{n+j}),
        \qquad
        z_{n+j}^{T_k}=\hat z_{n+j}^{T_k},
        \quad k\in[K].
    \$
    Thus, all candidate-specific transformed training
    samples and fitted models remain unchanged.
    By the clipped-score property, the score at the
    $j$-th test point also remains unchanged, as do
    all other scores. Consequently, the modified
    p-values and $\hat k_j$ coincide in the two
    constructions, so the entire score vector
    $\cV^{(j)}$ remains the same.
    This proves Definition~\ref{def:monotone}.

    We return to the original observation indices and
    now show that $\cR$ satisfies
    Definition~\ref{def:R_mon_invar}.
    To compute the oracle selection size $\bar R_j$,
    we apply BH to the oracle
    modified p-values
    \$
        \bar p_\ell^{(j)}(k)
        =
        \frac{
        \sum_{i=n_1+1}^{n}
        \ind\{V_i^*(k)\le\hat V_{n+\ell}^*(k)\}
        +
        \ind\{V_{n+j}^*(k)\le\hat V_{n+\ell}^*(k)\}
        }{n_2+1},
        \quad \ell\in[m],~\ell\ne j,
    \$
    where the oracle leave-one-out scores are 
        $V_i^*(k) = v(X_i,Y_i\given\hat\mu_{i,k}^*,k)$,
         $i\in\{n_1+1,\dots,n,n+j\}$, and 
        $\hat V_{n+\ell}^*(k)
        =
        v(X_{n+\ell},c_{n+\ell}
          \given\hat\mu_{n+\ell,k}^*,k)$,
         $\ell\in[m]$. 
    Here,
    $\hat\mu_{\ell,k}^*
    =\cA_k(\mathbf D_{-\ell}^{*,(k)})$
    for $\ell=n_1+1,\dots,n+m$, with 
        $\mathbf D_{-\ell}^{*,(k)} := (Z_1^{T_k},\dots,Z_n^{T_k},
        \hat Z_{n+1}^{T_k},\dots,\hat Z_{n+j-1}^{T_k},
        Z_{n+j}^{T_k},
        \hat Z_{n+j+1}^{T_k},\dots,\hat Z_{n+m}^{T_k})
        \setminus\tilde Z_\ell^{*,T_k}$, 
    where $\tilde Z_\ell^{*,T_k}=Z_\ell^{T_k}$ if
    $\ell\in\{n_1+1,\dots,n,n+j\}$ and
    $\tilde Z_\ell^{*,T_k}=\hat Z_\ell^{T_k}$ otherwise. 
    Similar to the proof of
    Proposition~\ref{lemma:loo_full_valid}, these oracle
    scores are obtained in the same way as the observed
    scores except that $\hat Z_{n+j}$ is replaced by
    $Z_{n+j}$, with transformed records used for training.

    We denote by $\bar\cS_j(k)$ the output of BH applied
    to $\{\bar p_\ell^{(j)}(k)\}_{\ell\ne j}\cup\{0\}$
    at the same level $q$, and define the oracle
    selected model index
    $\bar k_j:=\argmax_{k\in[K]}|\bar\cS_j(k)|$.
    We then set $\bar R_j=|\bar\cS_j(\bar k_j)|$.

    For conceptual clarity, these quantities are
    implicitly functions of hypothesized inputs
    \$
        &\hat\cZ_j
        :=
        (Z_{1:n},\hat Z_{(n+1):(n+m)})
        =
        \hat{\mathbf z}_j
        :=
        (z_{1:n_1}',z_{n_1+1:n},
        \hat z_{(n+1):(n+m)}),\\
        &\cZ_j
        :=
        (Z_{1:n},\hat Z_{(n+1):(n+j-1)},
        Z_{n+j},\hat Z_{(n+j+1):(n+m)})\\
        &\qquad
        =
        \mathbf z_j
        :=
        (z_{1:n_1}',z_{n_1+1:n},
        \hat z_{(n+1):(n+j-1)},z_{n+j},
        \hat z_{(n+j+1):(n+m)}).
    \$

    We first show the permutation invariance of $\cR$.
    This relies on the fact that
    $\{\bar p_\ell^{(j)}(k)\}_{\ell\ne j}$ are invariant
    to permutations of the calibration observations and
    $Z_{n+j}$ in $\cZ_j$, for every $k\in[K]$.
    Denoting $\cI_\calib=\{n_1+1,\dots,n\}$, we have
    \$
        \bar p_\ell^{(j)}(k)
        &=
        \frac{
        \sum_{i\in\cI_\calib}
        \ind\{V_i^*(k)\le\hat V_{n+\ell}^*(k)\}
        +
        \ind\{V_{n+j}^*(k)\le\hat V_{n+\ell}^*(k)\}
        }{n_2+1}\\
        &=
        \frac{
        \sum_{i\in\cI_\calib\cup\{n+j\}}
        \ind\{V_i^*(k)\le\hat V_{n+\ell}^*(k)\}
        }{n_2+1}.
    \$
    Since $T_k$ is applied observationwise and $\cA_k$
    is symmetric, $\hat V_{n+\ell}^*(k)$ is invariant
    to permutations of
    $\{Z_{n_1+1},\dots,\allowbreak Z_n,Z_{n+j}\}$ for $\ell\ne j$.
    On the other hand, for each
    $i\in\cI_\calib\cup\{n+j\}$, the model obtained by
    leaving out that observation is unchanged when
    that observation moves to another position.
    Thus, the corresponding LOO scores are permutation
    equivariant.
    Denoting the oracle modified p-value after
    permutation $\pi$ by $\bar p_\ell^{(j)}(k,\pi)$,
    we have
    \$
        \bar p_\ell^{(j)}(k,\pi)
        =
        \frac{
        \sum_{i\in\cI_\calib\cup\{n+j\}}
        \ind\{V_{\pi(i)}^*(k)
              \le\hat V_{n+\ell}^*(k)\}
        }{n_2+1}
        =
        \bar p_\ell^{(j)}(k).
    \$
    Since BH is a deterministic function of the input
    p-values, $\bar\cS_j(k)$ is permutation invariant
    with respect to
    $\{Z_{n_1+1},\dots,Z_n,Z_{n+j}\}$ for every $k\in[K]$.
    Consequently, $\bar k_j$ is also permutation
    invariant with the prescribed independent
    tie-breaking. This establishes the invariance
    condition in Definition~\ref{def:R_mon_invar},
    since $\bar R_j=|\bar\cS_j(\bar k_j)|$.

    We now show the first condition in
    Definition~\ref{def:R_mon_invar}.
    On the null event $\{Y_{n+j}\le c_{n+j}\}$,
    null compatibility gives
    $Z_{n+j}^{T_k}=\hat Z_{n+j}^{T_k}$ for every $k$.
    Hence the transformed training samples and fitted
    models coincide, and the clipped-score property
    further gives
    \$
        &\mathbf D_{-\ell}^{*,(k)}
        =\mathbf D_{-\ell}^{(k)},
        \qquad
        \hat\mu_{\ell,k}^*=\hat\mu_{\ell,k},\\
        &V_i^*(k)=V_i(k)\quad(i\in\cI_\calib),
        \qquad
        \hat V_{n+\ell}^*(k)=\hat V_{n+\ell}(k)
        \quad(\ell\in[m]),\qquad V_{n+j}^*(k)=\hat V_{n+j}(k).
    \$
    Therefore,
    $\bar\cS_j(k)=\cS_j(k)$ for every $k$ and
    $\bar k_j=\hat k_j$ with the same tie-breaking
    randomness. This implies $\hat R_j=\bar R_j$
    on the null event, since the two auxiliary sizes
    are obtained by the same procedure from identical
    candidate scores.
    We thus conclude that $\cR$ obeys
    Definition~\ref{def:R_mon_invar}.
    The FDR conclusion follows from
    Theorem~\ref{thm:optcs_fdr_control}.
\end{proof}

\subsection{Proof of Lemma~\ref{lem:fdr_decomp}}
\label{app:proof_fdr_decomp}

\begin{proof}[Proof of Lemma~\ref{lem:fdr_decomp}.]
    We prove the result for three pruning options separately.

    \vspace{0.5em}
    \noindent \textbf{Case 1: Heterogeneous pruning}.  By definition, we have
    \$
        \fdr &= \EE \Bigg[ \frac{\sum_{j=1}^{m} \ind\{j \in \cS, Y_{n+j} \leq c_{n+j}\}}{1 \vee |\cS|} \Bigg] \\&
        = \EE \Bigg[ \frac{\sum_{j=1}^m \ind\{p_j \leq s_j, Y_{n+j} \leq c_{n+j}\} \ind\{\xi_j \leq |\cS|/\hat{R}_j \} }{1 \vee |\cS|} \Bigg]
    \$  
    where the second inequality relies on that $r^* = |\cS|$, a property that follows from the proof of Proposition~\ref{prop:Rj}. Fix any $j \in [m]$. Once $j \in \cS$, the rejection set of our pruning method will not change if $\xi_j$ is sent to 0 while keeping other $\xi$'s unchanged, by the step-up nature of the pruning process. Denoting $\cS_{\xi_j \rightarrow 0}$ as the rejection set obtained by replacing $\xi_j$ by 0 in the pruning step, we have
    \$
        \fdr &= \sum_{j=1}^m \sum_{k=1}^m \EE \Bigg[ \frac{\ind\{|\cS| = k\}}{k} \ind\{p_j \leq s_j, Y_{n+j} \leq c_{n+j}\} \ind\{\xi_j \leq k/\hat{R}_j \} \Bigg] \\
        &\leq  \sum_{j=1}^m \sum_{k=1}^m \EE \Bigg[ \frac{\ind\{|\cS_{\xi_j \rightarrow 0}| = k\}}{k} \ind\{p_j \leq s_j, Y_{n+j} \leq c_{n+j}\} \ind\{\xi_j \leq k/\hat{R}_j \} \Bigg] \\
        &\leq \sum_{j=1}^m \sum_{k=1}^m \EE \Bigg[ \frac{\ind\{|\cS_{\xi_j \rightarrow 0}| = k\}}{\hat{R}_j} \ind\{p_j \leq s_j, Y_{n+j} \leq c_{n+j} \} \Bigg] \\ &= \sum_{j=1}^m \EE \Bigg[ \frac{\ind\{p_j \leq s_j, Y_{n+j}\leq c_{n+j} \} }{\hat{R}_j}  \Bigg] 
    \$
    where the second equality follows from the independence between $\xi_j$ and $(\cS_{\xi_j \rightarrow 0}, p_j,\hat{R}_j)$. This shows the heterogeneous case.

    \vspace{0.5em}
    \noindent \textbf{Case 2: Homogeneous pruning}. In this case, we use a conditional PRDS property of $\{\xi \hat{R}_j\}_{j=1}^m$ which we will specify later. Let $\EE_\cD$ be the conditional expectation given all the data, so the only randomness lies in $\xi$. By the independence of $\xi$, we would still have $\xi \sim \text{Unif}([0,1])$ after conditioning. We define the conditional FDR:
    \$
        \fdr(\cD) := \EE_\cD \Bigg[ \frac{\sum_{j=1}^{m} \ind\{j \in \cS, Y_{n+j} \leq c_{n+j}\}}{1 \vee |\cS|} \Bigg]
    \$
    where $\EE_\cD$ is with respect to the randomness in $\xi$ conditional on $\cD := \cD_\labelled \cup \{(X_{n+j}, Y_{n+j}, c_{n+j})\}_{j=1}^m$. Note that $\ind\{p_j \leq s_j\}$ are deterministic conditional on the data. We denote $\cS^{(1)} := \{j: p_j \leq s_j\}$ as the first-step rejection set, and $\cH_0 := \{j: Y_{n+j} \leq c_{n+j}\}$ as the set of null hypotheses, both being deterministic given the data. Then, we have the decomposition
    \@ \label{eq:hete_decomp}
        \fdr(\cD) &= \sum_{k=1}^m \EE_\cD \Bigg[ \ind\{|\cS| = k\} \frac{\sum_{j \in \cS^{(1)} \cap \cH_0} \ind\{\xi\hat{R}_j \leq k \} }{k} \Bigg] \nonumber \\
        &= \sum_{j \in \cS^{(1)} \cap \cH_0} \sum_{k=1}^m \EE_\cD \Bigg[ \frac{ \ind\{\xi\hat{R}_j \leq k \} }{k} (\ind\{|\cS| \leq k\} - \ind\{|\cS| \leq k - 1\} ) \Bigg].
    \@
    For any $j \in [m]$, we define $\fdr(\cD, j)$ as the summand in~\eqref{eq:hete_decomp}.
    Since $|\cS| \leq |\cS^{(1)}|$ deterministically, we know 
    \$
        \ind\{\xi\hat{R}_j \leq k \}(\ind\{|\cS| \leq k\} - \ind\{|\cS| \leq k - 1\} ) = 0
    \$
    for any $k > |\cS^{(1)}| =: m_1$. The case $m_1=0$ is trivial, so assume $m_1\geq 1$. Hence,
    \$
        &\fdr(\cD, j) 
        = \sum_{k=1}^{m_1} \EE_\cD \Bigg[ \frac{1}{k}\ind\{\xi\hat{R}_j \leq k \} \cdot (\ind\{|\cS| \leq k\} - \ind\{|\cS| \leq k - 1\} ) \Bigg] \\
        &\leq \frac{\min\{1,m_1/\hat{R}_j\}}{m_1} + \sum_{k=1}^{m_1-1} \frac{1}{k} \EE_\cD \big[ \ind\{\xi\hat{R}_j \leq k \} \ind\{|\cS| \leq k\} \big] \\
        &\qquad \qquad \qquad \qquad  - \sum_{k=1}^{m_1-1} \frac{1}{k+1} \EE_\cD \big[ \ind\{\xi\hat{R}_j \leq k+1 \} \ind\{|\cS| \leq k\} \big] \\
        &\leq \frac{\min\{1,m_1/\hat{R}_j\}}{m_1} + \sum_{k=1}^{m_1-1} \Bigg\{ \frac{\PP_\cD(\xi \hat{R}_j \leq k)}{k} \PP_\cD\Big( |\cS| \leq k \Biggiven \xi \hat{R}_j \leq k \Big) \\
        &\qquad \qquad \qquad \qquad - \frac{\PP_\cD(\xi \hat{R}_j \leq k+1)}{k+1} \PP_\cD\Big( |\cS| \leq k \Biggiven \xi \hat{R}_j \leq k+1 \Big) \Bigg\}.
    \$
    Here, $\PP_\cD$ is with respect to the conditional distribution of $\xi$ given $\cD$. Now, we would rely on the \emph{positive regression dependence on a subset} (PRDS) property~\citep{benjamini2001control} of $\xi \hat{R}_j$ as follows to control $\fdr(\cD, j)$. For completeness, we include the definition of PRDS below. Lemma~\ref{lemma:b1} is the same as Lemma C.2 of~\cite{jin2023model}, hence we omit the proof here.

    \begin{definition}[PRDS]
\label{def:prds}
A random vector $X=(X_1,\dots,X_m)$ is PRDS on a subset $\cI$ if for any $i\in \cI$ and any increasing set $D$, 
the probability $\PP(X\in D\given X_i=x)$ is increasing in $x$. 
\end{definition}

    \begin{lemma}[Lemma C.2 of~\cite{jin2023model}] \label{lemma:b1}
        Let $a_1, \dots, a_k \in \mathbb{R}$ be nonnegative fixed constants, and let $\xi \sim \textnormal{Unif}([0,1])$. Then, the random variables $\{a_1\xi, \dots, a_{j-1}\xi, a_{j+1}\xi, \dots, a_k\xi\}$ are PRDS in $a_j\xi$ for any $j\in [k]$.
    \end{lemma}

    By Lemma~\ref{lemma:b1}, we see that
    \$
        \PP_\cD\Big( |\cS| \leq k \Biggiven \xi \hat{R}_j \leq k \Big) \leq \PP_\cD\Big( |\cS| \leq k \Biggiven \xi \hat{R}_j \leq k+1 \Big).
    \$
    Since $\xi$ is independent from $\hat{R}_j$, 
    \$
        &\fdr(\cD, j) \\ 
        &\leq \frac{\min\{1,m_1/\hat{R}_j\}}{m_1} + \sum_{k=1}^{m_1-1} \bigg\{ \frac{\PP_\cD(\xi \hat{R}_j \leq k)}{k} - \frac{\PP_\cD(\xi \hat{R}_j \leq k+1)}{k+1}  \bigg\} \PP_\cD\Big( |\cS| \leq k \Biggiven \xi \hat{R}_j \leq k \Big) \\
        &= \frac{\min\{1,m_1/\hat{R}_j\}}{m_1} + \sum_{k=1}^{m_1-1} \bigg\{ \frac{\min\{1,k/\hat{R}_j\} }{k} - \frac{\min \{1,(k+1)/\hat{R}_j\} }{k+1}  \bigg\} \PP_\cD\Big( |\cS| \leq k \Biggiven \xi \hat{R}_j \leq k \Big),
    \$
    Since $\frac{\min\{1,k/\hat R_j\}}{k}
=\min\{1/k,1/\hat R_j\}$ is non-increasing in $k$, the coefficients in the last summation is non-negative. Bounding the conditional probabilities by $1$ gives  
    \$
        \fdr(\cD, j) &\leq 
        \frac{\min\{1,m_1/\hat R_j\}}{m_1}
+\sum_{k=1}^{m_1-1}
\left\{
\frac{\min\{1,k/\hat R_j\}}{k}
-
\frac{\min\{1,(k+1)/\hat R_j\}}{k+1}
\right\}\\
&=
\min\{1,1/\hat R_j\}
\le \frac{1}{\hat R_j}.
    \$
    Putting above bounds together, we have
    \$
        \fdr(\cD) \leq \sum_{j \in \cS^{(1)} \cap \cH_0} \frac{1}{\hat{R}_j} = \sum_{j=1}^m \frac{\ind\{p_j \leq s_j, Y_{n+j} \leq c_{n+j}\}}{\hat{R}_j}.
    \$
    Finally, marginalizing over the randomness in $\cD$ concludes the proof of the homogeneous case.

    \vspace{0.5em}
    \noindent \textbf{Case 3: Deterministic pruning}. In this case, we have
    \$
        \cS = \{j: p_j \leq s_j, \hat{R}_j \leq r^*\} \quad \text{and} \quad r^* = |\cS|.
    \$
    by Proposition~\ref{prop:Rj}. Therefore, for any $j \in \cS$, we have $\hat{R}_j \leq |\cS|$, and
    \$
        \fdr =
        \EE \Bigg[ \frac{\sum_{j=1}^{m} \ind\{j \in \cS, Y_{n+j} \leq c_{n+j}\}}{1 \vee |\cS|} \Bigg] 
        &\leq \EE \Bigg[ \frac{\sum_{j=1}^{m} \ind\{j \in \cS, Y_{n+j} \leq c_{n+j}\}}{\hat{R}_j} \Bigg] \\
        &\leq \EE \Bigg[ \frac{\sum_{j=1}^{m} \ind\{p_j \leq s_j, Y_{n+j} \leq c_{n+j}\}}{\hat{R}_j} \Bigg].
    \$
    This concludes the proof of the deterministic case and hence also the proof of Lemma~\ref{lem:fdr_decomp}.
\end{proof}
 
\subsection{Proof of Proposition~\ref{prop:power_msel}}
\label{app:subsec_proof_power_msel}

\begin{proof}[Proof of Proposition~\ref{prop:power_msel}]
For each $k\in [K]$ and $\ell\in [m]$, we define 
    \$
    p_\ell^*(k) = F(\hat{V}_{n+\ell}(k);k),
    \$
    where $F(s;k) = \PP(v(X_i,Y_i;k)\leq s)$ for any $s\in \RR$ and $k\in [K]$, and recall that 
    \$
    G(t;k) = \PP\Big( F( v(X_i,c_i;k) ;k) \leq t \Big),\quad t\in [0,1]. 
    \$
    
    We first characterize the selection performance of SCS with each conformity score function. By definition and using the stopping-time interpretation of the BH procedure in~\cite{storey2004strong}, we have $\cS_{k,\textnormal{SCS}} = \{j\in [m]\colon p_j(k)\leq \hat\tau(k)\}$, where (since $n_2=n$)
    \$
    p_j(k) = \frac{1+\sum_{i=1}^n \ind\{V_i(k) \leq \hat{V}_{n+j}(k) \}}{n+1},
    \$
    and 
    \$
    \hat\tau(k) = \sup\bigg\{ t\in [0,1]\colon \frac{mt}{1\vee\sum_{j=1}^m \ind\{p_j (k) \leq t\}} \leq q\bigg\}.
    \$
    The uniform law of large numbers implies that as $m,n\to \infty$,
    \$
    \sup_{j\in [m], k\in [K]} \big| p_j(k) - p_j^*(k) \big| ~\asto ~ 0, 
    \$
    and since $\{p_j^*(k)\}_{j=1}^m$ are i.i.d.~random variables, applying the uniform law of large numbers to them implies 
    \$
    \sup_{t\in [0,1]} \bigg| \frac{1\vee\sum_{j=1}^m \ind\{p_j (k) \leq t\}}{m} - G(t;k)  \bigg| ~\asto~ 0. 
    \$
    Recall that $t^*_k = \sup\{t\colon t/G(t;k)\leq q\}$. Due to the continuity of $G(t;k)$ near $t_k^*$ and the fact that there exists $t'\in (t_k^*-\epsilon, t_k^*)$ with $t'/G(t';k)<q$ for any sufficiently small $\epsilon>0$, we know $t_k^* = q \cdot G(t_k^*;k)$, and 
    \$
    \sup_{k} \big| \hat\tau(k) - t_k^* \big| ~\asto ~ 0. 
    \$
    The uniform law of large numbers thus implies 
    \$
    \sup_{j\in [m]}\sup_{k\in [K]} \big||\cS_{k,\textnormal{SCS}}|/m - G(t_k^*; k) \big| ~ \asto ~ 0. 
    \$
    Since the maximizer $k^*\in [K]$ of $G(t_k^*;k)$ is unique, taking any constant $\epsilon < (G(t_{k^*}^*;k^*) - \max_{k\neq k^*} G(t_k^*,k) )/2$, the above a.s.~convergence implies that 
    \$
    \PP\big( \argmax_{k\in[K]} |\cS_{k,\textnormal{SCS}}| = k^* \big) \geq 
    \PP\Big( & |\cS_{k^*,\textnormal{SCS}}|/m > G(t_{k^*}^*; k^*) - \epsilon, \\ 
    &\max_{ k\neq k^*}  |\cS_{k,\textnormal{SCS}}|/m < \max_{k\neq k^*} G(t_k^*;k)+\epsilon\Big) \to 1
    \$  
    for any subsequence of $(m,n)$ as $m,n\to \infty$.  
    Also, defining $\hat{k}=\argmax_{k\in [K]}|\cS_{k^*,\textnormal{SCS}}|$, the above analysis implies 
    \@\label{eq:max_size_power_asymp}
    \textnormal{Power}(\cS_{\hat{k},\textnormal{SCS}}) &=  \frac{1}{m}\sum_{j=1}^m \ind\{p_j(\hat{k}) \leq \hat\tau(\hat{k}),~ Y_{n+j}> c_{n+j}\} \bigg] \notag \\ 
    &\stackrel{P}{\to} \PP \Big( F( \hat{V}_{n+j}(k^*);k^*) \leq t_k^*,~ Y_{n+j}> c_{n+j}  \Big).
    \@

    We then proceed to characterize the proxy selection sets $\cS_j(k)$, the individual thresholds $\hat{R}_j = \max_{k\in[K]} |\cS_j(k)|$, and the final selection set of OptCS-MSel. By definition and using the stopping-time interpretation of the BH procedure in~\cite{storey2004strong}, we have $\cS_j(k) = \{\ell\in[m]\colon\allowbreak \tilde{p}_\ell^{(j)} \leq \hat\tau_j(k)\}$, where 
    \$
    \hat\tau_j(k) = \sup\bigg\{t\in [0,1]\colon \frac{mt}{1+\sum_{\ell\neq j}\ind\{\tilde{p}_\ell^{(j)}\leq t\}} \leq q\bigg\}.
    \$
    Recalling the definition (since $n_2=n$)
    \$
  \tilde{p}_\ell^{(j)} (k) = \frac{\sum_{i=1}^{n} \ind\{V_i(k) \leq \hat{V}_{n+\ell}(k)\}+ \ind\{\hat{V}_{n+j}(k)\leq \hat{V}_{n+\ell}(k)\}}{n+1}.
    \$ 
    The uniform law of large numbers implies that 
    \$
    \sup_{k\in [K]}\sup_{j\in[m], \ell\in [m]} \big| \tilde{p}_\ell^{(j)}(k) - p_\ell^*(k) \big| ~\asto ~ 0. 
    \$
    Since $\{p_\ell^*(k)\}_{\ell=1}^m$ are i.i.d.~random variables, the above convergence, together with the uniform law of large numbers, implies that 
    \$
    \sup_{t\in [0,1]} \sup_{j\in [m]} \bigg| \frac{1+ \sum_{\ell\neq j} \ind\{\tilde{p}_\ell^{(j)}\leq t\} }{m}  - G(t;k) \bigg| ~\asto ~ 0. 
    \$
    With similar arguments as in our analysis of $\cS_{k,\textnormal{SCS}}$, we thus know that  
    $\sup_{k\in [K],j\in [m]}  \big|\hat\tau_j(k) - t_k^*\big| ~\asto ~ 0$.  
    The uniform law of large numbers thus implies 
    \$
    \sup_{j\in [m]}\sup_{k\in [K]} \big||\cS_j(k)|/m - G(t_k^*; k) \big| ~ \asto ~ 0. 
    \$
    Since the maximizer $k^*\in [K]$ of $G(t_k^*;k)$ is unique, taking any constant $\epsilon < (G(t_{k^*}^*;k^*) - \max_{k\neq k^*} G(t_k^*,k) )/2$, the above a.s.~convergence implies that 
    \$
    \PP\big( \,\hat{k}_j = k^*,~\forall j\in [m] \big) \geq 
    \PP\Big( &\min_{j\in [m]} |\cS_j(k^*)|/m > G(t_{k^*}^*; k^*) - \epsilon, \\ 
    &\max_{j\in [m],k\neq k^*} |\cS_j(k)|/m < \max_{k\neq k^*} G(t_k^*;k)+\epsilon\Big) \to 1
    \$  
    for any subsequence of $(m,n)$ as $m,n\to \infty$. On the other hand, the above analysis also implies that  
    $\sup_{j\in [m]} \big|\hat{R}_j/m - G(t_{k^*}^*;k^*) \big| ~\asto ~ 0$. 
    Recall the definition of the OptCS p-values
    \$
    p_j = \frac{1 + \sum_{i =1}^{n} \ind\{v(X_i,Y_i;\hat{k}_j) \leq v(X_{n+j},c_{n+j};\hat{k}_j)\}}{n + 1}.
    \$
    Based on the above analysis, it is not difficult to show that 
    \@\label{eq:conv_opt_pval}
    \sup_{j\in [m]} \big| p_j - p_j^*(k^*) \big| ~\asto ~ 0,
    \@
    and
    \@\label{eq:conv_opt_pval_cdf}
    \sup_{t\in[0,1]} \bigg| \frac{1}{m}\sum_{j=1}^m \ind {\{  p_j\leq t \}} - G(t;k^*)  \bigg|~\asto ~0. 
    \@
    We now proceed to analyze the final selection set of OptCS with homogeneous pruning, following similar ideas as in the proof of Proposition 3.4 in~\cite{jin2023model}. The homogeneous pruning yields 
    $
\cS_{\textnormal{OptCS}} = \big\{ j\colon  p_j \leq q\hat{R}_j/m, ~ \xi   \hat{R}_{j}   \leq  r^*  \big\}, 
$ 
where 
\$
r^* := \max\bigg\{ r\colon  \frac{1}{m}\sum_{j=1}^m \ind {\{  p_j\leq q\hat{R}_j/m, \, \xi   \hat{R}_j/m \leq r/m \}}   \geq r/m \bigg\}.
\$ 
Equivalently, we have $\cS_{\textnormal{OptCS}} = \big\{ j\colon  p_j \leq q\hat{R}_j/m, ~ \xi   \hat{R}_{j}/m   \leq r^{**} \big\}$, where 
\$
r^{**} = \max\bigg\{ r\in [0,1] \colon  \frac{1}{m}\sum_{j=1}^m \ind {\{  p_j\leq q\hat{R}_j/m, \, \xi   \hat{R}_j/m \leq r \}}   \geq r  \bigg\}.
\$
Fixing any $\delta>0$, we denote the event 
\$
\cE(\delta) := \Bigg\{ \sup_{t\in[0,1]} \bigg| \frac{1}{m}\sum_{j=1}^m \ind {\{  p_j\leq t \}} - G(t;k^*)  \bigg|> \delta\Bigg\} \bigcup \Bigg\{ \sup_{j\in [m]} \big|\hat{R}_j/m - G(t_{k^*}^*;k^*) \big| >\delta \Bigg\}.
\$
The above analysis implies that $\PP(\cE(\delta))\to 0$ for any subsequence of $(m,n)$ as $m,n\to \infty$. Also recall that $t_{k^*}^*=q\cdot G(t_{k^*}^*;k^*)$. Thus, on the event $\cE(\delta)^c$, for any $r\in (0,1)$, it holds that 
\$
&\frac{1}{m}\sum_{j=1}^m \ind {\{  p_j\leq q\hat{R}_j/m, \, \xi   \hat{R}_j/m \leq r \}} \\
& \leq 
\frac{1}{m}\sum_{j=1}^m \ind  \Big\{  p_j\leq q  G(t_{k^*}^*;k^*)+q\delta, \, \xi    \big[G(t_{k^*}^*;k^*) - \delta\big] \leq r  \Big\}  \\
& =
\frac{1}{m}\sum_{j=1}^m \ind  \Big\{  p_j\leq  t_{k^*}^* +q\delta, \, \xi    \big[G(t_{k^*}^*;k^*) - \delta\big] \leq r  \Big\} \\ 
&\leq \Big( G(t_{k^*}^* +q\delta ;k^*) + \delta \Big)  \ind\Big\{ \xi    \big[G(t_{k^*}^*;k^*) - \delta\big] \leq r\Big\} \leq G(t_{k^*}^* +q\delta ;k^*) + \delta ,
\$
and on the other hand, also 
\$
&\frac{1}{m}\sum_{j=1}^m \ind {\{  p_j\leq q\hat{R}_j/m, \, \xi   \hat{R}_j/m \leq r \}} \\
& \geq 
\frac{1}{m}\sum_{j=1}^m \ind  \Big\{  p_j\leq q  G(t_{k^*}^*;k^*)-q\delta, \, \xi    \big[G(t_{k^*}^*;k^*) + \delta\big] \leq r  \Big\}  \\
& =
\frac{1}{m}\sum_{j=1}^m \ind  \Big\{  p_j\leq  t_{k^*}^* -q\delta, \, \xi    \big[G(t_{k^*}^*;k^*) + \delta\big] \leq r  \Big\} \\ 
&\geq \Big( G(t_{k^*}^* -q\delta ;k^*) - \delta \Big)  \ind\Big\{ \xi    \big[G(t_{k^*}^*;k^*) + \delta\big] \leq r\Big\}.
\$
These imply that 
\$
\PP\Big( r^{**} \geq G(t_{k^*}^* -q\delta ;k^*) -\delta  \Big) \geq \PP\Big( \cE(\delta)^c  ~ \text{and} ~\xi    \big[G(t_{k^*}^*;k^*) + \delta\big] \leq G(t_{k^*}^* -q\delta ;k^*) -\delta  \Big) 
\$
and 
\$
\PP\Big( r^{**} \leq G(t_{k^*}^* +q\delta ;k^*) + \delta  \Big) \geq 1- \PP(\cE(\delta)). 
\$
Taking $\delta\downarrow 0$, we know that 
$r^{**} \stackrel{P}{\to} G(t_{k^*}^*;k^*)$. Thus, with similar ideas we can show that 
\@\label{eq:power-opt-msel-asymp}
\text{Power}(\cS_{\textnormal{OptCS}}) &= \frac{1}{m}\sum_{j=1}^m \ind\{p_j \leq q \hat{R}_j/m, \xi \hat{R}_j /m \leq r^{**}, Y_{n+j}>c_{n+j}\} \notag \\ 
& \stackrel{P}{\to} \PP\Big( p_j^*(k^*) \leq t_k^*,~ Y_{n+j}> c_{n+j}  \Big) \notag \\ 
& = \PP\Big( F( \hat{V}_{n+j}(k^*);k^*) \leq t_k^*,~ Y_{n+j}> c_{n+j}  \Big).
\@
Combining~\eqref{eq:max_size_power_asymp} and~\eqref{eq:power-opt-msel-asymp}, we thus have $\text{Power}(\cS_{\textnormal{OptCS}}) - \text{Power}(\cS_{\hat{k},\textnormal{SCS}}) = o_P(1)$. We thus conclude the proof of Proposition~\ref{prop:power_msel}.
\end{proof}

\section{Additional details and results for experiments}

\subsection{Simulation setups for Section~\ref{subsec:simu_score_sel}} \label{subsec:simulation_setup}

In both settings, the thresholds $c_{n+j}$ are set to $c_{n+j} \equiv c = 0$.
For Liang's settings, the configuration of $X_i$, $\theta_i$ and $\epsilon_i$ for $i \in [d]$ are summarized in Table~\ref{tab:liang_settings}. The feature dimension is taken to be $d = 300$ and the noise level is taken to be $\sigma = 3$. The degree of freedom used in $t$-distribution is $\nu = 3$. The regression function is defined as $\mu(x) = X_i^\top\theta_i$.

\renewcommand{\arraystretch}{1.5}
\begin{table}[htbp]
    \centering
    \small
    \begin{tabular}{c|c|c|c}
        \hline
        Setting & $X_i$ & $\theta_i$ & $\epsilon_i$ \\
        \hline
        1 & $\cN(0, \mathbf{I}_d)$ & $\ind\{i \text{ mod } 20 = 0\}$ & $\sigma \cdot \cN(0,1)$ \\
        \hline
        2 & $\cN(0, \mathbf{I}_d)$& $\ind\{i \text{ mod } 20 = 0\}$ & $\sigma \cdot t_\nu(1)$ \\
        \hline
        3 & $\cN(0, \mathbf{I}_d)$ & $1/d$ & $\sigma \cdot \cN(0,1/d)$ \\
        \hline
        4 & $t_\nu(0, \mathbf{I}_d)$ & $\ind\{i \text{ mod } 20 = 0\}$ & $\sigma \cdot \cN(0,1)$ \\
        \hline
    \end{tabular}
    \captionsetup{font=small}
    \caption{Details of the four data generating processes used in Liang's settings.}
    \label{tab:liang_settings}
\end{table}

In Liang's settings, we consider eleven combinations of models and scores for the model/score selection process. Following the setups in~\cite{liang2024conformal}, we fit nine linear conditional quantile regression models  at quantiles $\alpha \in\{0.1,0.2,\dots,0.9\}$ using a randomly selected subset of features of size $d / 10$. For these models, the conformity score used is $M\cdot \ind\{Y > c\} - \hat{q}_{\alpha}(x)$ with $M=100$, where $\hat{q}_{\alpha}$ is the fitted quantile regression model at level $\alpha$. The remaining two model options are a shared linear model $\hat\mu(\cdot)$ estimating the conditional mean of $Y$ given $X$, fitted using $d / 5$ randomly selected features, with two choices of conformity scores: the standard clipped score $M\cdot\ind\{Y > c\} - \hat\mu(x)$ and a studentized version $M\cdot\ind\{Y > c\} - \hat\mu(x) / \hat\sigma(x)$ with $M=1000$, where $\hat\sigma$ is a linear conditional standard error model estimating the absolute prediction error $\EE\big[ |Y-\hat\mu(x)| \biggiven X=x \big]$. The conditional standard deviation model $\hat\sigma$ is trained on an additional dataset of 100 observations (since all methods have access to the same set of pre-trained model, the comparison is still fair). All the models are fitted using functions in the \texttt{scikit-learn} Python package. 
In this setting, the quality of models mainly depends on how many ``true'' features are included in the regression modeling. 

The details for Jin's settings are summarized in Table~\ref{tab:jin_settings}. In Section~\ref{subsec:simu_score_sel}, we set the feature dimension at $d = 20$, and the noise level at $\sigma = 1$. Settings $(1,3)$ and $(2,4)$ have the same regression function, yet with different heterogeneity structure in the noise distribution. 

\renewcommand{\arraystretch}{2}
\begin{table}[htbp]
    \centering
    \small
    \begin{tabular}{c|c|c|c}
        \hline
        Setting & $X_i$ & $\mu(\cdot) $ & $\epsilon_i$ \\
        \hline
        1 & $\text{Unif}[-1, 1]^{d}$ & \makecell{$x_1 \ind\{x_2 > 0, x_4 > 0.5\} \cdot (0.25 + x_4)$ \\ \hspace{0.5cm} $+x_1 \ind\{x_2 \leq 0, x_4 \leq -0.5\} \cdot (x_4 - 0.25)$} & $\sigma \cdot \cN(0,1)$ \\
        \hline
        2 & $\text{Unif}[-1, 1]^{d}$ & $2(x_1x_2 + e^{x_4} - 1)$ & $\frac{3}{2}\sigma \cdot \cN(0,1)$ \\
        \hline
        3 & $\text{Unif}[-1, 1]^{d}$ & \makecell{$x_1 \ind\{x_2 > 0, x_4 > 0.5\} \cdot (0.25 + x_4)$ \\ \hspace{0.5cm} $+x_1 \ind\{x_2 \leq 0, x_4 \leq -0.5\} \cdot (x_4 - 0.25)$} & $\sigma \cdot (5.5 - |\mu(x)|)/2$ \\
        \hline
        4 & $\text{Unif}[-1, 1]^{d}$ & $2(x_1x_2 + e^{x_4} - 1)$ & $\sigma \cdot (5.5 - |\mu(x)|)/2$ \\
        \hline
    \end{tabular}
    \captionsetup{font=small}
    \caption{Details of the four data generating processes used in Jin's settings.}
    \label{tab:jin_settings}
\end{table}

For Jin's settings, we consider $24$ options for the combination of prediction models and scores. We fitted random forest (using the \texttt{quantile-forest} Python package) and gradient boosting (using the \texttt{scikit-learn} Python package) quantile regressors at nine quantile levels $\alpha\in \{0.1,0.2,\dots,0.9\}$, yielding $18$ models. For these models, the conformity score is defined as $v(x,y) = M\cdot\ind\{Y > c\} - \hat{q}_{\alpha}(x)$ with $M=1000$. Additionally, we fit  a random forest conditional mean model (with \texttt{scikit-learn} Python package), $\hat\mu$, which is shared across the remaining six model options. These six options use conformity scores of the form $M\cdot\ind\{Y > c\} - \hat\mu(x) / \hat\sigma(x)$ with $M=1000$ and $6$ options for the function $\hat\sigma(\cdot)$: a constant value of 1 (equivalent to the standard clipped score) and five conditional mean squared error models, including gradient boosting, random forest, support vector machine, Lasso, and Ridge regression, all using the \texttt{scikit-learn} package. 
The model $\hat\sigma$ is trained on an additional dataset of 100 observations.

\subsection{Additional results for Section~\ref{subsec:simu_score_sel}} \label{subsec:more_simu}

\begin{figure}[H]
    \centering
    \small
    \includegraphics[width=0.9\linewidth]{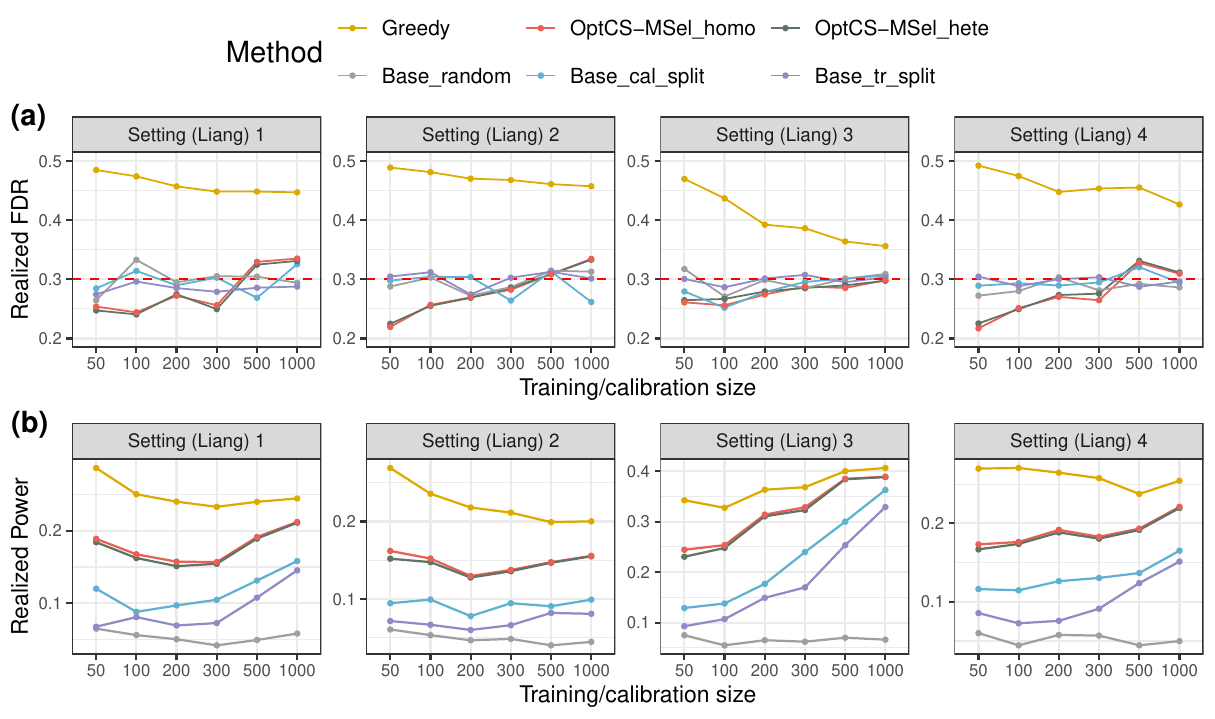}
    \captionsetup{font=small}
    \caption{Additional results for Section~\ref{subsec:simu_score_sel} when varying sample sizes for Liang's settings. Each subplot corresponds to one data generating process, with $x$-axis being $n=n_\train=n_\calib$, while fixing  $q=0.3$.}
    \label{fig:simu_sel_liang_scaleup}
\end{figure}

\begin{figure}[H]
    \centering
    \small
    \includegraphics[width=0.9\linewidth]{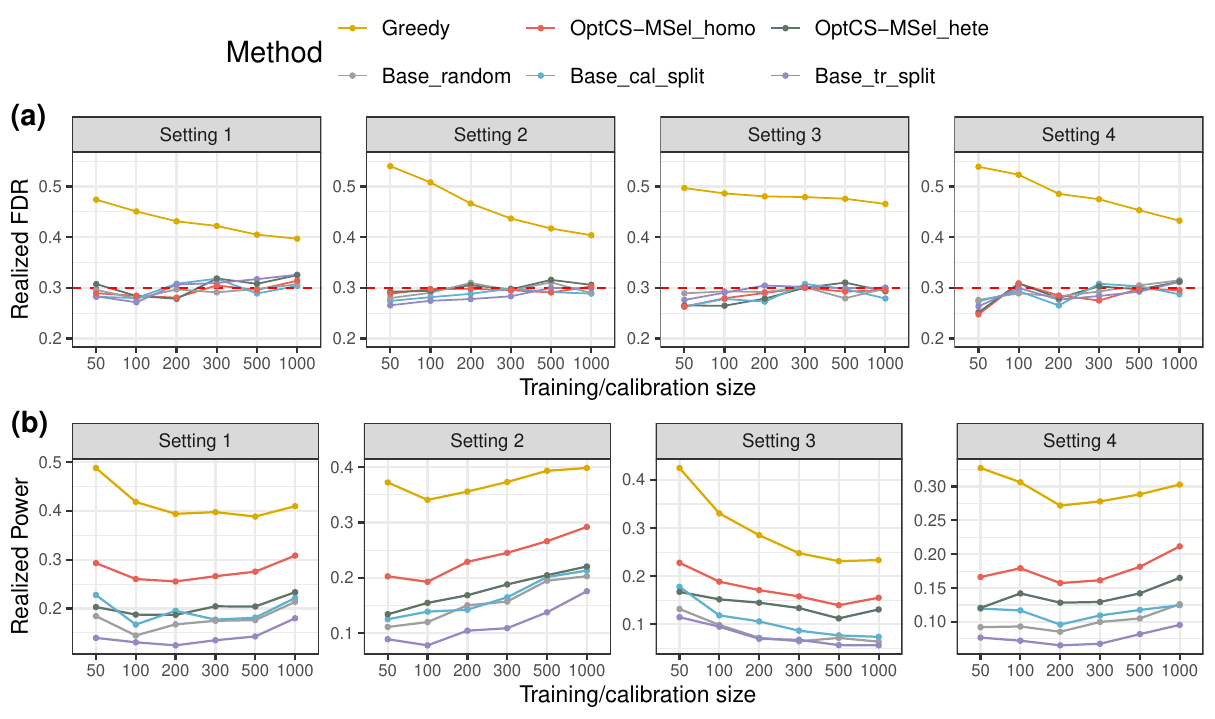}
    \captionsetup{font=small}
    \caption{Additional results for experiments in Section~\ref{subsec:simu_score_sel} when varying sample sizes for Jin's settings. Details are otherwise the same as Figure~\ref{fig:simu_sel_liang_scaleup}.}
    \label{fig:simu_sel_jin_scaleup}
\end{figure}
 
In this part, we include additional simulation results that show the performance of OptCS-MSel and baselines when the training and calibration sample sizes vary. 
Figure~\ref{fig:simu_sel_liang_scaleup} presents the empirical FDR and power for Liang's settings, and Figure~\ref{fig:simu_sel_jin_scaleup} presents these results for Jin's settings. 
 
In addition to the above sample-size experiments, we add a targeted
``Power versus Sample Size'' analysis for OptCS-MSel to evaluate the
finite-sample tradeoff between data reuse and the conservativeness introduced by pruning, fixing the trained models. 
We use Liang's settings, where the candidate scores have visibly different qualities
and the effect of model selection is pronounced. We fix \(n_1=100\) samples for training
the candidate scores and \(m=100\) test samples, and vary the number of calibration samples
\( 
n_2\in\{50,100,200,500,1000\}.
\)
We also vary the target FDR level
\(q\in\{0.2,0.25,\ldots,0.5\}\), and average over \(N=200\) runs.
We compare OptCS-MSel with homogeneous pruning,
\texttt{Base\_cal\_split}, the data-splitting baseline that uses part of the calibration data for
model selection and an independent calibration fold for final SCS/BH, and
\texttt{Base\_random}, which randomly selects a candidate score.

The results are shown in Figure~\ref{fig:msel_n2}. Across the displayed settings,
OptCS-MSel empirically controls FDR and achieves higher power than both baselines. The
gap relative to \texttt{Base\_random} demonstrates the benefit of targeted score selection, while
the gap relative to \texttt{Base\_cal\_split} shows that the gain from avoiding an additional
calibration split can outweigh the finite-sample conservativeness introduced by pruning.
The advantage of OptCS-MSel persists across the range of \(n_2\), showing the advantage of reusing calibration data for targeted model selection.

\begin{figure}
    \centering
    \includegraphics[width=\linewidth]{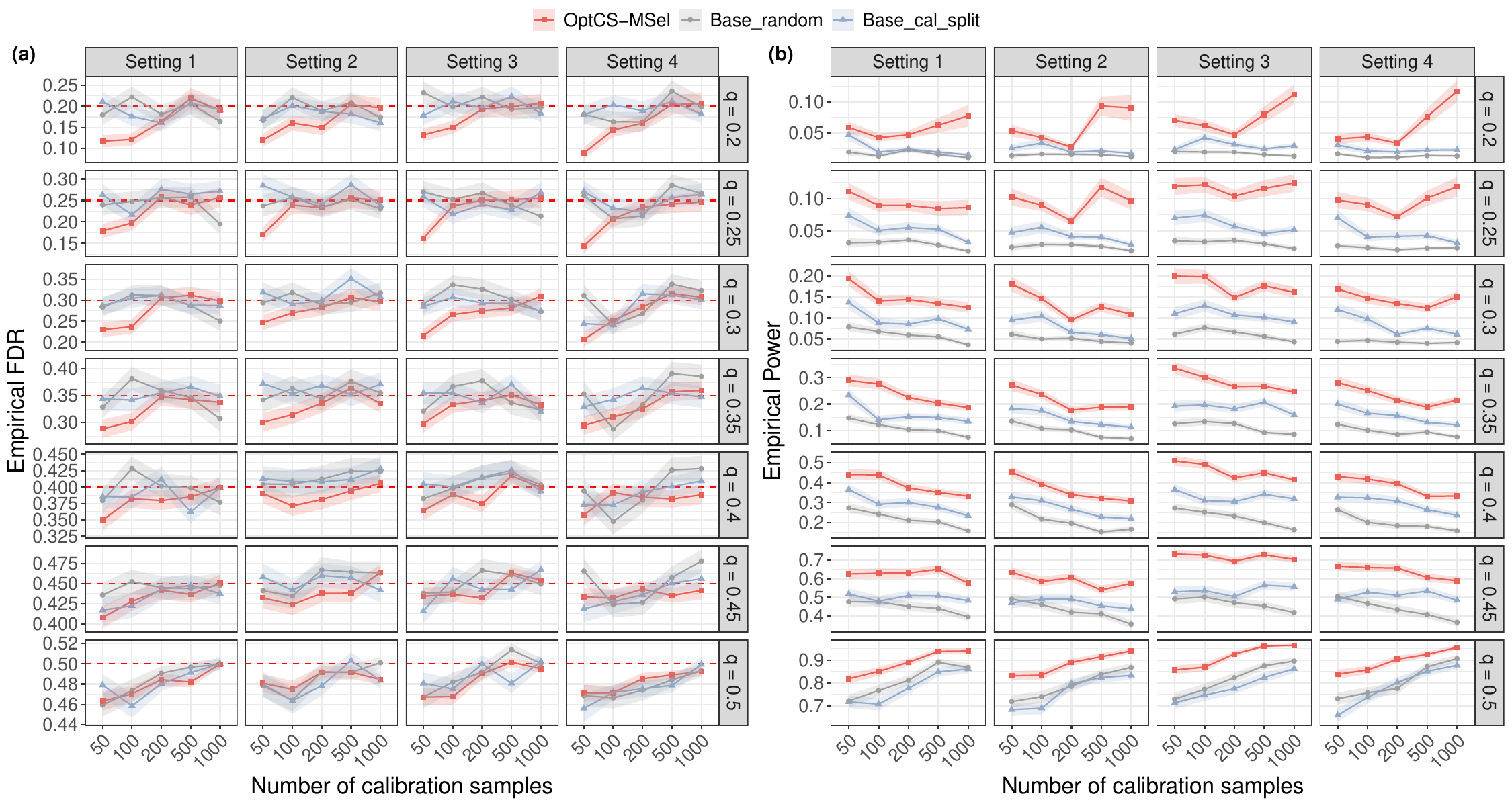}
    \caption{Power vs.~Sample Size analysis for OptCS-MSel. (a) Realized FDR and (b) power of OptCS-MSel and sample splitting baseline \texttt{Base\_cal\_split} and non-model-selection baseline \texttt{Base\_random} for varying target FDR levels $q\in\{0.2, 0.25,\dots,0.5\}$ (rows) and Liang's settings (column), averaged over $N=200$ runs. The $x$-axis is the number of calibration samples $n_2$. We fix $n_1=100$ for consistent model training and $m=100$ for evaluation.}
    \label{fig:msel_n2}
\end{figure}

\subsection{Standard deviation for Section~\ref{subsec:simu_score_sel}}
\label{app:subsec_simu_std}

\begin{figure}[H]
    \centering
    \small
    \includegraphics[width=0.85\linewidth]{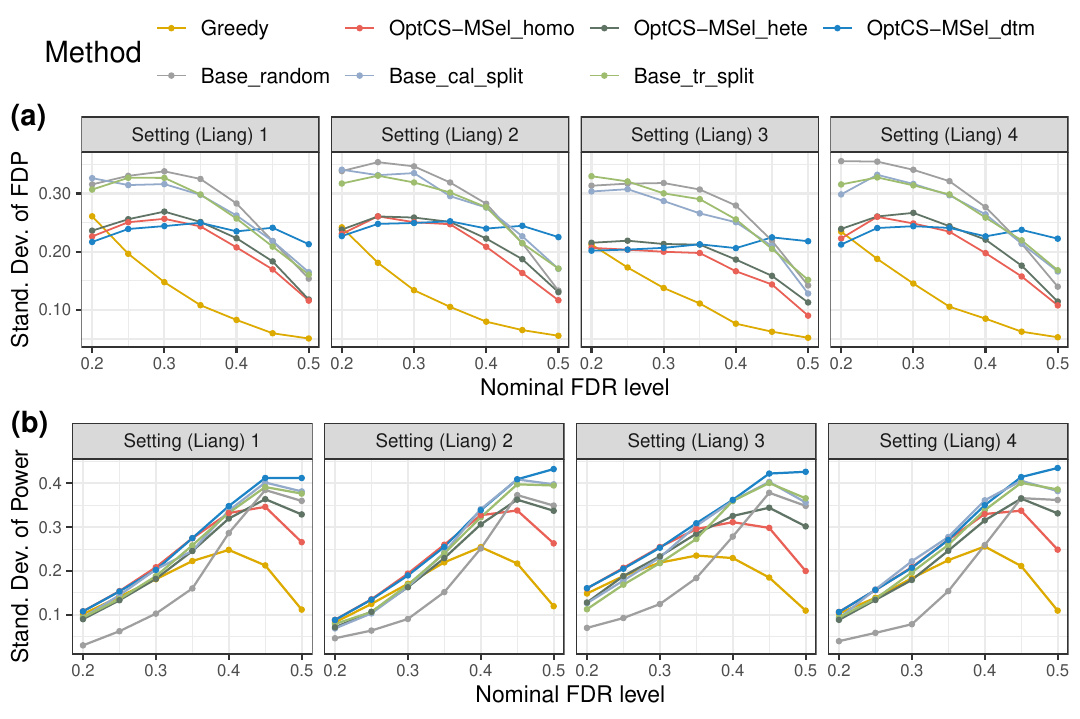} 
    \captionsetup{font=small}
    \caption{Standard deviation of FDP (a) and power (b) in Liang's settings. Details are otherwise the same as Figure~\ref{fig:simu_score_sel_Liang}.}
    \label{fig:simu_std_liang}
\end{figure}

\begin{figure}[H] 
    \centering
    \small
    \includegraphics[width=0.85\linewidth]{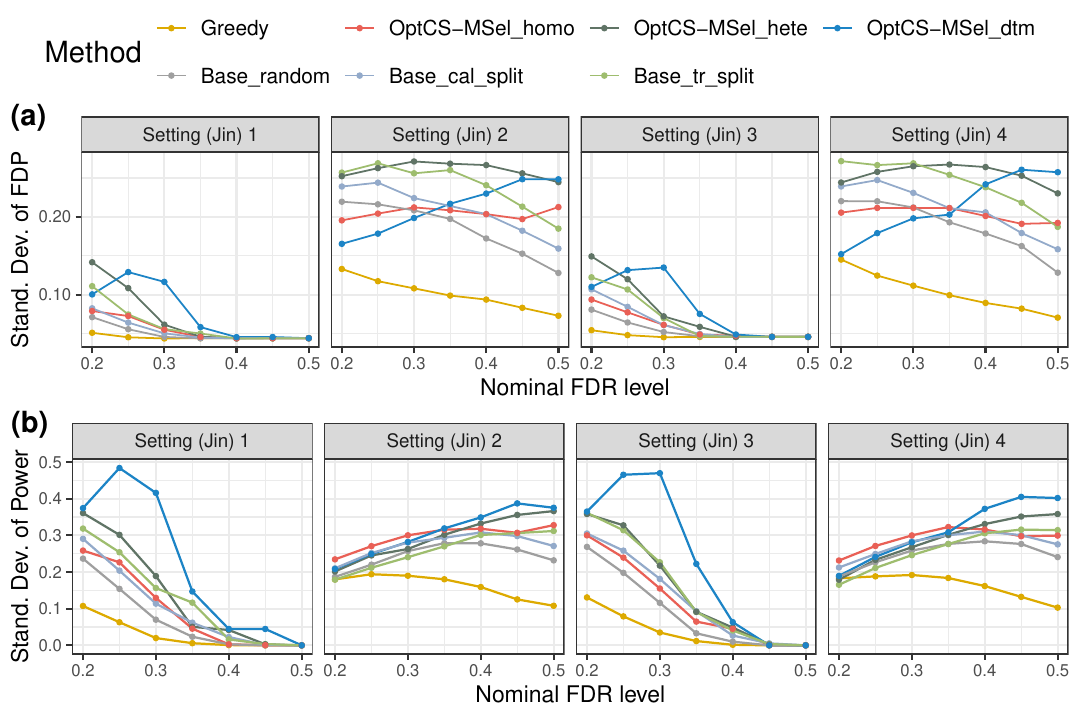}
    \captionsetup{font=small}
    \caption{Standard deviation of FDP (a) and power (b) in Jin's settings. Details are otherwise the same as Figure~\ref{fig:simu_score_sel_Jin}.}
    \label{fig:simu_std_jin}
\end{figure}

 In this subsection, we discuss the stability of OptCS-MSel with different pruning strategies (including the deterministic pruning).  Figures~\ref{fig:simu_std_liang} and~\ref{fig:simu_std_jin} show the standard deviation of different methods in Liang's settings and Jin's settings, respectively. We observe that among the three variants, the homogeneous pruning has the smallest (marginal) standard deviation. It appears that the random pruning ``smooths'' the selection result, even though it adds randomness in the process. In contrast, the deterministic pruning tends to have the highest standard deviation in most cases. 

\subsection{Simulation setups for Section~\ref{subsec:simu_loo}} \label{subsec:loo_setup}

In Section~\ref{subsec:simu_loo}, we use the four Jin's settings detailed below, fixing  $n = 500, m = 100$. The fixed model is $v(x,y) = My - \hat\mu(x)$, where $\hat\mu\colon \cX\to [0,1]$ is fitted by SVR using  \texttt{scikit-learn}, and $M=100$. 
The details for Jin's settings used in Section~\ref{subsec:simu_loo} can be found in Table~\ref{tab:loo_jin_settings}. We adapt the settings from~\cite{jin2023selection} to a classification setting,  where the binary labels are defined via $Y_i = \ind\{\mu(X_i) + \epsilon_i >0\}$.  
The feature dimension is $d=10$ to limit the computation costs, and the noise level is $\sigma=0.5$.  
In Table~\ref{tab:loo_jin_settings}, the settings 1 and 3 slightly modify the settings 1 and 3 in Table~\ref{tab:jin_settings} to make sure that increasing the training sample size leads to more powerful models at least in the split conformal selection procedure. 
We note that in the original settings 1 and 3 in Table~\ref{tab:jin_settings}, increasing the sample size reduces selection power in \basename potentially due to the specific choice of function classes, yet this does not affect the results in Section~\ref{subsec:simu_score_sel} since the sample sizes are fixed there. 
We will thus use these four settings for both Section~\ref{subsec:simu_loo} and Section~\ref{subsec:simu_full_msel}. 

\renewcommand{\arraystretch}{2}
\begin{table}[htbp]
    \centering
    \small
    \begin{tabular}{c|c|c|c}
        \hline
        Setting & $X_i$ & $\mu(\cdot) $ & $\epsilon_i$ \\
        \hline
        1 & $\text{Unif}[-1, 1]^{d}$ & $\frac{1}{2}\ind\{x_1x_2 > 0\} + \ind\{x_1x_2 \leq 0\} + x_4$ & $2\sigma \cdot \cN(0,1)$ \\
        \hline
        2 & $\text{Unif}[-1, 1]^{d}$ & $2(x_1x_2 + e^{x_4} - 1)$ &  $\frac{3}{2}\sigma \cdot \cN(0,1)$ \\
        \hline
        3 & $\text{Unif}[-1, 1]^{d}$ & $\frac{1}{2}\ind\{x_1x_2 > 0\} + \ind\{x_1x_2 \leq 0\} + x_4$ & $\sigma \cdot (5.5 - |\mu(x)|)/2$ \\
        \hline
        4 & $\text{Unif}[-1, 1]^{d}$ & $2(x_1x_2 + e^{x_4} - 1)$ & $\sigma \cdot (5.5 - |\mu(x)|)/3$ \\
        \hline
    \end{tabular}
    \captionsetup{font=small}
    \caption{Details of the four data generating processes used in Jin's settings for Section~\ref{subsec:simu_loo}.}
    \label{tab:loo_jin_settings}
\end{table}

The details of the two  ``partial'' baselines in our evaluation are as follows: 
\vspace{0.5em}
\begin{enumerate}[label=(\roman*).]
    \item \texttt{OptCS-Full\_random}: Our OptCS-Full procedure where the training process uses over-sampling, with a randomly selected model class (i.e., full-data training but no model selection). 
    \item \texttt{OptCS-Full\_split}: We first split the data into $\cD_\sel$ ($50\%$) and $\cD_\calib$ ($50\%$), then use $\cD_\sel$ to select a model class, and run our OptCS-Full procedure using $\cD_\calib$ as the ``full data'' with the selected model class. Inside $\cD_\sel$, we split the data into a ``calibration'' fold ($50\%$) and a ``test'' fold ($50\%$), run OptCS-Full with each model class, and select the one with the largest selection set. 
\end{enumerate} 

\subsection{Simulation setups for Section~\ref{subsec:simu_full_msel}}
\label{app:subsec_simu_full_msel}

In this section, we use the same data-generating processes as in Appendix~\ref{subsec:simulation_setup} for Liang's settings and Appendix~\ref{subsec:loo_setup} for Jin's settings. The thresholds $c_{n+j}$ are set to 0. To repeat the coefficients again, we use a feature dimension of $d=10$, and noise level $\sigma=0.5$ for Jin's settings. For Liang's settings, we use feature dimension $d=300$, noise level $\sigma=3$, and degree of freedom $\nu=3$ for the $t$-distributions involved in the data generation.

For Liang's settings, we consider 7 combinations of models and scores. Three linear conditional quantile models at quantile level $\alpha = 0.25, 0.5, 0.75$ were fitted using randomly selected subset of features of size $d/10$. Similar to setups in section~\ref{subsec:simulation_setup}, the conformity scores used here takes the form $v(x,y) = M\cdot\ind\{Y > c\} - \hat{q}_{\alpha}(x)$ for $M=1000$.
For the remaining 4 options, we fit two linear regression models with $d/10$ and $d/5$ randomly chosen features, and pair them with either a linear regression error model $\hat\sigma$ or without error estimation, i.e. $\hat\sigma \equiv 1$. In the case where an error model is employed, the set of all labeled data is first partitioned to two subsets for fitting the linear regression model (70\%) and the error model (30\%) in the first place, respectively. The conformity scores for these options are $M\cdot\ind\{Y > c\} - \hat\mu(x) / \hat\sigma(x)$ for $M=1000$. All these models are fitted using functions in the \texttt{scikit-learn} Python package.  
As we adopt two options with $d/5$ features as opposed to $d/10$, the model accuracies vary significantly in this setting. This explains the relatively big performance gap between the random baselines and other competing methods.

For Jin's settings, we consider 9 pairs of candidate model and scores. We fitted three random forest quantile regressions at level $\alpha = 0.25, 0.5, 0.75$ (from \texttt{quantile-forest} package), and again use conformity scores of the form $v(x,y) = M\cdot\ind\{Y > c\} - \hat{q}_{\alpha}(x)$. We also consider $3$ conditional mean models: support vector regression, Lasso and Ridge regression. Each of them is paired with either a conditional error model $\hat\sigma$ fitted with SVR or without error estimation (i.e., the clipped score). Again, the conformity scores for these options are $M\cdot\ind\{Y > c\} - \hat\mu(x) / \hat\sigma(x)$, and when an error model is used, a 70\%-30\% data split is done to train the conditional mean model and error model respectively.

\subsection{Datasets and models for Section~\ref{subsec:app_drug}}
\label{app:subsec_drug_model}

Each dataset is randomly split into training ($\cD_\train$, $60\%$), calibration ($\cD_\calib$, $20\%$) and test ($\cD_\test$, $20\%$) folds.   We use 8 state-of-the-art molecule embeddings: \texttt{DGL\_AttentiveFP}, \texttt{Morgan}, \texttt{CNN}, \texttt{rdkit\_2d\_normalized}, \texttt{DGL\_GCN}, \texttt{DGL\_NeuralFP}, \texttt{DGL\_GIN\_AttrMasking} and \texttt{DGL\_GIN\_ContextPred}. Based on them, we train 8 deep learning models using the DeepPurpose library~\citep{huang2020deeppurpose}. Each prediction model $\hat\mu$ is then paired with the clipped score  $M\ind\{Y > c\} - \hat\mu(x)$ for $M=1000$. For the first four tasks (except \texttt{davis}), we also fit 9 $\texttt{CNN}$ quantile regression models $\hat{q}_{\alpha}(\cdot)$ with $\alpha \in \{0.1,0.2,\dots,0.9\}$, paired with the clipped score $M\ind\{Y > c\} - \hat{q}_{\alpha}(x)$, $M=1000$. 

We set $c_{n+j}$ as the $0.7$-th quantile of all the labels in the data for the first four regression tasks, while for the drug-target-interaction prediction task on \texttt{davis} we vary $c_{n+j}$.  
For the drug-target-interaction task on the \texttt{davis} dataset, we use varying values of thresholds $c_{n+j}$ to test the performance of OptCS under such a scenario. Specifically, for each test compound, the threshold $c_{n+j}$ is defined as the 0.8-th quantile of the responses for training compounds sharing the same target sequence.

\subsection{Datasets and pre-trained model for Section~\ref{subsec:app_LLM}}
\label{app:subsec_data_LLM}

In Section~\ref{subsec:app_LLM}, we use the same subset (p10, p11 and p12 folders) of the MIMIC-CXR dataset~\citep{johnson2019mimic} as in~\cite{gui2024conformal}, which is accessed from the PhysioNet project page \url{https://physionet.org/content/mimic-cxr/2.0.0/} under the PhysioNet Credentialed Health Data License 1.5.0. In our experiments, we draw a subset of images in the test folder determined by the same split as \cite{gui2024conformal}. In this way, the randomness is purely from randomly splitting the data into labeled data and test samples. 

The foundation model for generating the radiology reports is the one fine-tuned in \cite{gui2024conformal}. We include the details here for completeness. Specifically, this vision-language model combines the Vision Transformer \texttt{google/vit-base-patch16-224-in21k} pre-trained on ImageNet-21k\footnote{\url{https://huggingface.co/google/vit-base-patch16-224-in21k}} as the image encoder and GPT2 as the text decoder.  
Each raw image is resized to $224\times 224$ pixels. 
The model  is fine-tuned on a hold-out dataset with a sample size of $43{,}300$ for $10$ epochs with a batch size of $8$, and other hyperparameters are set to default values. 
When generating reports, all the parameters are kept the same as the conformal alignment paper; we refer the readers to \cite[Appendix C.2]{gui2024conformal} for these details. 

\subsection{Detailed setup for Section~\ref{subsec:app_LLM}}
\label{app:subsec_detail_LLM}

We use exactly the same procedures as \cite{gui2024conformal} to compute $12$ features which (heuristically) measure the uncertainty of LLM-generated outputs:
\vspace{0.5em}
\begin{itemize}%
    \item \emph{Input uncertainty scores} (\texttt{Lexical\_Sim}, \texttt{Num\_Sets}, \texttt{SE}). Following~\cite{kuhn2023semantic}, we compute a set of features that measure the uncertainty of each LLM input through similarity among multiple answers. The features include lexical similarity (\texttt{Lexical\_Sim}), the \texttt{rouge-L} similarity among the  answers. In addition, we use a natural language inference (NLI) classifier to categorize the $M$ answers into semantic groups, and compute the number of semantic sets (\texttt{Num\_Sets}) and semantic entropy (\texttt{SE}).  Following \cite{kuhn2023semantic, lin2023generating}, 
    we use an off-the-shelf DeBERTa-large model \citep{he2020deberta} as the NLI predictor.
    \item \emph{Output confidence scores} (\texttt{EigV(J/E/C)}, \texttt{Deg(J/E/C)}, \texttt{Ecc(J/E/C)}). We also follow~\citep{lin2023generating} to compute features that measure the so-called output confidence: with $M$ generations, we compute the eigenvalues of the graph Laplacian (\texttt{EigV}), the pairwise distance of generations based on the degree matrix (\texttt{Deg}), and the Eccentricity (\texttt{Ecc}) which incorporates the embedding information of each generation. Note that each quantity is associated with a similarity measure; we follow the notations in \cite{lin2023generating} and use the suffix \texttt{J}/\texttt{E}/\texttt{C} to differentiate similarities based on the Jaccard metric, NLI prediction for the entailment class, and NLI prediction for the contradiction class, respectively. 
\end{itemize}
\vspace{0.5em}

For experiments with OptCS-Full, we consider three model classes, logistic regression and random forests from the \texttt{scikit-learn} Python library and XGBoost from the \texttt{xgboost} Python library, utilizing all the above measures as the features. 

The first setup in experiments with OptCS-Full-MSel is as follows. 
We treat the task as a regression problem, and fit 3 linear quantile regressors and 3 random forest quantile regressors at quantile level $\alpha = 0.25, 0.5, 0.75$. They are paired with the conformity score $v(x, y) = My -\hat{q}_\alpha(x)$ with $M=1000$. For the remaining 6 options, we use conformity score of the form $My - \hat\mu(x)/\hat\sigma(x)$, where $\hat\mu$ is taken to be a random forest model, support vector regression model or linear regression model. Each of them is paired with no error estimation, i.e. $\hat\sigma \equiv 1$, or a Lasso error model fitted using 70\% of the total training data.

The second setup for OptCS-Full-MSel experiments is as follows. 
We group the features into four groups: 
\vspace{0.5em}
\begin{enumerate}[label=(\alph*)]
    \item Input uncertainty scores: (\texttt{Lexical\_Sim}, \texttt{Num\_Sets}, \texttt{SE}).
    \item Output confidence scores with Jaccard metric: (\texttt{EigV(J)}, \texttt{Deg(J)}, \texttt{Ecc(J)}).
    \item Output confidence scores with NLI prediction for entailment: (\texttt{EigV(E)}, \texttt{Deg(E)}, \texttt{Ecc(E)}).
    \item Output confidence scores with NLI prediction for contradiction: (\texttt{EigV(C)}, \texttt{Deg(C)}, \texttt{Ecc(C)}).
\end{enumerate}
\vspace{0.5em}

We train three classification model classes (logistic regression, random forest, and XGBoostRF, using the \texttt{scikit-learn} and \texttt{xgboost} Python library) with each individual group of features and all four groups of features, respectively, leading to a total of $15$ classification models $g$. Then, we use the same conformity score $v(x,y) = My-g(x)$ for $M=1000$. 

The baselines considered in the above two setups are very similar to those in Section~\ref{subsec:simu_full_msel}. For clarity, we summarize them again here:

\vspace{0.5em}
\begin{enumerate}[label=(\roman*).]
    \item \texttt{OptCS-Full-MSel}: Our OptCS-Full-MSel procedure where the training process uses over-sampling.
    \item \texttt{Greedy}: Randomly split the labeled data into  $\cD_{\train}$ (62.5\%) and $\cD_\calib$ (37.5\%). After training each candidate model on $\cD_\train$, we use the conformity score function which leads to the largest selection set in SCS, with $\cD_\calib$ as the calibration data.
    This configuration of data splitting is to agree with the ratio suggested in~\cite{gui2024conformal}.
    \item \texttt{Base\_random}: Split the data with the same ratio as in (ii), and then apply \basename with a randomly chosen model.
    \item \texttt{Base\_split\_112}: Randomly split the labeled data into $\cD_\train$, $\cD_\sel$, and $\cD_\calib$ with ratio 1:1:2. We use $\cD_\train$ to train all the models, use $\cD_\sel$  to run \basename (after additional splitting) and select the model with largest selection set, then use $\cD_\calib$ to run SCS with the test data and the selected model. 
    \item \texttt{Base\_split\_121}: Similar to (iv), with  data split ratio 1:2:1 for $\cD_\train$, $\cD_\sel$, and $\cD_\calib$.
    \item \texttt{Base\_split\_211}: Similar to (iv), with  data split ratio 2:1:1 for $\cD_\train$, $\cD_\sel$, and $\cD_\calib$.
    \item \texttt{Base\_split\_111}: Similar to (iv), with  data split ratio 1:1:1 for $\cD_\train$, $\cD_\sel$, and $\cD_\calib$.
\end{enumerate}
\vspace{0.5em}

\section{Additional simulations and ablation studies}

\subsection{Comparison with other full-data training strategies} \label{app:comp_other_full}

In this part, we conduct simulations in Jin's settings to compare our specific full-data training strategy in OptCS-Full and OptCS-Full-MSel, namely, leave-one-out training with over-sampling, with other alternatives. 
First, we use the four Jin's settings as in Section~\ref{subsec:simu_loo} and compare the following procedures:

\begin{enumerate}[label=(\alph*).]
\item OptCS-Full: Our current OptCS-Full with over-sampling (ours).
\item OptCS-Non-LOO: With all test null imputation, train the model without LOO.
\item OptCS-Full-none: With all test null imputation, train the models  with no over-sampling.  
\item OptCS-Full-weight: With all test null imputation, training the model with class weights $w_1 = (n^{(0)}+m)/(n+m)$ and $w_0 = n^{(1)}/(n+m)$,  where $n^{(0)}$ and $n^{(1)}$ are the numbers of labeled samples with $\tilde Y_i=0$ and $\tilde Y_i=1$. 
\item OptCS-Full-removal: With all-test null imputation, randomly remove $m$ negative samples before training. 
\item Split: vanilla SCS as baseline to show the impact of test sample size $m$.
\item Oracle: an unachievable baseline with $n_{\text{labeled}}$ training and $n_{\text{labeled}}$ calibration samples.
\end{enumerate}
Here, (a) is our current procedure, (b-e) represent  different training strategies after all test null imputation, where (b) does not use LOO, and (c) does not address class imbalance, and (d-e) are two alternative and commonly used ways to address class imbalance.  
For consistency with Section~\ref{subsec:simu_loo}, we fix the support vector machine as the model class from \texttt{scikit-learn}. 
The results are reported in Figure~\ref{fig:simu_full_alt_n200}. 

OptCS-Full is consistently more powerful than the Split baseline, and often close to the \emph{unachievable} Oracle baseline (the power of both tends to decline as $m$ grows, as the multiple testing problem becomes more challenging). Regarding OptCS-based algorithms with alternative training strategies,  the comparison with OptCS-Non-LOO shows the power of LOO training, and the comparison with others shows the efficacy of the over-sampling strategy. The sample weighting approach (which is conceptually similar to over-sampling) sometimes yields better performance, but overall, the current OptCS-Full presents robust power gain. 

\begin{figure}
    \centering
    \includegraphics[width=\linewidth]{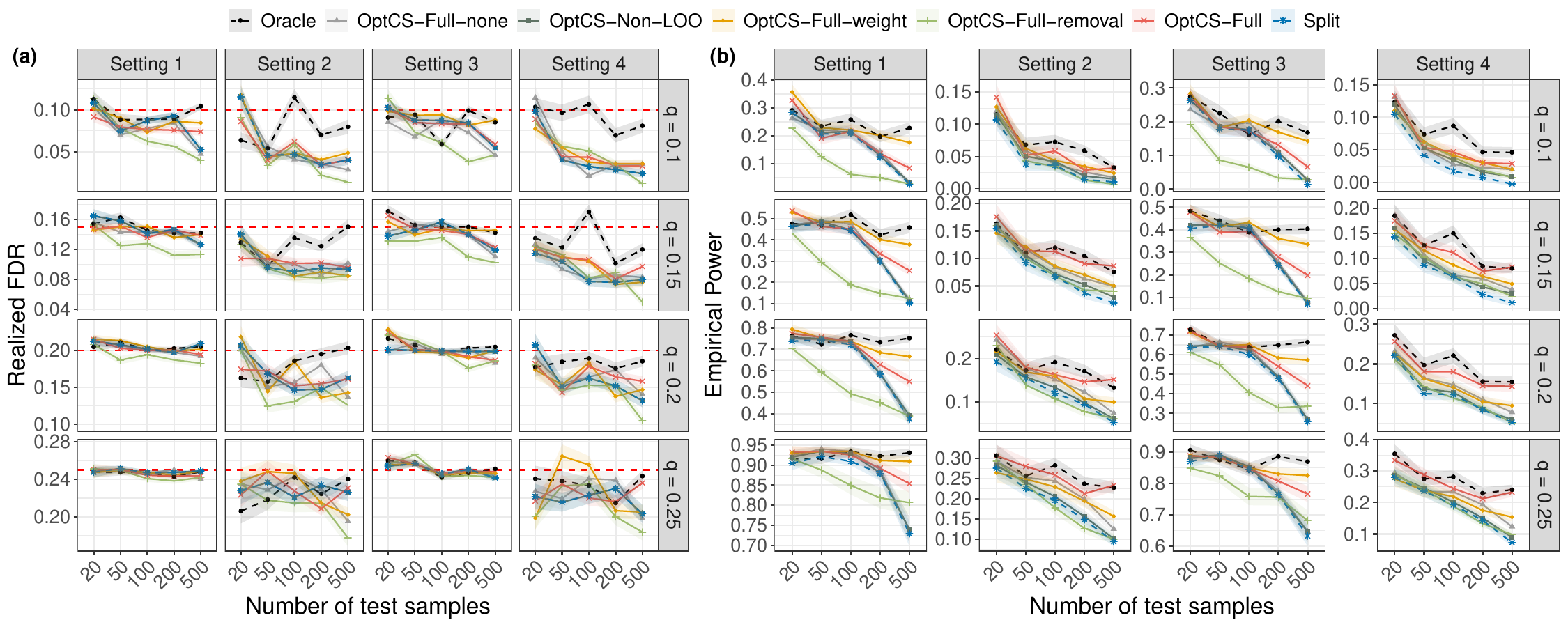}
    \caption{(a) Empirical FDR and (b) Power for OptCS-Full (over-sampling, red) and baselines in Jin's settings across  test sample sizes $m\in \{20,50,100,200,500\}$ and target FDR levels $q\in \{0.1,0.15,0.2,0.25\}$, fixing $n_{\text{labeled}}=200$ and model class as support vector machine.}
    \label{fig:simu_full_alt_n200}
\end{figure}

Similarly, we evaluate OptCS-Full-MSel together with other training strategies than over-sampling and leave-one-out training. The methods include:
\begin{enumerate}[label=(\alph*)]
\item OptCS-Full-MSel with over-sampling (ours).
\item OptCS-FM-Non-LOO: imputing all test samples and train a classifier without LOO; other steps (model selection and FDR calibration) are the same. 
\item OptCS-FM-none: OptCS-Full-MSel with no over-sampling.  
\item OptCS-FM-weight: OptCS-Full-MSel with class weights $w_1 = (n_0+m)/(n+m)$ and $w_0 = n_1/(n+m)$. 
\end{enumerate}

We use the same Liang's settings  and model classes as in Section~\ref{subsec:simu_full_msel}.
The results are in Figure~\ref{fig:msel_compare_liang} with (a) empirical FDR and (b) empirical power. We observe that OptCS-Full-MSel with over-sampling (red) consistently achieves superior power across all settings and target FDR levels, which further demonstrate the power of the over-sampling strategy.

\begin{figure}
    \centering
    \includegraphics[width=\linewidth]{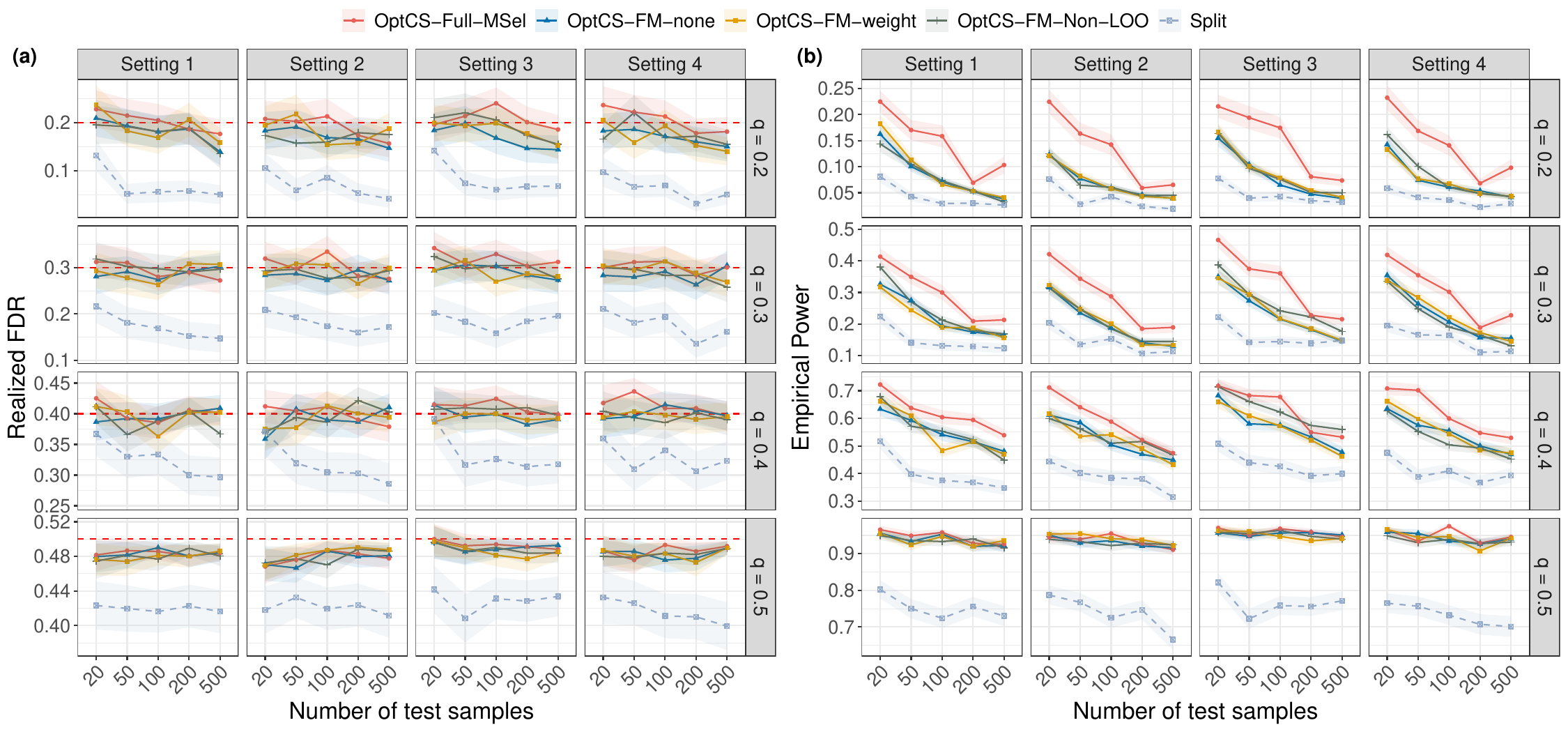}
    \caption{(a) Empirical FDR and (b) Power for OptCS-Full-MSel (over-sampling, red) and baselines in Liang's settings and model classes identical to Section~\ref{subsec:simu_full_msel} across  test sample sizes $m\in \{20,50,100,200,500\}$ and target FDR levels $q\in \{0.2,0.3,0.4,0.5\}$, fixing $n_{\text{labeled}}=200$.}
    \label{fig:msel_compare_liang}
\end{figure}

\subsection{Evaluating the impact of null imputation}
\label{app:subsec_null_impute_experiment}

As discussed at the end of Section~\ref{subsec:loo}, with the computationally efficient OptCS-Full with all-test null imputation, a practical concern is the dilution of signal due to such imputation. 
In complement to the theoretical discussion in Appendix~\ref{app:subsec_null_impute_discuss}, here we empirically investigate the performance of OptCS-Full under various strength of signal dilution. Our results demonstrate the robust performance of OptCS-Full, as well as the particular power of OptCS-Full-sep, which can be useful when $m$ is large and signal dilution is a concern. 

Throughout, we fix $n_{\text{labeled}}=200$ and vary (i) the test sample size $m\in \{20,50,100,200,500\}$, and (ii) the null proportion by first generating a larger sample size from the correct distribution and randomly select positive and negative samples, so the positive samples takes a fraction $\gamma\in \{0.4, 0.5, 0.6, 0.7, 0.8\}$ (outside of this region, the power is either zero or one).  
For consistency, for OptCS-Full, we use Jin's settings  in Section~\ref{subsec:simu_loo} and fix the model class as random forests from \texttt{scikit-learn} whose performance changes visibly with training sample size in SCS. We compare the following procedures: 
\begin{itemize}
\item OptCS-Full: our main paper variant with oversampling. 
\item OptCS-Full-sep: the variant in Appendix~\ref{app:subsec_discuss_loo}.
\item Split: vanilla SCS with $50\%/50\%$ training-calibration split.
\item Oracle: vanilla SCS with $n_{\train}=n_{\calib}=200$, serving as an unachievable benchmark where the same number of labeled data is used for training and calibration, separately. 
\end{itemize}

The results across the various settings are summarized in Figure~\ref{fig:scale_full}, which shows the realized FDR (a) and power (b). As expected, OptCS-Full and OptCS-Full-sep achieve much higher power than the Split baseline, and often comparable to the unachievable Oracle baseline, showing the benefit of full-data training. The performance of OptCS-Full with all-test null imputation moderately declines with the unlabeled sample size $m$; this is expected since the true signal is diluted when $m$ becomes substantially large. 
However, the power decline of OptCS-Full seems at least comparable to that of Split, showing that the power decline may be partially due to the fact that the multiple testing problem itself becomes more difficult. 
In contrast, OptCS-Full-sep, which avoids all-test null imputation at a higher computational cost, maintains high power with little power decline as $m$ increases. Therefore, we recommend using OptCS-Full when $m$ is moderately large, and using OptCS-Full-sep for relatively large values of $m$ when the practitioner is willing to contribute more intensive computation to fully exploit the labeled data. 

\begin{figure}
    \centering
    \includegraphics[width=\linewidth]{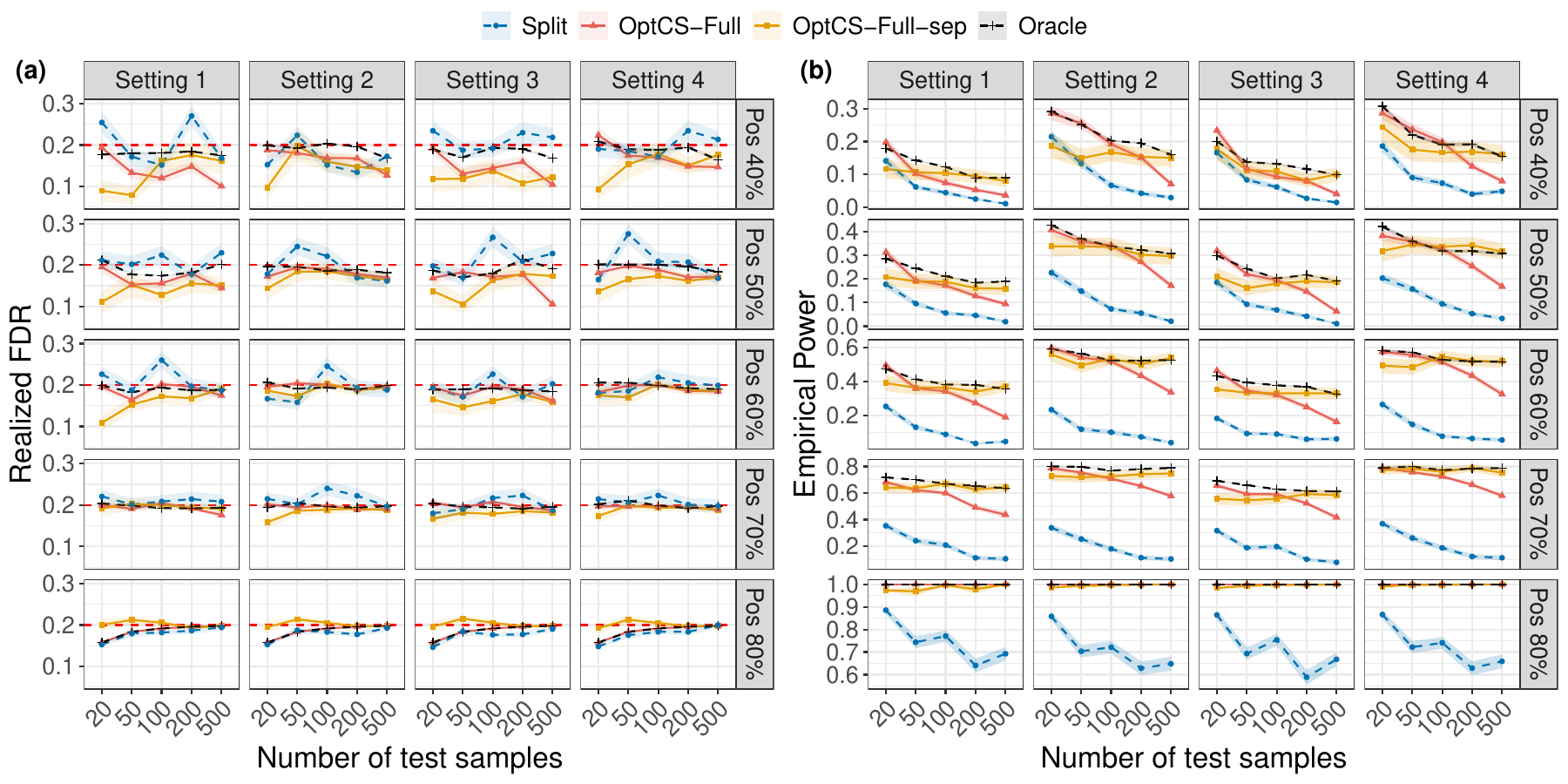}
    \caption{Comparison of OptCS-Full and OptCS-Full-sep with split and oracle baselines across Jin's settings while varying the fraction of positive samples $\gamma\in \{0.4,0.5,0.6,0.7,0.8\}$ and the number of test samples $m\in \{20,50,100,200,500\}$, fixing the number of labeled data at $n_{\text{labeled}}=200$, FDR target level at $q=0.2$, and model class as random forests from \texttt{scikit-learn}. Each subplot corresponds to one setting and one value of $\gamma$, and within each subplot, the $x$-axis is $m$.}
    \label{fig:scale_full}
\end{figure}

\begin{figure}
    \centering
    \includegraphics[width=\linewidth]{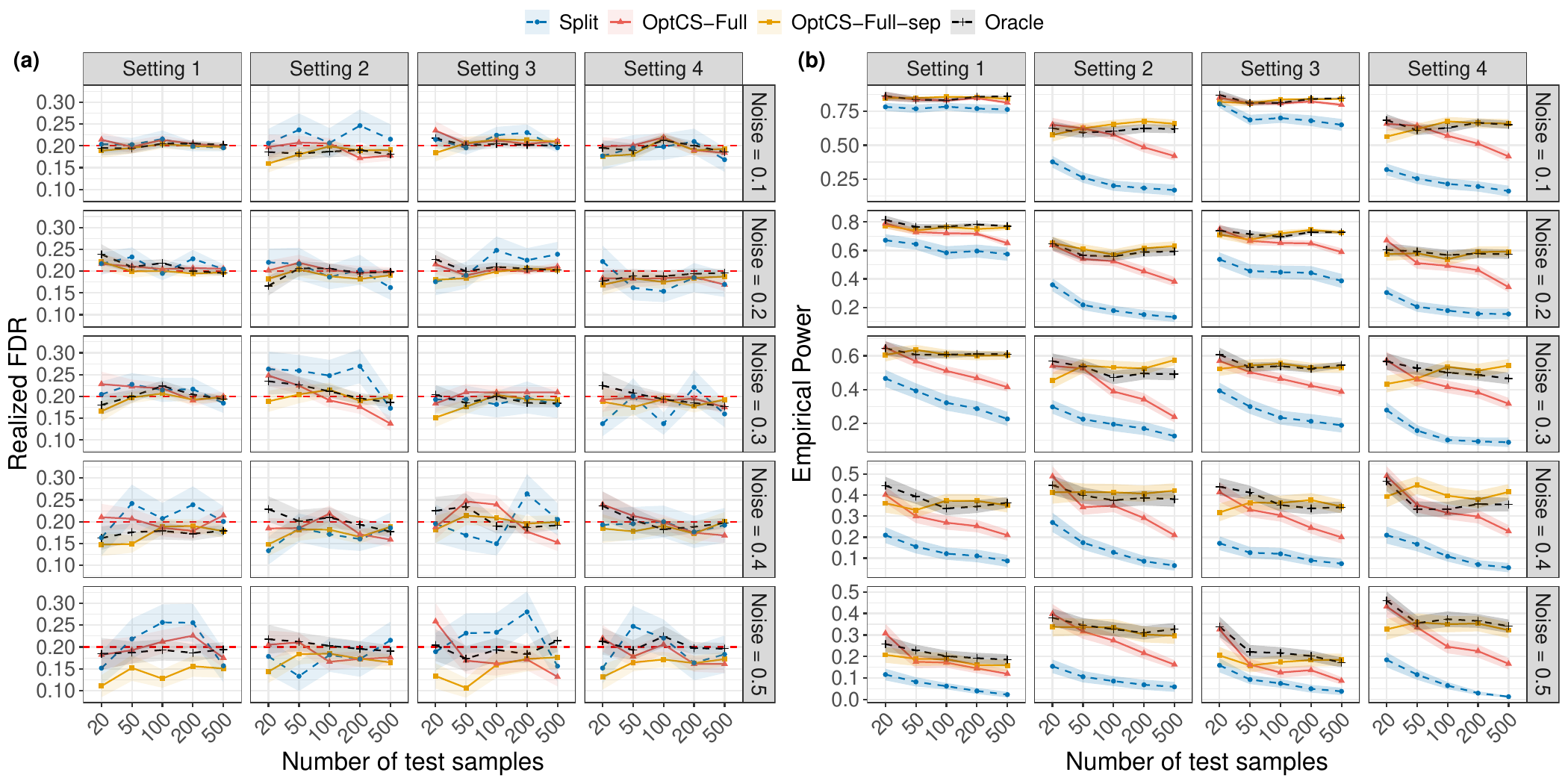}
    \caption{(a) Empirical FDR and (b) Power of OptCS-Full and OptCS-Full-sep with split and oracle baselines across Jin's settings while varying the noise standard deviation $\sigma\in \{0.1,0.2,0.3,0.4,0.5\}$ and the number of test samples $m\in \{20,50,100,200,500\}$, fixing the number of labeled data at $n_{\text{labeled}}=200$, FDR target level at $q=0.2$, and model class as random forests from \texttt{scikit-learn}. Each subplot corresponds to one setting and one value of $\sigma$, and within each subplot, the $x$-axis is the value of $m$.}
    \label{fig:scale_full_signal}
\end{figure}

We further examine the effect of signal strength.  Here, we fix the positive fraction at
\(\gamma=0.5\) and vary the noise level
\(\sigma\in\{0.1,0.2,0.3,0.4,0.5\}\) in Jin's settings, where smaller \(\sigma\) corresponds to
stronger signal (so the performance gap of models trained over half/full labeled sets is smaller). We keep \(n_{\mathrm{labeled}}=200\), \(q=0.2\), and the random forest model
class as above, and vary \(m\in\{20,50,100,200,500\}\).

The results are summarized in Figure~\ref{fig:scale_full_signal}. Both OptCS-Full and
OptCS-Full-sep maintain empirical FDR control across signal strengths. In terms of power,
both methods improve over the Split baseline and are often close to the
unachievable Oracle benchmark, demonstrating the benefit of full-data training. As expected,
OptCS-Full can lose power when \(m\) is large, due to stronger signal dilution from 
null imputation. OptCS-Full-sep is more stable in these regimes, at the cost of additional
model fitting. Overall, OptCS-Full and OptCS-Full-sep remain robust interpolations between Split and Oracle under various performance gaps. 

For OptCS-Full-MSel, we use Liang's settings with the same 7 model classes following Section~\ref{subsec:simu_full_msel}.  We compare the following procedures: 
\begin{itemize}
\item OptCS-Full-MSel: our main paper variant with oversampling. 
\item OptCS-Full-MSel-sep: the variant in Appendix~\ref{app:subsec_discuss_loo} with one null imputation.
\item Split: vanilla SCS with $1/2/1$ training-selection-calibration split (the most powerful three-split variant based on  experiments in Section~\ref{subsec:simu_full_msel}). 
\end{itemize}

Here no ``oracle'' baseline is included since it's less straightforward to mimic an oracle procedure in the context of model selection. 
The results across various settings are in Figure~\ref{fig:full_msel_scale}. Again,  OptCS-Full-MSel remains more powerful than Split even with $m=500$ and $n_{\text{labeled}}=200$. Although its power decreases with $m$, the decline seems to  match the decline inherent in SCS. On the other hand, OptCS-FM-sep, which avoids all-test null imputation, is robust to the increase in $m$, providing a powerful variant when signal dilution is a concern.

\begin{figure}
    \centering
    \includegraphics[width=\linewidth]{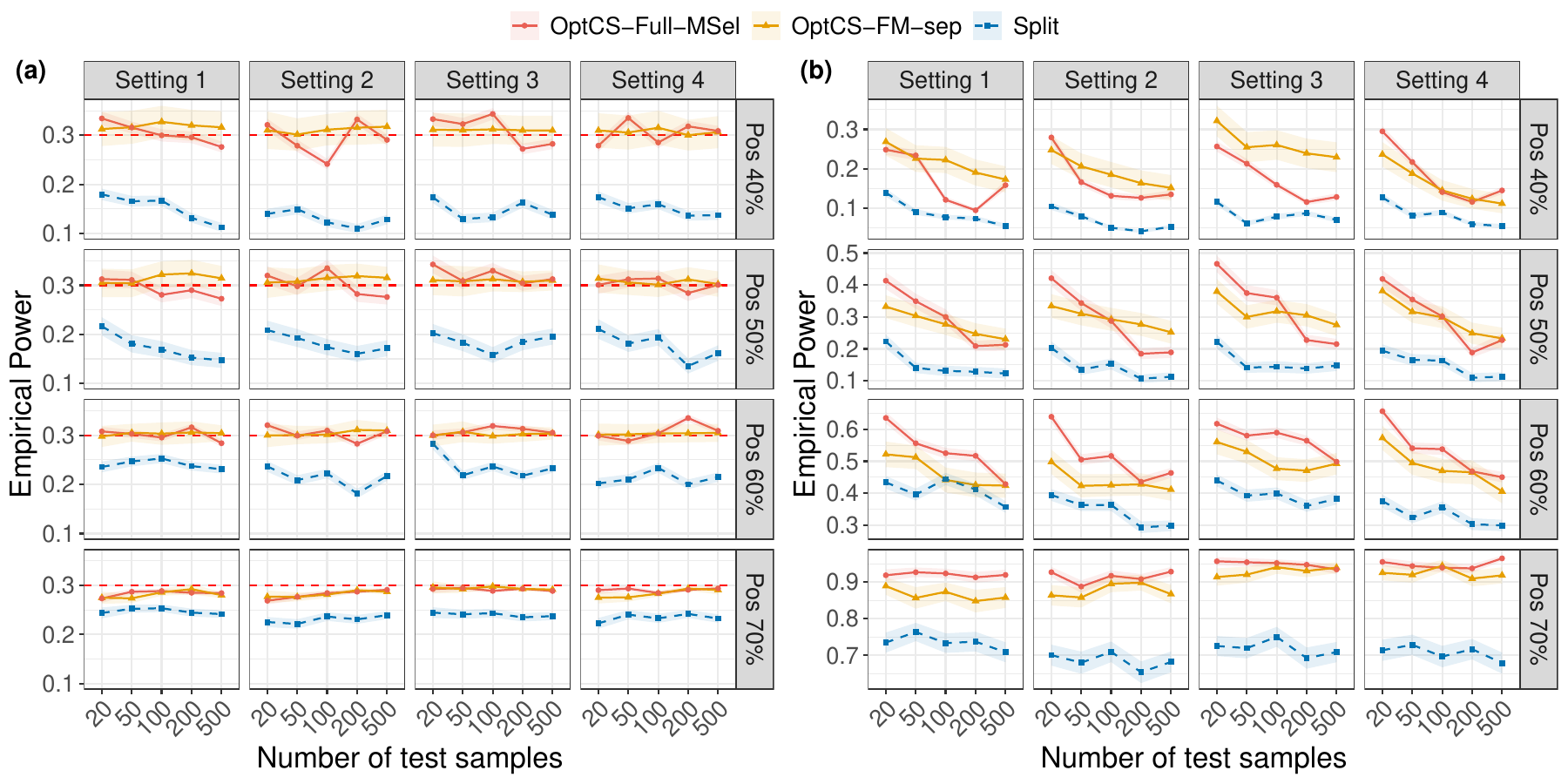}
    \caption{Comparison of OptCS-Full-MSel and OptCS-Full-MSel-sep with split baseline across Liang's settings while varying the fraction of positive samples $\gamma\in \{0.4,0.5,0.6,0.7\}$ and the number of test samples $m\in \{20,50,100,200,500\}$, fixing the number of labeled data at $n_{\text{labeled}}=200$, FDR target level at $q=0.3$, with the same model classes as Section~\ref{subsec:simu_full_msel}. Each subplot corresponds to one setting and one value of $\gamma$, and within each subplot, the $x$-axis is  $m$.}
    \label{fig:full_msel_scale}
\end{figure}

\subsection{Comparison with outlier detection and adaptive conformal selection methods}
\label{app:subsec_exp_od}

In this section, we compare OptCS-Full(-MSel) with other baselines that can perform full-data training or model selection, namely, the recent adaptive conformal selection method of~\cite{gui2025acs} and outlier detection methods of~\cite{marandon2024adaptive},~\cite{liang2024integrative}, and~\cite{lee2025full}. 
 
\subsubsection{Comparing OptCS-Full with novelty detection baselines}

First, we compare OptCS-Full against the following: 
\begin{itemize} 
\item AdaDetect, the method of~\cite{marandon2024adaptive} using the official implementation in the paper github repository with the same scikit-learn model as OptCS.
\item AdaDetect-Storey, the AdaDetect method together with Storey's correction with default parameter $\lambda=0.5$.
\item IC-Split, the method of \cite{liang2024integrative} with sample splitting, using the paper github repository with the same scikit-learn model with the same split ratio as in their implementation. The method relies on both one-class classifiers and binary classifiers; the one-class classifier is fixed as the isolation forest, and the binary classifier is the same as OptCS.
\item IC-TCV, the method of \cite{liang2024integrative} with $K$-fold cross-validation, using the paper github repository with the same $K=5$ splits as in their implementation. Given the p-values, we apply the BH to obtain the selection set (whose FDR control is unclear).
\item IC-TCV-CC, applying the conditional calibration method for FDR control using the paper github repository, and otherwise the same as IC-TCV.
\end{itemize}

We use the four settings of~\cite{jin2023selection} in Section~\ref{subsec:simu_loo}, fixing the FDR target level at $q=0.2$. To ensure fair comparison, all methods are given the same model class (classifiers from \texttt{scikit-learn} following the official implementations of AdaDetect and IC).

\begin{figure} 
\centering 
\includegraphics[width=0.8\textwidth]{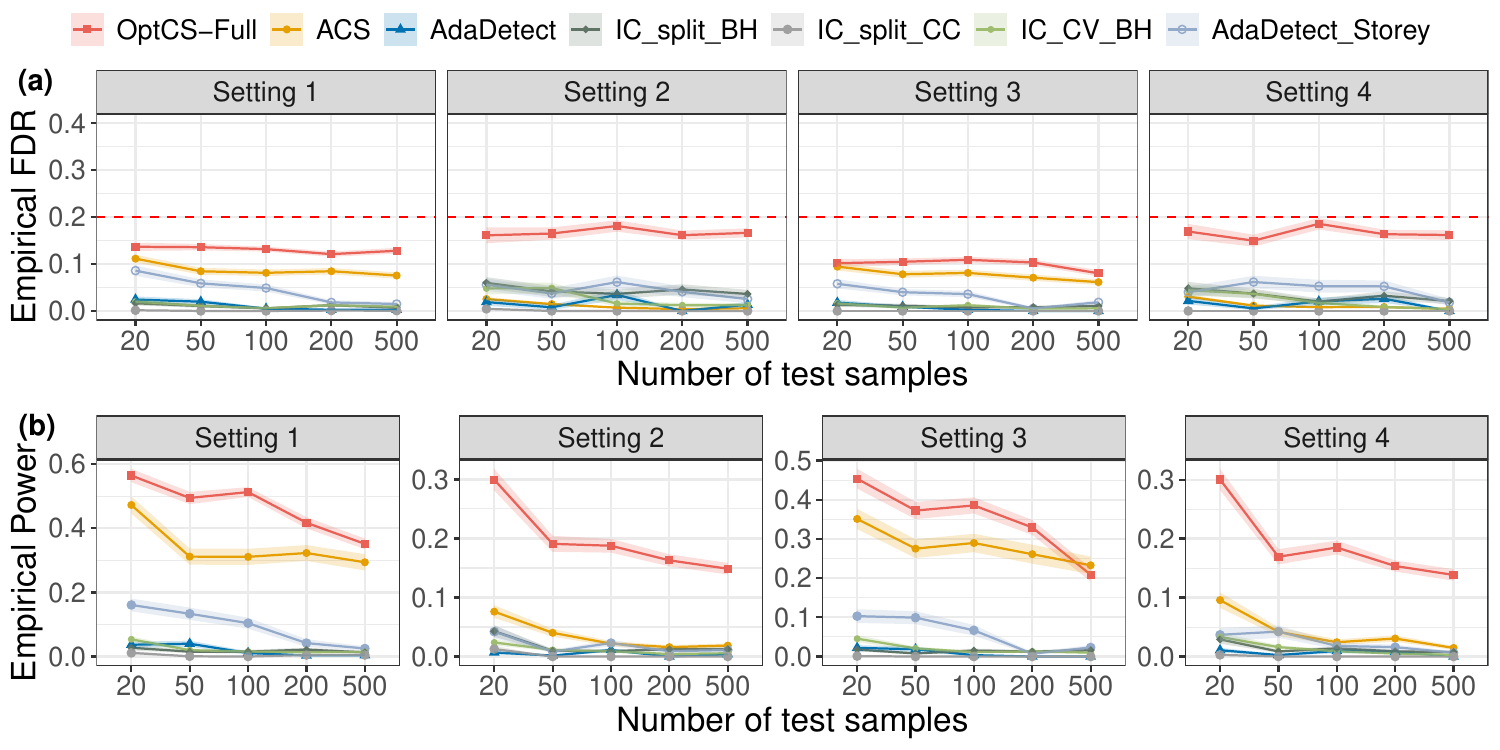}
\caption{(a) Empirical FDR and (b) Power of OptCS-Full vs.~conformal selection methods (ACS) and outlier detection methods (AdaDetect and IC-based methods) in Jin's settings across various test sample sizes $m\in \{20,50,100,200,500\}$, fixing $n_{\text{labeled}}=200$ and model class as SVM in \texttt{scikit-learn}.}
\label{fig:od_fix_svr}
\end{figure}
The results are summarized in Figure~\ref{fig:od_fix_svr}. OptCS-Full demonstrates superior power across various sample sizes and settings. In particular, the AdaDetect method and IC variants, whose p-values and model training strategies are designed for outlier detection, typically do not exploit the FDR budget; the Storey's correction improves upon AdaDetect's FDR budget usage, but the improvement in power is less significant. ACS with adaptive model update achieves higher power than the outlier detection baselines, but is still less powerful than OptCS-Full in most cases. We conjecture that when the selection power is moderately high, ACS stops before utilizing a substantial fraction of labeled data (i.e., when ACS stops, there are still a fraction of labeled data that is masked and did not enter training), which might contribute to its slightly inferior power. 
In summary,  OptCS is competitive against these baselines, likely because (i) it leverages ``clipped'' scores to exploit FDR budget, (ii) its training is designed to utilize all the labeled data and is  powerful. 

We then apply the same methods to the large-language model application, using the same processing steps as in Section~\ref{subsec:app_LLM}. We fix the model class among gradient boosting, logistic regression, random forests, and support vector machine, and the FDR target level is $q=0.2$. The results are summarized in Figure~\ref{fig:od_llm_full} (next page). Again, OptCS-Full is robustly more powerful in most cases, showing the benefit of full-data training. We note that the advantage  of OptCS-Full is particularly significant for powerful model classes like logistic regression and random forest. For less powerful models like gradient boosting and SVM, OptCS-Full can degrade a bit when $m$ is very large, perhaps because these models are more sensitive to signal dilution.

\begin{figure} 
\centering 
\includegraphics[width=0.8\textwidth]{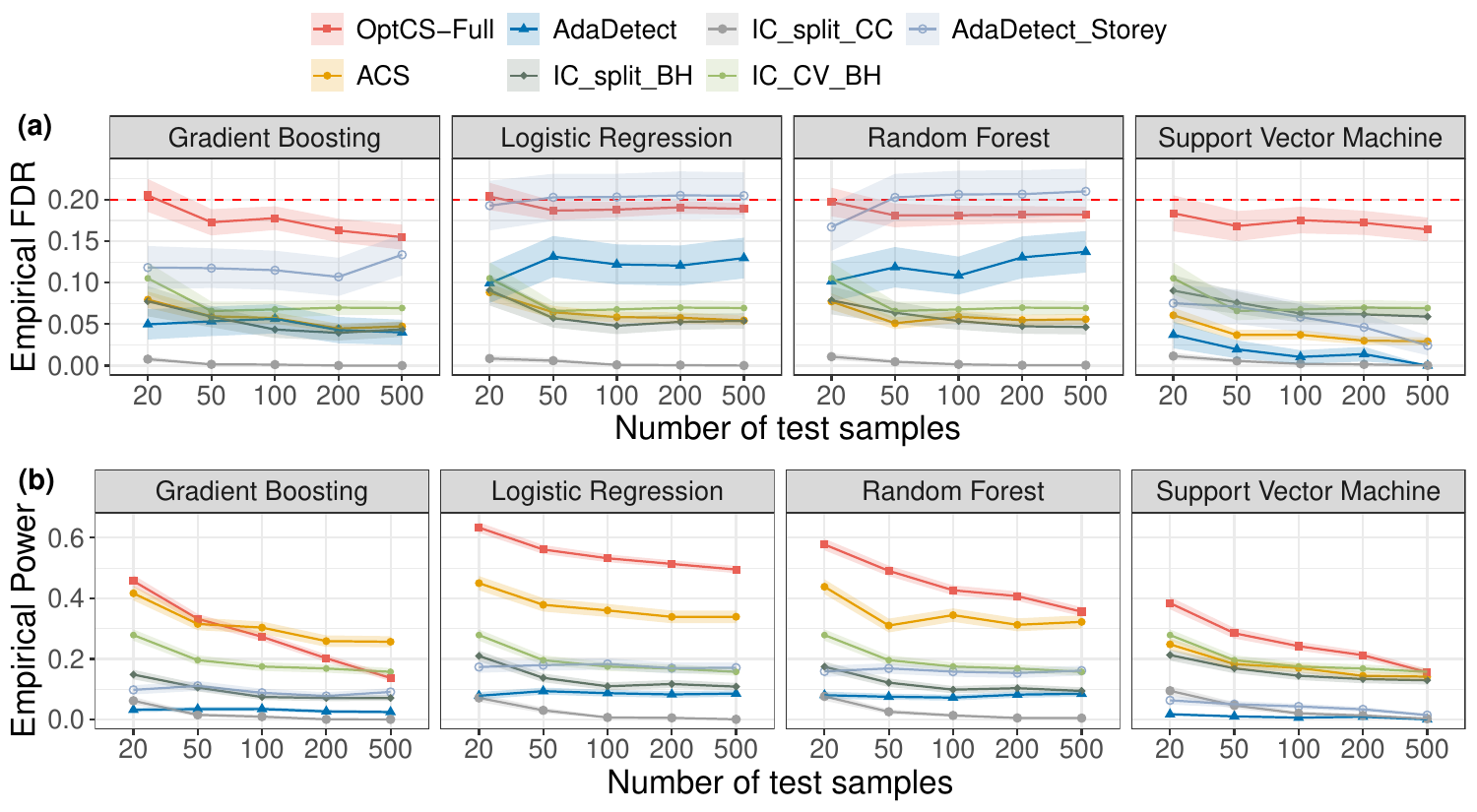}
\caption{(a) Empirical FDR and (b) Power of OptCS-Full vs.~conformal selection methods (ACS) and outlier detection methods (AdaDetect and IC-based methods) in LLM application (with different model classes) across various test sample sizes $m\in \{20,50,100,200,500\}$, fixing $n_{\text{labeled}}=200$, $q=0.2$.}
\label{fig:od_llm_full}
\end{figure}

\subsubsection{Comparing OptCS-Full-MSel against baselines}

Second, we compare OptCS-Full-MSel with the method of \cite{liang2024integrative} which is able to perform model selection (based on median difference), and ACS which uses prediction accuracy to perform model selection while updating the models: %
\begin{itemize}  
\item OptCS-Full-MSel, our method with several model classes (to be specified later).
\item IC-Split, the method of \cite{liang2024integrative} with sample splitting and BH selection (with FDR control), using the paper github repository with the same models as ours. Apart from the three binary classifiers, it also chooses from one-class classifiers identical to the paper repository. 
\item IC-TCV, the method of \cite{liang2024integrative} with $K$-fold cross-validation and BH selection (no formal FDR control), using the implementation in the paper github repository and same model classes as IC-Split.
\item IC-TCV-CC, applying the conditional calibration method for FDR control using the paper github repository, and otherwise the same as IC-TCV above.
\item ACS, applying the variant with model updates, with the method in the paper that selects model based on certain prediction accuracy measure. 
\end{itemize}

We first conduct simulation studies. While OptCS-Full-MSel supports general conformity scores (its model selection relies on mimicking the selection performance, which applies to scores of any form), the model selection strategy in other methods only applies to standard prediction models (classifiers). Thus, the 9 model classes in Section~\ref{subsec:simu_full_msel} of the main text, which wraps around point predictors in various ways, are not natively applicable to the competing methods. For fair comparison  without modifying the implementation details, we use the four Jin's settings with three model choices: gradient boosting, random forests, and support vector machine, all from the \texttt{scikit-learn} library. 

The results are summarized in Figure~\ref{fig:od_msel_q02}, where OptCS-Full-MSel demonstrates superior power across various sample sizes and settings. The IC-variants, whose model selection strategy (based on discrepancy of predictions across groups) and multiple testing procedure are tailored for outlier detection settings, yield lower power. ACS, which is also based on the clipped score, is able to achieve higher power than other baselines but is less powerful than OptCS-Full-MSel. We conjecture that this is because (i) ACS does not utilize all the labeled data for training due to stopping, and (ii) the model selection strategy in OptCS-Full-MSel is more directly tailored for selection performance.

\begin{figure}
\centering 
\includegraphics[width=0.8\textwidth]{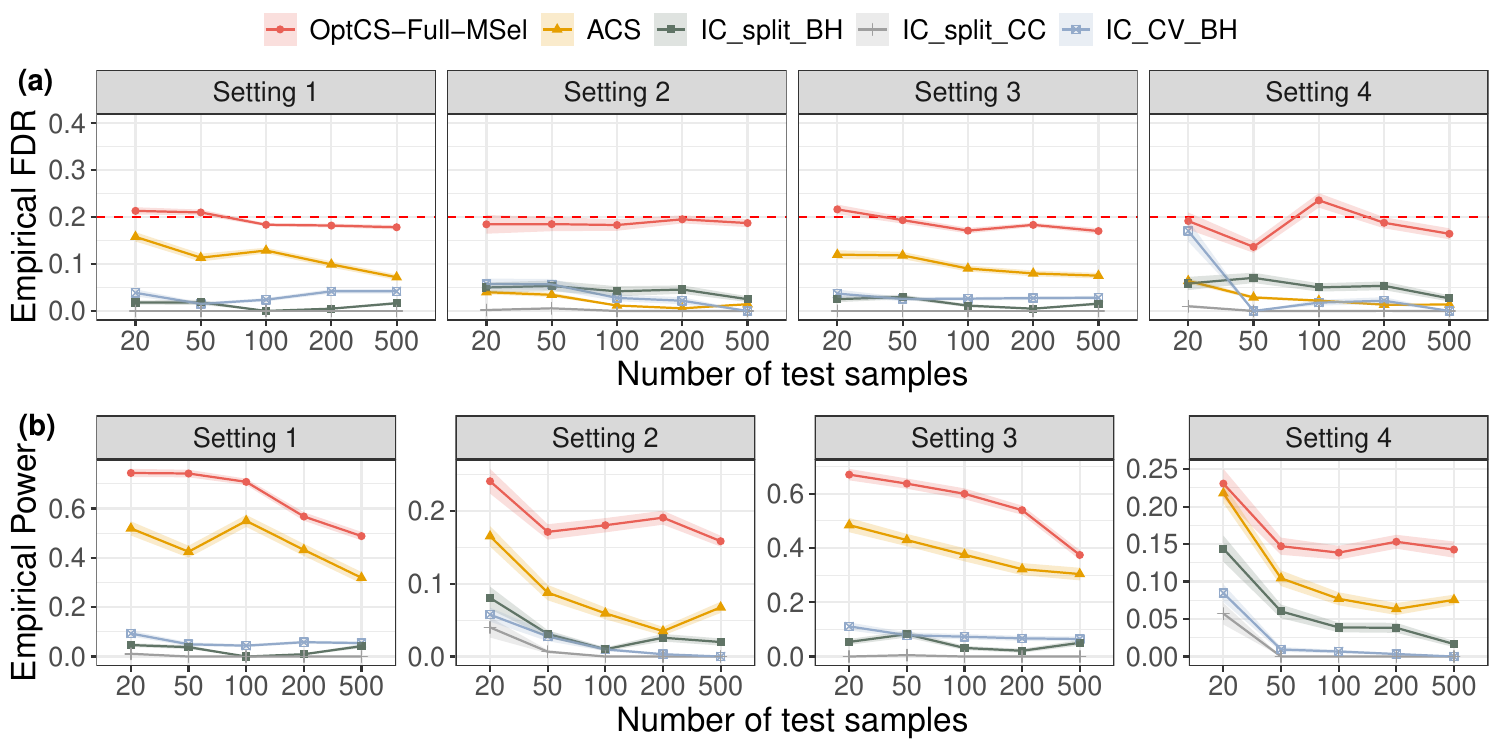}
\caption{Comparing OptCS-Full-MSel against conformal selection methods (ACS) and outlier detection methods (IC variants) in Jin's settings across various test sample sizes $m\in \{20,50,100,200,500\}$, fixing $n_{\text{labeled}}=200$ and model classes as $\{$GB, RF, SVR$\}$ in \texttt{scikit-learn}.}
\label{fig:od_msel_q02}
\end{figure}

\begin{figure} 
\centering 
\includegraphics[width=0.8\textwidth]{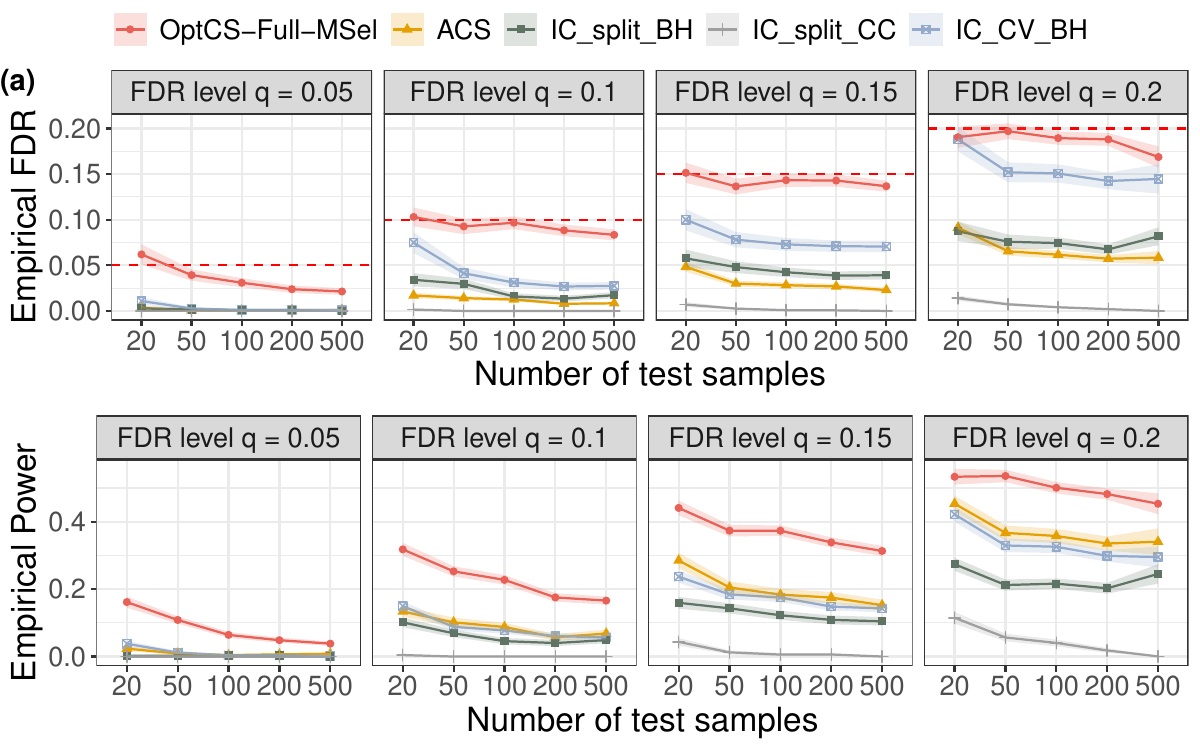}
\caption{Comparing OptCS-Full-MSel against ACS and IC variants  in LLM application across various test sample sizes $m\in \{20,50,100,200,500\}$ and FDR target level $q\in\{0.05,0.1,0.15,0.2\}$, fixing $n_{\text{labeled}}=200$ and model classes as $\{$GB, RF, SVR, LR$\}$ in \texttt{scikit-learn}.}
\label{fig:od_msel_llm}
\end{figure}

Finally, we apply these procedures to the LLM data as before. The model selection is performed over four models: GB, RF, SVM, and logistic regression (with IC methods using additional one-class classification methods following their official implementation). The results are summarized in Figure~\ref{fig:od_msel_llm}, where OptCS-Full-MSel robustly demonstrates superior power, likely because (i) it uses all the labeled data for training, and (ii) the model selection criterion directly targets selection performance instead of prediction accuracy.

\subsubsection{Comparing OptCS-Full and full-conformal novelty detection}
\label{app:subsec_compare_fcnd_empirical}

In this section, we provide a targeted comparison with the full-conformal novelty detection
(FCND). As discussed in Appendix~\ref{app:fcnd_connection}, after the binary reduction
used by OptCS-Full, FCND becomes a close e-value-based comparator. The goal of this
section is to evaluate the direct FCND adaptation in our supervised conformal-selection
setting, using the official \texttt{fcnd} implementation.

\paragraph{Adapting FCND to our setting.}
The adaptation of FCND to our setting is analogous to AdaDetect and IC. We take the labeled null units
\( 
\{i:Y_i\le c_i\}
\)
as the reference/inlier sample, and treating the unlabeled test units as the test sample. The
true test labels are used only for evaluating empirical FDR and power. This is the natural
mapping from the thresholded conformal-selection problem to FCND's novelty-detection
formulation. However, this adaptation uses the data differently
from OptCS-Full: the labeled positives $\{i\colon Y_i>c_i\}$ are not part of the null reference sample and are not natively used in
the same supervised way. In contrast, OptCS-Full uses both labeled positives and labeled
negatives to train the supervised score, and imputes test labels as null to preserve validity.
We compare the following methods:
\begin{itemize}
    \item {OptCS-Full}: our default full-data method with leave-one-out (LOO) training.
    \item {OptCS-Full-Non-LOO}: the non-LOO variant of OptCS-Full, which requires only
    one model fit.
    \item {FCND}: the default 1-FC e-BH procedure implemented in the official
    \texttt{fcnd} package.
    \item {FCND-LOO}: the LOO-1-FC e-BH procedure implemented in the official
    \texttt{fcnd} package.
\end{itemize}
The one-fit procedures, FCND and OptCS-Full-Non-LOO, represent the computationally light
full-data regime, whereas FCND-LOO and OptCS-Full represent the corresponding LOO
regime. For FCND, we report the one-class novelty-detection scores implemented in the
official package, including Isolation Forest, Local Outlier Factor, and one-class SVM. Since
the raw e-BH procedure can underuse the FDR budget, we also implement the null-fraction
correction described by~\cite{lee2025full}. Unless otherwise stated, the FCND variants shown
below use this null-fraction correction. This correction improves FDR-budget utilization,
but it does not change the underlying data-use pattern: the score is still constructed from a
null-reference sample rather than from both labeled positives and labeled negatives.

\paragraph{Simulation studies.}
We first evaluate the methods in Jin's settings used in Section~\ref{subsec:simu_loo}. We fix the FDR target
level at \(q=0.2\), vary the number of labeled samples
\(n\in\{100,200,500\}\), and vary the test sample size
\(m\in\{20,50,100,200,500\}\). For OptCS-Full and OptCS-Full-Non-LOO, we use a random
forest classifier from \texttt{scikit-learn} as a representative supervised learner. For FCND,
we use the official one-class score implementations described above. Each configuration is
averaged over \(N=200\) runs.

The results are shown in Figure~\ref{fig:fcnd_sim}. The FCND variants control FDR, but
remain substantially less powerful than OptCS-Full across the displayed settings. The
null-fraction correction improves the FDR-budget side of FCND, but does not close the
power gap. This supports the data-use interpretation above: in these supervised conformal
selection problems, labeled positives carry useful information for score construction, and
OptCS-Full directly exploits this information through supervised score training. This shows the distinction of data usage and validity mechanism. The
non-LOO variant of OptCS-Full can be competitive in some low-label regimes, but it may
lose power when the fitted supervised model is strong, illustrating the computation--power
tradeoff behind the LOO construction.

\begin{figure}[t]
    \centering
    \includegraphics[width=\textwidth]{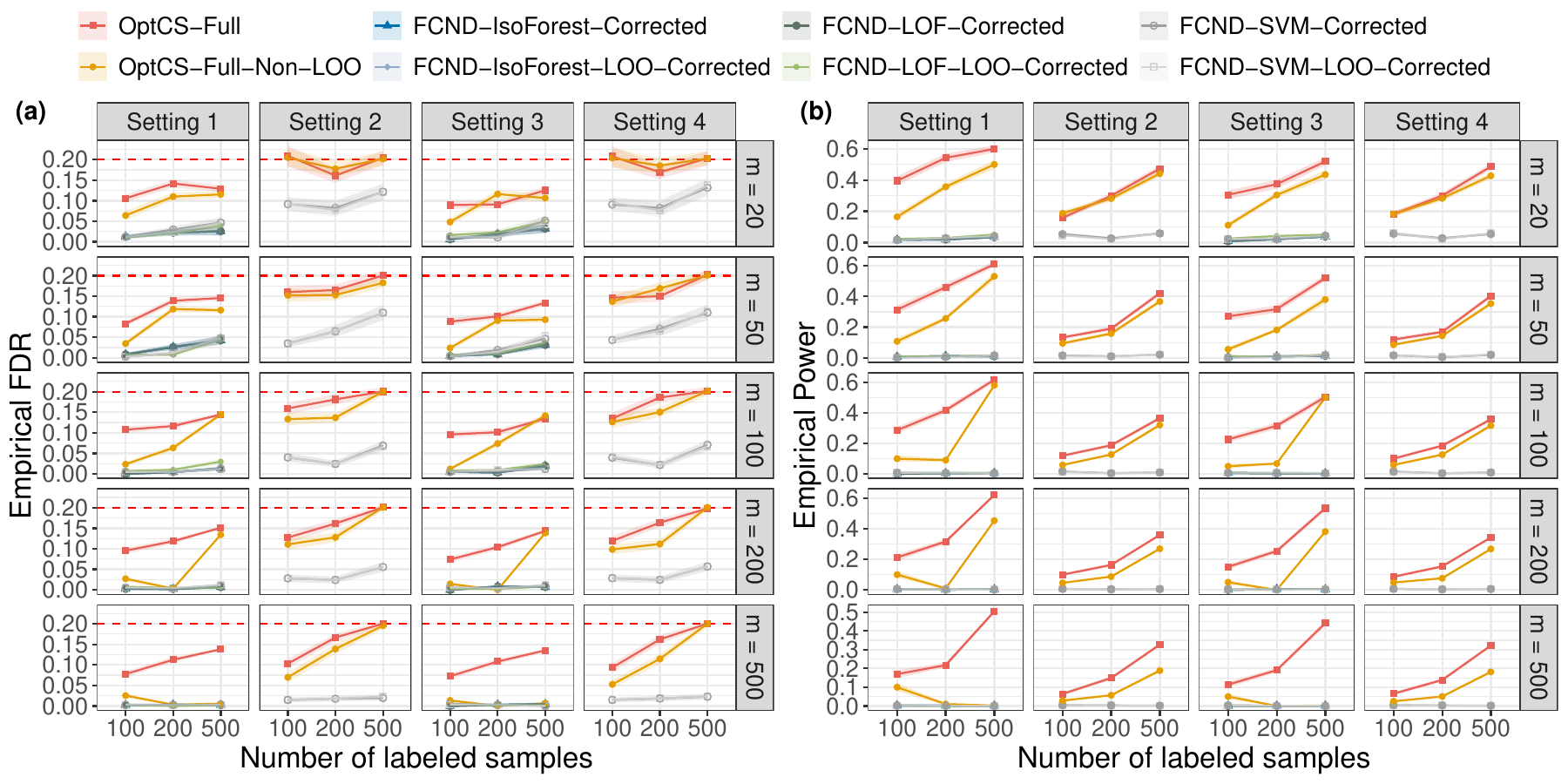}
    \caption{Empirical FDR and power of OptCS-Full, OptCS-Full-Non-LOO, and FCND
    novelty-detection baselines in Jin's settings at target FDR level \(q=0.2\), varying the
    test sample size \(m\) (rows) and labeled sample size \(n\) (x-axis), averaged over
    \(N=200\) runs. OptCS uses a random forest classifier from \texttt{scikit-learn}. FCND
    variants use the official \texttt{fcnd} implementation with the null-fraction correction.}
    \label{fig:fcnd_sim}
\end{figure}

\paragraph{Real-data comparison in the LLM application.}
We also compare the methods in the LLM reliability application of Section~\ref{subsec:app_LLM}. The dataset
is a random subset of MIMIC-CXR, and the features are the uncertainty measures described
in Appendix~\ref{app:subsec_detail_LLM}, following~\cite{gui2024conformal}. We fix
\(n_{\mathrm{labeled}}=200\), \(q=0.2\), and vary
\(m\in\{20,50,100,200,500\}\). For OptCS-Full and OptCS-Full-Non-LOO, we use four
lightweight supervised model classes: gradient boosting, logistic regression, random forests,
and support vector machines, all from \texttt{scikit-learn}. For FCND, we use the same
one-class score families as above with the null-fraction correction.  

The results are summarized in Figure~\ref{fig:fcnd_llm}. The real-data comparison shows
the same qualitative pattern as the simulations. FCND controls FDR but has substantially
lower power, while OptCS-Full and OptCS-Full-Non-LOO benefit from using labeled positives
and negatives in the supervised score. 
These results again reflect the
structural difference between null-reference novelty detection and supervised conformal
selection.

\begin{figure}[t]
    \centering
    \includegraphics[width=\textwidth]{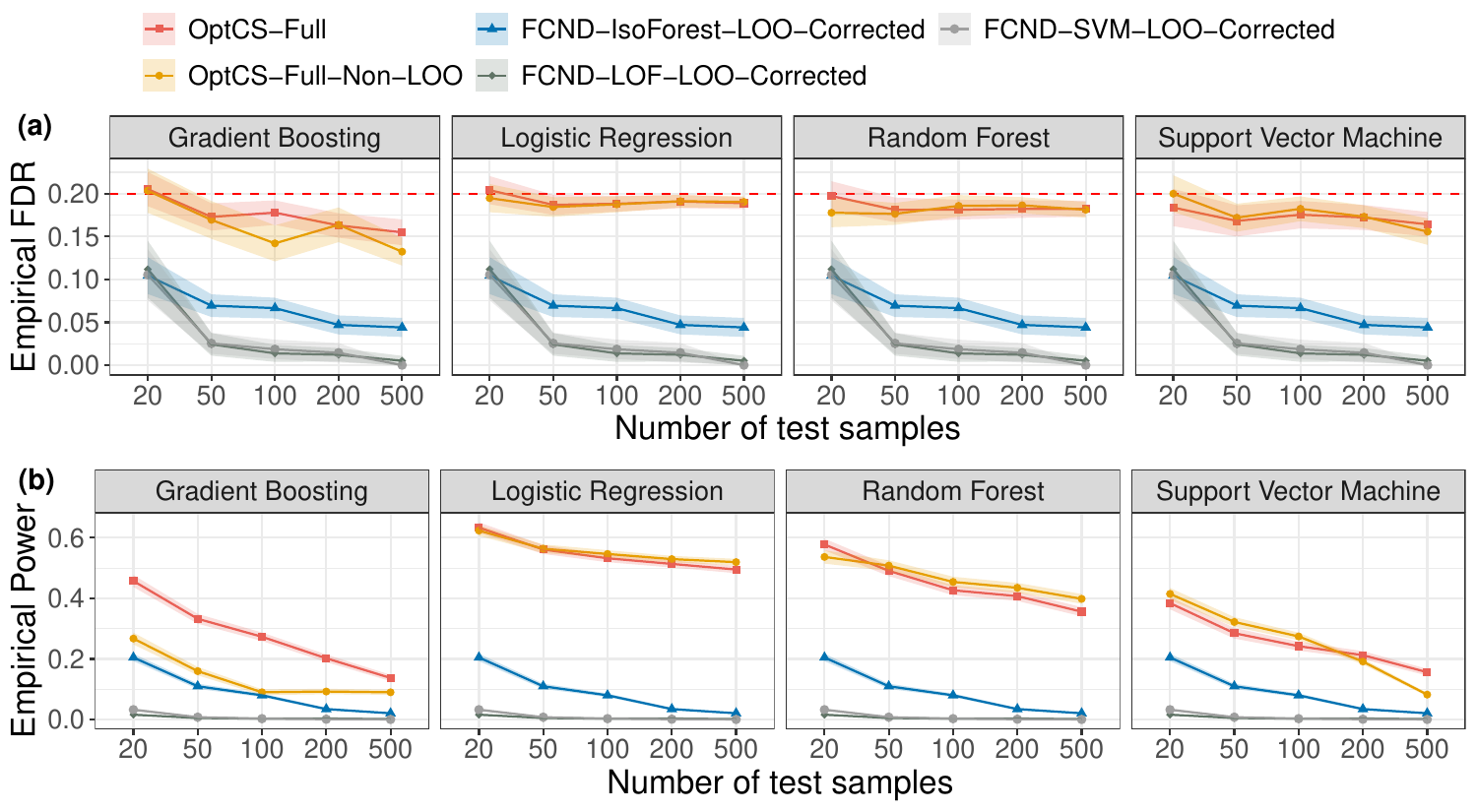}
    \caption{Empirical FDR and power of OptCS-Full and OptCS-Full-Non-LOO versus FCND variants, in the LLM
    application across test sample sizes \(m\in\{20,50,100,200,500\}\) and model classes,
    fixing \(n_{\mathrm{labeled}}=200\) and \(q=0.2\). FCND variants use the official
    \texttt{fcnd} implementation with the null-fraction correction.}
    \label{fig:fcnd_llm}
\end{figure}

Overall, these experiments confirm that the difference between FCND and OptCS-Full is not
only conceptual. FCND provides a valid e-value-based novelty-detection route in the binary
setting, but the direct null-reference adaptation does not natively exploit labeled positives in
the same supervised way as OptCS-Full. In supervised conformal-selection problems, this
data-use difference can lead to a substantial power advantage for OptCS while maintaining
finite-sample FDR control.

\subsection{Comparing training outcome encoding in OptCS-Full and OptCS-Full-MSel}
\label{app:subsec_encoding}

We evaluate the impact of outcome encoding functions $T(\cdot,\cdot)$ on the performance of OptCS-Full and OptCS-Full-MSel. 
We compare the binary encoding \(T_{\mathrm{bin}}(y,c)=\mathbf{1}\{y>c\}\) with the positive-part encoding \(T_{\mathrm{pos}}(y,c)=(y-c)_+\), while keeping the data generation, model class, and all other steps the same. For Jin's settings, we use the continuous responses $Y_i=\mu(X_i)+\epsilon_i$ without the binarization; Liang's settings already have continuous responses. In both cases $c_{n+j}\equiv 0$, so the two encodings differ only in the training labels.

First, we apply OptCS-Full with the two transformed outcomes in training, across the four Jin settings using gradient boosting, random forest, ridge regression, and support vector regression. The nominal level is fixed at $q=0.2$, and the sample sizes are $n=200$, $m=100$. The results averaged over $N=100$ independent runs are in Figure~\ref{fig:simu_encoding_full}.
We observe that both encodings generally maintain empirical FDR near the nominal level. The positive-part encoding achieves slight power gains across models in Settings 1 and 3 while the power comparison can be case-dependent in settings 2 and 4. Overall, the two encodings lead to similar power.

\begin{figure}[htbp]
    \centering
    \includegraphics[width=0.6\linewidth]{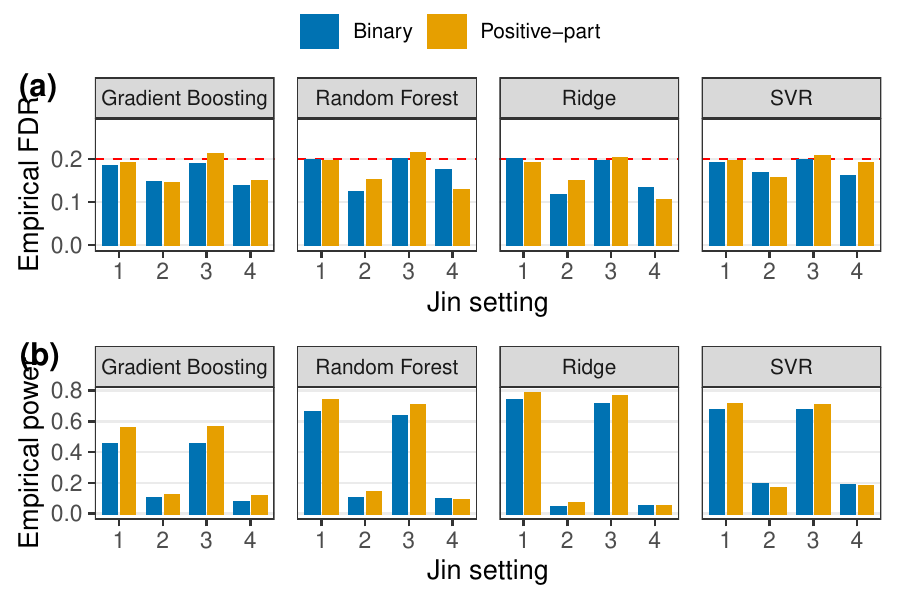}
    \caption{Empirical (a) FDR and (b) power of the binary and positive-part encodings for OptCS-Full across the four Jin settings. Each column  corresponds to a model class. Results are averaged over $N=100$ independent runs with $n=200$, $m=100$, and $q=0.2$. The red dashed line marks the nominal FDR level.}
    \label{fig:simu_encoding_full}
\end{figure}

For OptCS-Full-MSel, we compare a Binary library, a Positive library, and their union, termed the Mixed library, in Liang's settings following Section~\ref{subsec:simu_full_msel}. We fix $n=200$, $m=100$, and vary the nominal level $q\in \{0.2, 0.3, 0.4, 0.5\}$; the results are in Figure~\ref{fig:simu_encoding_msel}. Across the settings and model class libraries, OptCS-Full-MSel maintains the FDR below the target levels. The Mixed library yields slight but visible power gains in several of the more heterogeneous Liang's settings, while the comparison between Binary and Positive can change across context without uniform domination of either. Overall, the performance is similar across encoding functions.

\begin{figure}[htbp]
    \centering
    \includegraphics[width=0.75\linewidth]{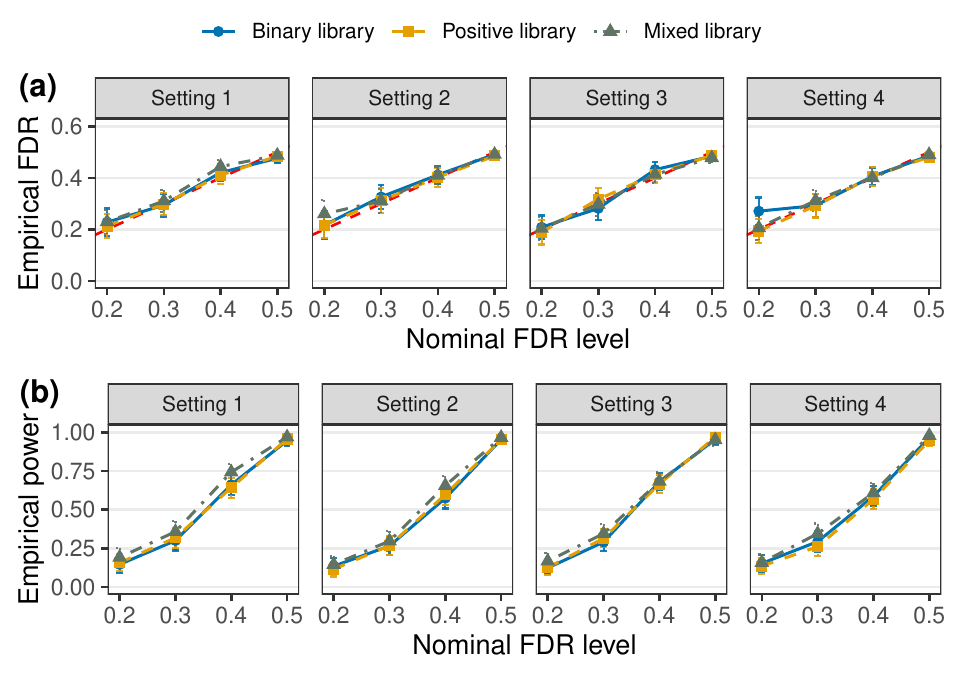}
    \caption{Empirical (a) FDR and (b) power of the binary, positive-part, and mixed encodings for OptCS-Full-MSel across the four Liang's settings. Results are averaged over $N=100$ independent runs with $n=200$, $m=100$, and varying $q\in\{0.2, 0.3, 0.4, 0.5\}$. The red dashed line marks the nominal FDR level.}
    \label{fig:simu_encoding_msel}
\end{figure}

\subsection{Runtime Comparison for OptCS-Non-LOO} \label{app:nonLOO}

In Appendix~\ref{app:subsec_discuss_loo}, we presented additional variants of OptCS-Full and OptCS-Full-MSel for settings where computational resources are not a primary concern and the goal is to maximize selective power. At the other end of the spectrum, one may face more stringent computational constraints while still seeking improved power over the na\"ive baselines. This can be achieved using the OptCS-Full-Non-LOO and OptCS-Full-MSel-Non-LOO variants, in which each model $\hat\mu_{\ell,k}$ is trained on the full dataset rather than the leave-one-out dataset $\mathbf{D}_{-\ell}^{(k)}$. As a result, the entire procedure requires fitting only $K$ models, instead of $(n_2+m)\times K$ models.

As discussed in the main text and illustrated in Appendix~\ref{app:comp_other_full}, this modification may entail a modest loss in power: the prediction $\hat\mu_{\ell,k}(X_\ell)$ is obtained from a model trained on $Z_\ell^{T_k}$ itself, which may make it less informative for constructing the corresponding score. Nevertheless, avoiding repeated leave-one-out fitting can substantially reduce computational cost. Table~\ref{tab:non_loo_runtime} reports the runtimes of the split baseline, OptCS-Full-MSel, and OptCS-Full-MSel-Non-LOO using the first case in Liang's settings, with various sample sizes. Runtimes are averaged over 100 repetitions and include selection at all 19 nominal levels. All computations in this section were performed on an Intel Core i7-12700KF CPU.

\begin{table}[htbp]
  \centering
  \caption{Mean runtime $\pm$ standard error (seconds) across 100 seeds.}
  \label{tab:non_loo_runtime}
  \small
  \setlength{\tabcolsep}{8pt}
  \renewcommand{\arraystretch}{1.15}
  \begin{tabular*}{\linewidth}{@{\extracolsep{\fill}}rrccc@{}}
    \toprule
    $n_{\mathrm{label}}$ & $n_{\mathrm{test}}$
      & \texttt{split\_121} & MSel-LOO & MSel-Non-LOO \\
    \midrule
    100 &  20 & $0.0255 \pm 0.0006$ & $ 2.4013 \pm 0.0298$ & $0.0299 \pm 0.0013$ \\
    100 &  50 & $0.0266 \pm 0.0001$ & $ 3.5795 \pm 0.0201$ & $0.0365 \pm 0.0009$ \\
    100 & 100 & $0.0313 \pm 0.0003$ & $ 6.2689 \pm 0.0358$ & $0.0529 \pm 0.0011$ \\
    100 & 200 & $0.0416 \pm 0.0001$ & $14.5313 \pm 0.0302$ & $0.0837 \pm 0.0005$ \\
    100 & 500 & $0.0722 \pm 0.0006$ & $69.3850 \pm 0.5157$ & $0.2218 \pm 0.0014$ \\
    \addlinespace[0.5em]
    200 &  20 & $0.0301 \pm 0.0004$ & $ 6.1515 \pm 0.0556$ & $0.0355 \pm 0.0005$ \\
    200 &  50 & $0.0324 \pm 0.0001$ & $ 8.2723 \pm 0.0561$ & $0.0458 \pm 0.0013$ \\
    200 & 100 & $0.0370 \pm 0.0001$ & $12.3128 \pm 0.0569$ & $0.0625 \pm 0.0014$ \\
    200 & 200 & $0.0475 \pm 0.0003$ & $23.1199 \pm 0.0818$ & $0.0925 \pm 0.0004$ \\
    200 & 500 & $0.0767 \pm 0.0004$ & $87.2296 \pm 0.3837$ & $0.2385 \pm 0.0031$ \\
    \addlinespace[0.5em]
    500 & 20 & $0.0600 \pm 0.0008$ & $41.7282 \pm 0.1987$ & $0.0881 \pm 0.0011$ \\
    500 & 50 & $0.0653 \pm 0.0010$ & $47.0760 \pm 0.1944$ & $0.1016 \pm 0.0012$ \\
    500 & 100 & $0.0691 \pm 0.0011$ & $60.4022 \pm 0.2945$ & $0.1231 \pm 0.0014$ \\
    500 & 200 & $0.0776 \pm 0.0012$ & $85.0055 \pm 0.6041$ & $0.1611 \pm 0.0015$ \\
    500 & 500 & $0.1144 \pm 0.0017$ & $206.3248 \pm 0.5524$ & $0.3456 \pm 0.0025$ \\
    \bottomrule
  \end{tabular*}
  \par\medskip
  \begin{minipage}{\linewidth}
    \footnotesize
    All configurations use data-generating setting 1. Each runtime covers all
    19 nominal levels ($q=0.05,0.10,\ldots,0.95$) and is counted once per seed.
  \end{minipage}
\end{table}

\end{document}